\documentclass[english,12pt,a4paper,dvipsnames]{book}
 
\usepackage[utf8]{inputenc}
\usepackage[T1]{fontenc}

\usepackage[english]{babel}
\usepackage[default,oldstyle, scale=.95]{opensans} 

\usepackage{amsmath}
\usepackage{amsfonts}
\usepackage{amssymb}
\usepackage{amsthm}
\usepackage{todonotes}
\usepackage{tcolorbox}
\usepackage{lettrine}
\usepackage{pgothic}

\usepackage[normalem]{ulem}

\usepackage{minitoc}
\dominitoc

\let\oldminitoc\minitoc
\renewcommand{\minitoc}{\oldminitoc\vspace*{0.5cm}}

\usepackage{caption}
\usepackage{xcolor} 
\definecolor{Prune}{RGB}{99,0,60} 
\definecolor{B1}{RGB}{49,62,72} 
\definecolor{C1}{RGB}{124,135,143}
\definecolor{D1}{RGB}{213,218,223}
\definecolor{A2}{RGB}{198,11,70}
\definecolor{B2}{RGB}{237,20,91}
\definecolor{C2}{RGB}{238,52,35}
\definecolor{D2}{RGB}{243,115,32}
\definecolor{A3}{RGB}{124,42,144}
\definecolor{B3}{RGB}{125,106,175}
\definecolor{C3}{RGB}{198,103,29}
\definecolor{D3}{RGB}{254,188,24}
\definecolor{A4}{RGB}{0,78,125}
\definecolor{B4}{RGB}{14,135,201}
\definecolor{C4}{RGB}{0,148,181}
\definecolor{D4}{RGB}{70,195,210}
\definecolor{A5}{RGB}{0,128,122}
\definecolor{B5}{RGB}{64,183,105}
\definecolor{C5}{RGB}{140,198,62}
\definecolor{D5}{RGB}{213,223,61}
\usepackage{mdframed}
\usepackage{multirow} 
\usepackage{multicol} 
\usepackage{graphicx}
\usepackage[absolute]{textpos} 
\usepackage{colortbl}
\usepackage{array}
\usepackage{geometry}
\usepackage{enumitem}
\usepackage{titlesec}
\usepackage{algorithm}
\usepackage[noend]{algpseudocode}
\usepackage{etoolbox}
\usepackage{tikz}
\usetikzlibrary{shapes, arrows.meta, positioning, calc, decorations.pathreplacing, backgrounds, shadows.blur}
\usepackage{ragged2e}
\usepackage[backref=page]{hyperref}
\hypersetup{
    colorlinks=true,    
    linkcolor=Prune,      
    citecolor=blue,     
    filecolor=red,  
    urlcolor=teal       
}
\usepackage[symbol=$\uparrow\;$]{footnotebackref}

\newenvironment{stratProposer}{\begin{tcolorbox}[colback=orange!10!white]}{\end{tcolorbox}}

\newenvironment{stratAttester}{\begin{tcolorbox}[colback=blue!10!white]}{\end{tcolorbox}}

\algrenewcommand{\algorithmiccomment}[1]{\hfill\hspace{1cm}\textcolor{gray}{$\triangleright$ #1}}

\DeclareMathOperator*{\argmax}{arg\,max}

\algdef{SE}{Upon}{EndUpon}[2]{\textbf{upon} \(\mbox{#1}\) \textbf{from}  \(\mbox{#2}\) \textbf{do}}{\textbf{end upon}}
\algtext*{EndUpon}

\AtBeginEnvironment{algorithmic}{\footnotesize}

\renewcommand*{\backref}[1]{}

\renewcommand*{\backrefalt}[4]{%
  \ifcase #1 %
    (Not cited.)%
  \or
    (Cited on page~#2)%
  \else
    (Cited on pages~#2)%
  \fi%
}

\makeatletter
\renewcommand*{\backrefxxx}[3]{%
  \ifx\relax#1\relax
  \else
    \hyperlink{cite.#2}{\textcolor{blue}{$\uparrow$~#1}}%
  \fi
}
\makeatother

\newcounter{mycite}
\let\oldcite\cite
\renewcommand{\cite}[1]{%
  \refstepcounter{mycite}%
  \hypertarget{cite.\themycite}{}%
  \oldcite{#1}%
}

\newtcolorbox{myquote}{colback=gray!10, colframe=white, 
    boxrule=0.5pt, 
    arc=4pt, 
    left=6pt, right=6pt, top=6pt, bottom=6pt}
\renewenvironment{quote}
  {\begin{myquote}}
  {\end{myquote}}

\newcommand{\chapterLettrine}[2]{\lettrine[lines=4,depth=-1,nindent=0em,lraise=0.2]{#1}{#2}} 

\newboolean{showcomments}
\setboolean{showcomments}{true}
\ifthenelse{\boolean{showcomments}}
{ \newcommand{\mynote}[3]{
     \fbox{\sffamily\scriptsize#1}
       {\small$\blacktriangleright$\textcolor{#3}{{ #2}}}}}
       {\newcommand{\mynote}[2]{}}

\def\barroman#1{\sbox0{#1}\dimen0=\dimexpr\wd0+1pt\relax
  \makebox[\dimen0]{\rlap{\vrule width\dimen0 height 0.06ex depth 0.06ex}%
    \rlap{\vrule width\dimen0 height\dimexpr\ht0+0.03ex\relax 
            depth\dimexpr-\ht0+0.09ex\relax}%
    \kern.5pt#1\kern.5pt}}

\newtheorem{theorem}{Theorem}[chapter]
\newtheorem{lemma}{Lemma}[chapter]
\newtheorem{definition}{Definition}[chapter]

\newtheorem{property}{Property}[chapter]
\newtheorem{observation}{Observation}[chapter]
\newtheorem{remark}{Remark}[chapter]
\newtheorem*{remark*}{Remark}

\addto\extrasenglish{
    
}
\newcommand{\isFrench}{false}

\begin{document}

\begin{titlepage}

\newgeometry{left=6cm,bottom=2cm, top=1cm, right=1cm}

\tikz[remember picture,overlay] \node[opacity=1,inner sep=0pt] at (-13mm,-135mm){\includegraphics{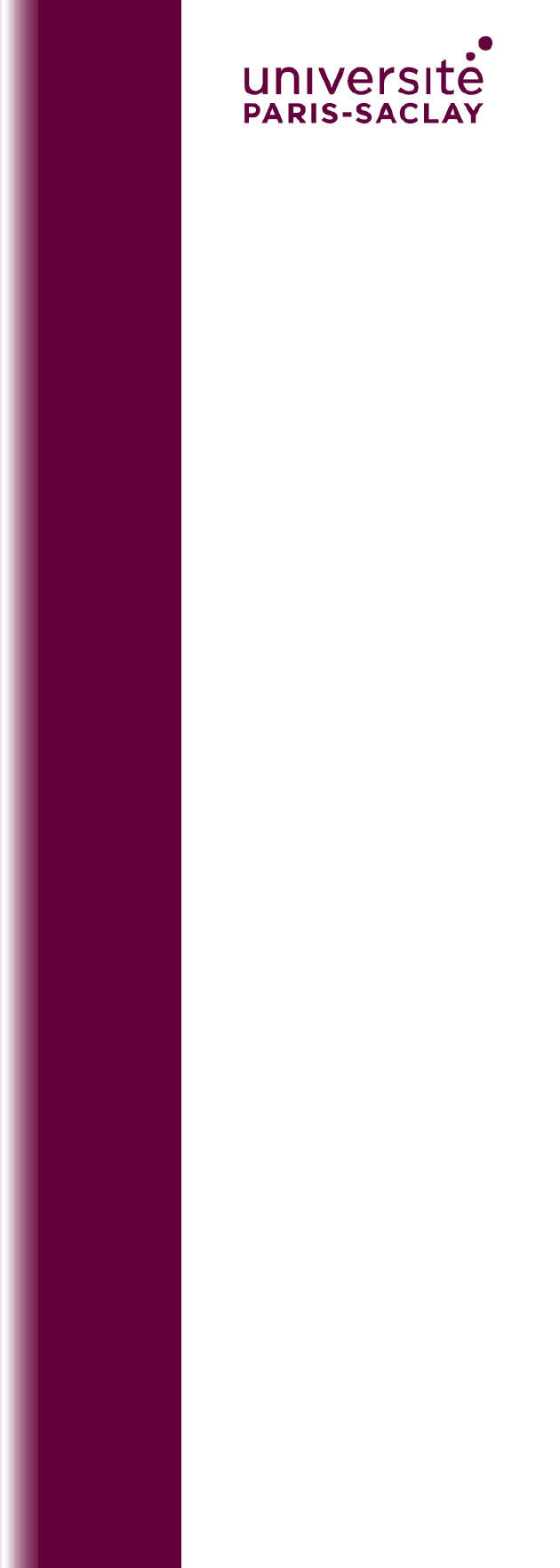}};

\newcommand{\bilingualThesis}{
    \ifthenelse{\boolean{\isFrench}}
        {THESE DE DOCTORAT}
        {PhD THESIS}
}

\color{white}

\begin{picture}(0,0)
\put(-152,-743){\rotatebox{90}{\Large \textsc{\bilingualThesis}}} \\
\put(-120,-743){\rotatebox{90}{NNT : 2024UPASG071}}
\end{picture}
 
\vspace{2cm}


\flushright
\vspace{10mm} 
\color{Prune}

\newcommand{\bilingualTitle}{
    \ifthenelse{\boolean{\isFrench}}
        {Approche Théorie des jeux pour l'\raisebox{0.3ex}{\'{}}\kern-0.5em Etude de la Robustesse des Blockchains}
        {A Game-Theoretic Approach to the Study of Blockchain's Robustness}
}

\fontsize{22}{26}\selectfont
  \Huge \bilingualTitle \\

\color{black}
\ifthenelse{\boolean{\isFrench}}
        {\normalsize
\Large{\textit{A Game-Theoretic Approach to the Study\\ of Blockchain's Robustness}} \\}
        {}


\fontsize{8}{12}\selectfont

\vspace{1.5cm}

\newcommand{\bilingualUniversity}{
    \ifthenelse{\boolean{\isFrench}}
        {Thèse de doctorat de l'université Paris-Saclay}
        {PhD Thesis from Paris-Saclay University}
}

\normalsize
\textbf{\bilingualUniversity} \\

\vspace{6mm}

\ifthenelse{\boolean{\isFrench}}
{
\small École doctorale n$^{\circ}$ 580 - Sciences et Technologies de l’Information et de la Communication (STIC)\\
\small Spécialité de doctorat: Informatique\\
\small Graduate School : Informatique. Référent : Faculté des sciences d'Orsay. \\
}{}
\vspace{6mm}

\ifthenelse{\boolean{\isFrench}}{
\footnotesize Thèse préparée dans l'unité de recherche \textbf{LIST}  (CEA, Université Paris Saclay),\\ sous la direction de \textbf{Sara TUCCI-PIERGIOVANNI}, cheffe de laboratoire, \\ et le co-encadrement de \textbf{Yackolley AMOUSSOU-GUENOU}, maître de conférences \\
}{
\footnotesize PhD thesis prepared at the research unit \textbf{LIST} (CEA, Paris-Saclay University),\\ under the supervision of \textbf{Sara TUCCI-PIERGIOVANNI}, laboratory head, \\ and co-supervised by \textbf{Yackolley AMOUSSOU-GUENOU}, associate professor. \\
}
\vspace{15mm}

\ifthenelse{\boolean{\isFrench}}{
\textbf{Thèse soutenue à Paris-Saclay, le 5 novembre 2024, par}\\
}{
\textbf{Thesis defended at Paris-Saclay, on November 5, 2024, by}\\
}
\bigskip
\Large {\color{Prune} \textbf{Ulysse PAVLOFF}} 

\vspace{\fill} 

\bigskip

\ifthenelse{\boolean{\isFrench}}{
\flushleft
\small {\color{Prune} \textbf{Composition du jury}}\\
{\color{Prune} \scriptsize {Membres du jury avec voix délibérative}} \\
\vspace{2mm}
\scriptsize
\begin{tabular}{|p{7cm}l}
\arrayrulecolor{Prune}
\textbf{Sylvain CONCHON} &  Président du jury  \\ 
Professeur des universités, Université Paris-Saclay   &   \\ 
\textbf{Sonia BEN MOKTHAR} &  Rapporteur \& Examinatrice \\ 
Directrice de recherche, LIRIS CNRS  &   \\ 
\textbf{Sébastien TIXEUIL} &  Rapporteur \& Examinateur \\ 
Professeur des universités, Sorbonne Université   &   \\ 
\textbf{Daniel AUGOT} &  Examinateur \\ 
Directeur de recherche, INRIA   &   \\ 
\textbf{Sophie CHABRIDON} &  Examinatrice \\ 
Professeure, Telecom SudParis   &   \\  

\end{tabular} 
}{
\flushleft
\small {\color{Prune} \textbf{Jury Composition}}\\
{\color{Prune} \scriptsize {Jury members with deliberative voting rights}} \\
\vspace{2mm}
\scriptsize
\begin{tabular}{|p{7cm}l}
\arrayrulecolor{Prune}
\textbf{Sylvain CONCHON} &  Jury President  \\ 
University Professor, Paris-Saclay University   &   \\ 
\textbf{Sonia BEN MOKTHAR} &  Reviewer \& Examiner \\ 
Director of Research, LIRIS CNRS  &   \\ 
\textbf{Sébastien TIXEUIL} &  Reviewer \& Examiner \\ 
University Professor, Sorbonne University   &   \\ 
\textbf{Daniel AUGOT} &  Examiner \\ 
Director of Research, INRIA   &   \\ 
\textbf{Sophie CHABRIDON} &  Examiner \\ 
Professor, Telecom SudParis   &   \\  

\end{tabular}

}

\end{titlepage}

\thispagestyle{empty}
\newgeometry{top=1.5cm, bottom=1.25cm, left=2cm, right=2cm}

\noindent 
\ifthenelse{\boolean{\isFrench}}{
\includegraphics[height=2.45cm]{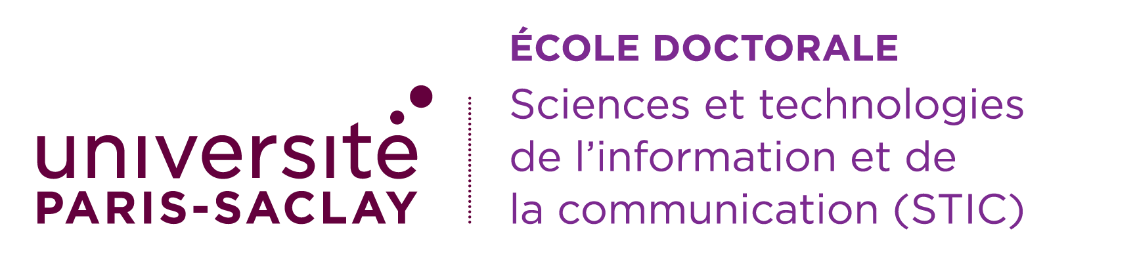}

\small

\vspace*{\fill}
\begin{mdframed}[linecolor=Prune,linewidth=1]

\textbf{Titre:} Approche Théorie des jeux pour l'Étude de la Robustesse des Blockchains

\noindent \textbf{Mots clés:} Blockchain, Ethereum, Systèmes Distribués, Théorie des Jeux

\begin{multicols}{2}
\hspace*{\fill} \textbf{Résumé} \hspace*{\fill} \\
Les blockchains ont suscité un intérêt mondial ces dernières années, prenant de plus en plus d'importance à mesure qu'elles influencent les technologies et la finance. Cette thèse explore la robustesse des protocoles blockchain, en se concentrant spécifiquement sur Ethereum Proof-of-Stake (PoS), un pilier des technologies blockchain. Nous définissons la robustesse en termes de deux propriétés essentielles : la sécurité, qui garantit qu'il n'y aura pas de blocs conflictuels permanents dans la blockchain, et la vivacité, qui assure l'ajout continu de nouveaux blocs permanents.

Notre recherche aborde l'écart entre les approches traditionnelles des systèmes distribués, qui classifient les agents comme étant soit honnêtes, soit byzantins (i.e., malveillants ou défaillants), et les modèles de théorie des jeux qui considèrent les agents rationnels motivés par des incitations. Nous explorons comment les incitations impactent la robustesse en utilisant les deux approches.

La thèse est composé de trois analyses distinctes. Nous commençons par formaliser le protocole Ethereum PoS, définissant ses propriétés et examinant les vulnérabilités potentielles du point de vue des systèmes distribués. Nous identifions certaines attaques qui peuvent compromettre la robustesse du système. Ensuite, nous analysons le mécanisme de fuite d'inactivité, une caractéristique clé d'Ethereum PoS, en soulignant son rôle dans le maintien de la vivacité du système lors de perturbations du réseau, mais au détriment de la sécurité. Enfin, nous utilisons des modèles de théorie des jeux pour étudier les stratégies des validateurs rationnels au sein d'Ethereum PoS, en identifiant les conditions dans lesquelles ces agents pourraient s'écarter du protocole prescrit pour maximiser leurs récompenses.

Nos résultats contribuent à une meilleure compréhension de l'importance des mécanismes d'incitation pour la robustesse des blockchains et donnent des pistes pour concevoir des protocoles blockchain plus résilients.

\end{multicols}

\end{mdframed}

\vspace*{\fill}

\newpage }{}

\includegraphics[height=2.45cm]{imgOfficial/logo_usp_STIC.png}

\vspace*{\fill}

\pagenumbering{gobble}

\begin{mdframed}[linecolor=Prune,linewidth=1]

\textbf{Title:} A game-theoretic approach to the study of Blockchain's Robustness

\noindent \textbf{Keywords:} Blockchain, Ethereum, Distributed Systems, Game Theory

\begin{multicols}{2}
\hspace*{\fill} \textbf{Abstract} \hspace*{\fill} \\
Blockchains have sparked global interest in recent years, gaining importance as they increasingly influence technology and finance.
This thesis investigates the robustness of blockchain protocols, specifically focusing on Ethereum Proof-of-Stake. We define robustness in terms of two critical properties: Safety, which ensures that the blockchain will not have permanent conflicting blocks, and Liveness, which guarantees the continuous addition of new reliable blocks.

Our research addresses the gap between traditional distributed systems approaches, which classify agents as either honest or Byzantine (i.e., malicious or faulty), and game-theoretic models that consider rational agents driven by incentives. We explore how incentives impact the robustness with both approaches.

The thesis comprises three distinct analyses. First, we formalize the Ethereum PoS protocol, defining its properties and examining potential vulnerabilities through a distributed systems perspective. We identify that certain attacks can undermine the system's robustness. Second, we analyze the inactivity leak mechanism, a critical feature of Ethereum PoS, highlighting its role in maintaining system liveness during network disruptions but at the cost of safety. Finally, we employ game-theoretic models to study the strategies of rational validators within Ethereum PoS, identifying conditions under which these agents might deviate from the prescribed protocol to maximize their rewards.

Our findings contribute to a deeper understanding of the importance of incentive mechanisms for blockchain robustness and provide insights into designing more resilient blockchain protocols.
\end{multicols}
\end{mdframed}
\vspace*{\fill}

\renewcommand{\thechapter}{\Roman{chapter}}
    
\titleformat{\chapter}[display]{\bfseries\Large\color{Prune}}
{   
    \vspace{-1cm}
    \centering 
    \begin{tikzpicture}
    \node (A) at (0,0) {};
    \node (B) at (2.5,0) {Chapter \thechapter };
    \node (C) at (5,0) {};
    \draw (A) -- (B);
    \draw (B) -- (C);
    \end{tikzpicture}
}
{.1ex}
{
    \centering
    \vspace{.7em}
    \LARGE
}
[\vspace{1cm}]
\titlespacing{\chapter}{0pc}{0ex}{0.5pc}

\titleformat{\section}[hang]{\bfseries\normalsize\color{Prune}}{\thesection\ .}{0.5pt}
{\vspace{0.1ex}
}
[\vspace{0.1ex}]
\titlespacing{\section}{1.5pc}{4ex plus .1ex minus .2ex}{.8pc}

\titleformat{\subsection}[hang]{\bfseries\small}{\thesubsection\ .}{1pt}
{\vspace{0.1ex}
}
[\vspace{0.1ex}]
\titlespacing{\subsection}{3pc}{2ex plus .1ex minus .2ex}{.1pc}

\frontmatter  
    \newgeometry{top=2cm, bottom=2.5cm, left=3cm, right=3cm}
    \vspace*{\fill}
\vspace{-4cm}

\begin{quote}
\hphantom{aa} ``Our time is such that those who feel certainty are stupid, while those with any imagination and understanding are filled with doubt and indecision." \\
\flushright - Bertrand Russell
\end{quote}






\vspace*{\fill}
    \chapter{Acknowledgement}

\pagenumbering{roman}

\ifthenelse{\boolean{\isFrench}}
{
I am deeply grateful to Sara Tucci-Piergiovanni and Yackolley Amoussou-Guenou for their invaluable support and guidance. Passer ces trois années à vos côtés a été incroyablement enrichissant. Vous avez su orienter mon intérêt et notre recherche pour toujours travailler sur des sujets stimulants. Merci pour tous ces échanges captivants et pour votre confiance durant toute cette aventure.

\

Merci aux rapporteurs Sonia Ben Mokhtar et Sébastien Tixeuil pour leur lecture attentive et aux retours détaillés sur mon travail.
Merci aux membres du jury, Sophie Chabridon, Daniel Augot et Sylvain Conchon, d'avoir accepté de faire partie de mon jury de thèse.

\

Je tiens également à remercier Jérôme Lang et Tristan Cazenave de m'avoir offert un aperçu de ce qu'était une thèse, et de m'avoir conforté dans l'idée d'en entreprendre une, grâce à l'excellente expérience que j'ai eue à vos côtés.

\

Merci à Ewa Pichon, qui m'a accompagné tout au long de ma thèse et a rendu possible d'aller échanger des connaissances scientifiques sans encombre. 

\

Un spécial merci à mes collègues de bureau pendant la plus grande partie de ma thèse, Alexandre Rapetti, Hector Roussille et Valentin Pasquale. Merci de m'avoir aidé dans les moments compliqués, mais surtout de vous être assurés que je ne travaille jamais trop dur trop longtemps.

\

Comment ne pas remercier mon colocataire, Léo Satgé, toujours présent pour m'épauler au quotidien, tout en me rappelant que faire une thèse ne me donne aucun avantage quand il s'agit de s'affronter sur des jeux de société ou vidéo.

\

Enfin, je tiens à remercier ma famille, qui m'a toujours soutenu durant ces trois années, et tout particulièrement mon père, Nicolas Pavloff. Les explications et questions qu'on s'est échangées m’ont été d’une aide précieuse. Merci à toi de m’avoir inspiré à entreprendre cette thèse et de m’avoir encouragé à chaque étape. 
}
{
I am deeply grateful to Sara Tucci-Piergiovanni and Yackolley Amoussou-Guenou for their invaluable support and guidance. Spending these three years alongside you has been incredibly enriching. You have skillfully guided my interests and our research to continuously focus on stimulating topics. Thank you for all these captivating discussions and for your trust throughout this journey.

\

Thank you to the reviewers, Sonia Ben Mokhtar and Sébastien Tixeuil, for their careful reading and detailed feedback on my work. Thank you to the jury members, Sophie Chabridon, Daniel Augot, and Sylvain Conchon, for agreeing to be part of my thesis defense committee.

\

I would also like to thank Jérôme Lang and Tristan Cazenave for giving me a glimpse of what a PhD journey entails and for reinforcing my desire to undertake one through the excellent experience I had working with you.

\

Thank you to Ewa Pichon, who supported me throughout my thesis and made it possible to share scientific knowledge seamlessly.

\

A special thanks to my office colleagues during most of my thesis, Alexandre Rapetti, Hector Roussille, and Valentin Pasquale. Thank you for helping me through challenging times and, above all, for making sure I never worked too hard for too long.

\

I must also thank my roommate, Léo Satgé, always there to support me daily while reminding me that being a PhD student doesn’t grant me any advantage when it comes to competing in board games or video games.

\

Finally, I would like to thank my family, who has always supported me throughout these three years, especially my father, Nicolas Pavloff. The explanations and questions we exchanged were of immense help. Thank you for inspiring me to undertake this PhD and for encouraging me at every step.
}

    \chapter{Reading tips}



Throughout the manuscript, every text in purple can be clicked to navigate to the indicated reference. Footnotes function in this manner as well\footnote{Clicking on the footnote's number will return you to the place of the footnote in the text.}. Moreover, citations are in blue and offer the same ability to take you back to the place of the citation in the text by clicking on the page number at the end of the citation 
\cite{attiya_synchronization_2023}.

    \newgeometry{top=2cm, bottom=2.5cm, left=2cm, right=2cm}
    \tableofcontents
    \newgeometry{top=2cm, bottom=2.5cm, left=3cm, right=3cm}
    
\mainmatter
    \chapter{Introduction} \label{chap:Introduction}

\minitoc

\chapterLettrine{T}{he} first block of the Bitcoin blockchain, the genesis block, contains a message taken from the title of the January 3, 2009 edition of The Times newspaper:\textit{"Chancellor on brink of second bailout for banks"}. 
This statement serves as a powerful symbol of what Bitcoin set out to challenge: a flawed financial system that has proven to be unworthy of trust.
Since then, blockchain technology has gained significant recognition and ignited widespread interest. 
To build a better understanding and start with a clear foundation of blockchains, we begin by outlining the core principles and essential knowledge for an informed discussion about blockchains.

Before explaining blockchains, let us be clear that blockchain and cryptocurrencies are different. Cryptocurrency is one application allowed by blockchains, but blockchains can be much more and this distinction is important.

\section{Blockchain}

Blockchain technology has revolutionized how we think about data storage, ownership, and decentralized systems. To better understand this technology, we break it down into three key components: what it achieves, how it functions, and its potential applications.

\subsection{Functionalities}

We begin with the functionalities brought by blockchains. What new capabilities does it offer that did not exist before? A simple way to describe blockchains is to say that \textbf{a blockchain is a decentralized computer}. Like a traditional computer, you can store data and run programs on it. The difference lies in its decentralized nature: it has no single owner or operator.

Imagine a global computer that anyone can access and use to store data and run programs without restriction, this is the promise of blockchain technology.
One interesting feature of blockchain is that, paradoxically, it is possible to create a strong sense of ownership over digital data on this public computer. Nakamoto first demonstrated that it was possible to enforce property rights over digital data with Bitcoin. His solution focused on digital currency, but this concept can be applied to any type of digital data stored on this decentralized "computer in the sky."

\subsection{Implementation}

The implementation refers to \textit{how} the functionalities are accomplished. What are the technological building blocks necessary to create this solution? This part is the most fascinating one, as it is where most of the research is conducted, and this manuscript is no exception.

Similar to how advancements in hardware have improved global communication and enabled new internet applications, blockchain technology is evolving.



The essential components used to build blockchains are:
\begin{itemize}
\item \textbf{Consensus Mechanism}: These are algorithms that enable the distributed network to agree on the state of the blockchain. Bitcoin uses \textit{Proof-of-Work}, also called \textit{Nakamoto Consensus}. Other mechanisms exist, such as \textit{Proof-of-Stake} and \textit{Byzantine Fault Tolerance}.
\item \textbf{Cryptography:} This branch of mathematics is essential for permitting ownership in a public environment through digital signatures. Cryptography is the backbone of blockchain's security, ensuring that transactions and data remain private, secure, and verifiable without a central authority. Numerous cryptographic tools are used in blockchain. One of these tools is the \textit{hashing}, which provides a unique fingerprint for any data. Cryptography is the reason why \textit{cryptocurrencies}—currencies that use blockchains—are named this way.
\item \textbf{Distributed Ledger}: 
This is the distributed database that maintains a continuously growing list of records, called blocks, literally the chain of blocks. In practice, this data is stored by every node, which are the computers working to maintain the blockchain. The consensus mechanism ensures that nodes agree on the data stored.
\end{itemize}

Research on blockchains aims to improve and create new solutions for any of these components. Our work focuses on the consensus mechanism.

\subsection{Purpose}

Now, the infamous question: “Okay, but what for?” Working in blockchain research, this might be the question I hear the most, even more than “What are blockchains?”

There are already compelling answers, ranging from traceability and decentralized money to digital ownership of goods and art. However, this is only the beginning.
As of the writing of this manuscript, blockchain has existed since 2009, about 15 years. What would the answer to this question have been for the internet 15 years after its development? At that time, the internet was mostly used for emails, you could send files, albeit very slowly, and web pages were still in their infancy. Now, the internet’s applications are almost limitless.

We are 15 years past the first use case of blockchain, it has only just begun. Like the internet, the reasons for using blockchain will become more numerous as time goes on. Just as the internet's potential became clearer over time, so too will blockchain applications expand and evolve.

\

The most well-known application of blockchain, cryptocurrency, which can be seen as decentralized money, is often deemed useless in the Western world. We in Western countries might not see the immediate need. However, millions of people in countries where the value of money is decreasing daily, where banks refuse to return what should be their money, have found value in blockchain. Blockchain provides a solution to actually own your money and use it without intermediaries. 

While these issues may not resonate with everyone now, as more people begin to understand and experience the benefits of true digital ownership and decentralized finance, the value of blockchain will become increasingly evident.

\

Understanding blockchain's functionalities, implementation, and potential applications is crucial as we continue to explore and refine this transformative technology. In the following section, we delve into the history of the first blockchain: Bitcoin.

\section{Bitcoin}
Bitcoin was described in a 2008 paper \cite{nakamoto_peer_2008} and implemented in 2009 by a mysterious author known as Satoshi Nakamoto. The technological components necessary to create Bitcoin had actually existed for about 15 years before its appearance. However, no solution had been found to make a system both decentralized and resistant to Sybil attacks and double spending which we detailed below. In this section, we will provide a historical overview of how Bitcoin came into existence and how it works.

Bitcoin answers the question: "How do you build a trustless system with an unknown number of participants?". The difficulty lies in the consensus mechanism required to make participants agree without knowing the number of participants. The problem of reaching consensus with a known number of participants has been addressed by the distributed systems field, notably through the Byzantine Generals Problem \cite{lamport_byzantine_1982}.

Several challenges arise when trying to create decentralized digital money. The first challenge is preventing double spending.

\paragraph{Double spending.} A major problem for digital money is that digital information can be copied effortlessly. Unlike a physical coin, which can only exist in one place at a time, digital assets can be copied, potentially allowing the same digital coin to be spent multiple times.

If you rely on a centralized authority such as a bank, this problem does not exist. The bank being the sole entity keeping track of your balance and authorizing transactions, it can prevent you from spending more than you own. This problem led the first digital money, \textit{eCash}, to rely on banks for this very reason. David Chaum founded Digicash in 1989, building on the work in \cite{chaum_untraceable_1988}, which enhanced his earlier work \cite{chaum_blind_1982}.

To mitigate this issue in a decentralized system, one solution is to have the participants vote on the state of balances after each transaction. However, this leads to the second problem:

\paragraph{Sybil attack.} A common issue when trying to reach consensus with an unknown set of participants is that some individuals could create multiple identities to increase their voting power, thereby gaining an unfair advantage in decision-making. This type of attack is called a Sybil attack, named after the book \textit{Sybil} by Schreiber, which tells the story of a woman with dissociative identity disorder who experienced having 16 different personalities.

Voting is crucial in a decentralized system, and several solutions prior to Bitcoin struggled with this problem. One such solution was Wei Dai's \textit{b-money}, proposed in 1998 \cite{wei_b-money_1998}. The ideas behind b-money are similar in many ways to Bitcoin, but to achieve consensus, Dai proposed using a fixed set of `trusted' servers to keep track of balances. However, choosing and trusting these servers contradicts the principle of decentralization. We have an email from Satoshi Nakamoto to Wei Dai, sent in 2008, that reads:
\begin{quote}
``I was very interested to read your b-money page. I'm getting ready to release a paper that expands on your ideas into a complete working system. Adam Back (hashcash.org) noticed the similarities and pointed me to your site.''
\end{quote}

Another problem mentioned by Dai in the b-money text is the fairness in the distribution of newly created money. While Digicash and b-money laid the groundwork for digital currencies, it was Bitcoin that finally solved the key issues, enabling the creation of a fully decentralized and secure currency system. The solution found by Nakamoto actually uses the work of Adam Back: the \textit{Proof-of-Work}.

\subsection{Bitcoin's Solution}

Bitcoin's solution to the problems of double spending and Sybil attacks is \textit{Proof-of-Work} (PoW). In 1992, a proposal to use a form of proof of work was presented by Dwork and Naor \cite{dwork_pricing_1992} as an anti-spam mechanism. They suggested requiring proof of computational work to send an email in order to prevent spam. Adam Back proposed a similar idea in 1997 \cite{back_hashcash_1997} and further developed it in 2002 \cite{back_hashcash_2002}, this time using cryptographic hashes. Let us explain PoW and hashes in more detail.

\paragraph{Hash.} A hash function is a cryptographic tool that creates a unique fingerprint for data. It takes any amount of data, scrambles it, and returns a short and unique result for that data.\footnote{To experience hash functions interactively, you can try this instructive website: \href{https://learnmeabitcoin.com/technical/cryptography/hash-function/}{online-SHA-256}.} Hash functions possess several desirable properties:
\begin{itemize}
\item \textbf{Deterministic:} For a given input, the function will always produce the same hash output.
\item \textbf{Irreversible:} Also known as pre-image resistance, this property ensures that given a hash value, it should be computationally infeasible to reverse-engineer the original input. The hash function is a one-way function.
\item \textbf{Avalanche effect:} Slightly different inputs should produce wildly different hash outputs.
\item \textbf{Collision resistance:} This property ensures that two different inputs do not produce the same hash output.
\end{itemize}

\

Now that we understand what a hash function is, we can explain the Bitcoin blockchain and its PoW mechanism.

Bitcoin is a distributed computer focused on saving data about digital money and transactions. To achieve this, blocks containing transactions are regularly added and saved by every participant. The order of these blocks is crucial since it serves as a historical record of transactions. It is important to verify that someone has received 10 bitcoins before they can spend them. Each block has a structure similar to the simplified representation shown in \autoref{fig:bitcoinBlock}. 

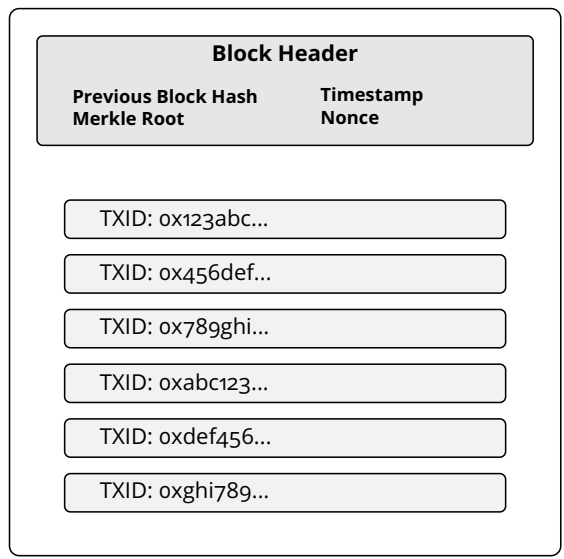
\begin{figure}
    \centering
    \resizebox{.5\linewidth}{!}{
        \begin{tikzpicture}[scale=1.0, every node/.style={scale=1.0}]

\draw[rounded corners=5pt, thick] (0,0) rectangle (10,10);

\draw[rounded corners=3pt, fill=gray!20, thick] (0.5,9.5) rectangle (9.5,7.5);
\node at (5,9.2) {\textbf{Block Header}};
\node[anchor=west] at (1,8.4) {\footnotesize \textbf{Previous Block Hash}};
\node[anchor=west] at (1,8) {\footnotesize \textbf{Merkle Root}};
\node[anchor=west] at (5.5, 8.4) {\footnotesize \textbf{Timestamp}};
\node[anchor=west] at (5.5, 8) {\footnotesize \textbf{Nonce}};

\newcommand{\transactionstart}{6.5} 
\newcommand{\transactionspacing}{1.0} 

\draw[rounded corners=3pt, fill=gray!10, thick] (1,\transactionstart) rectangle (9,\transactionstart-0.7);
\node[anchor=west] at (1.5,\transactionstart-0.35) { TXID: 0x123abc...};

\draw[rounded corners=3pt, fill=gray!10, thick] (1,\transactionstart-\transactionspacing) rectangle (9,\transactionstart-\transactionspacing-0.7);
\node[anchor=west] at (1.5,\transactionstart-\transactionspacing-0.35) { TXID: 0x456def...};

\draw[rounded corners=3pt, fill=gray!10, thick] (1,\transactionstart-2*\transactionspacing) rectangle (9,\transactionstart-2*\transactionspacing-0.7);
\node[anchor=west] at (1.5,\transactionstart-2*\transactionspacing-0.35) { TXID: 0x789ghi...};

\draw[rounded corners=3pt, fill=gray!10, thick] (1,\transactionstart-3*\transactionspacing) rectangle (9,\transactionstart-3*\transactionspacing-0.7);
\node[anchor=west] at (1.5,\transactionstart-3*\transactionspacing-0.35) { TXID: 0xabc123...};

\draw[rounded corners=3pt, fill=gray!10, thick] (1,\transactionstart-4*\transactionspacing) rectangle (9,\transactionstart-4*\transactionspacing-0.7);
\node[anchor=west] at (1.5,\transactionstart-4*\transactionspacing-0.35) { TXID: 0xdef456...};

\draw[rounded corners=3pt, fill=gray!10, thick] (1,\transactionstart-5*\transactionspacing) rectangle (9,\transactionstart-5*\transactionspacing-0.7);
\node[anchor=west] at (1.5,\transactionstart-5*\transactionspacing-0.35) { TXID: 0xghi789...};

\end{tikzpicture}
    }
    \caption{Simplified representation of a Bitcoin block.}
    \label{fig:bitcoinBlock}
\end{figure}

\paragraph{Bitcoin Block.} Each block contains a \textit{block header} that summarizes its information. The block header includes:
\begin{itemize}
\item \textbf{Previous Block Hash:} This links the blocks in order by referring to the hash of the preceding block, its unique fingerprint.
\item \textbf{Timestamp:} The time at which the block was created.
\item \textbf{Merkle Root:} The hash of the combined transactions (tx) in the block. This provides a single fingerprint representing all the transactions included in the block.
\item \textbf{Nonce:} A `number used once' that is useful for the Proof-of-Work mechanism, explained below.
\end{itemize}   

With this information, we can understand how Proof-of-Work operates.

\paragraph{Proof-of-Work (PoW).} To add a new block to the blockchain, the hash of its block header must begin with $x$ zeros. For instance, if $x=3$, the hash must start with three zeros, e.g., 00019... where only the first three digits matter. Since hash functions produce unpredictable results, the block proposer must generate many different blocks to find one whose hash starts with the required number of zeros. The probability of achieving a hash starting with three zeros is low, approximately 0.001. Therefore, successfully producing such a block serves as proof that a significant amount of computational work has been done.

The participants searching for valid blocks are called \textit{miners}. They generate numerous blocks by slightly modifying the nonce field in the block header, hoping to find a hash that meets the requirements. The value of $x$, which determines the difficulty of mining a block, varies such that a valid block is found approximately every 10 minutes. As more people mine for blocks, the value of $x$, the difficulty, increases. It is as if, while mining for gold in your own cave, you would find gold more easily when fewer people are mining for gold elsewhere and vice-versa.

\

Miners search for blocks and build on top of previous ones, but what happens if two miners find valid blocks simultaneously and attempt to add them to the blockchain at the same time?

\paragraph{Fork choice rule.} When several blocks share the same parent—meaning they reference the same previous block hash—we encounter a \textit{fork}. When multiple branches exist, the rule prescribed by Nakamoto is to build on the longest branch\footnote{In reality, the rule is to build on the branch of blocks that has consumed the most energy to build. In practice, this almost always corresponds to the longest chain.}, the one with the most blocks. If the branches are of the same length, you should add blocks to the one you saw first until, eventually, one becomes the longest.

\begin{quote}
    ``Each node must be prepared to maintain potentially several 'candidate' block chains, each of which may eventually turn out to become the longest one, the one which wins. Once a given block chain becomes sufficiently longer than a competitor, the shorter one can be deleted.'' \strut\hfill - Hal Finney \cite{finney_fork_2008}
\end{quote}


\ 
PoW solves all the problems of a decentralized computer. As stated in the Bitcoin white paper—the initial document explaining its functionality:

\begin{quote}
    ``The longest chain not only serves as proof of the sequence of events witnessed, but proof that it came from the largest pool of CPU power. As long as a majority of CPU power is controlled by nodes that are not cooperating to attack the network, they'll generate the longest chain and outpace attackers.''
\end{quote}

Thus, double-spending for a transaction becomes less likely as more blocks are built on top of the block containing the transaction. The current rule of thumb is to wait for six blocks to be built after a block to trust the transactions in it. This way, it is very unlikely that a new branch will appear and become the rightful chain while possibly excluding the transaction. The Sybil attack problem is solved by making one vote equal to one CPU. A miner cannot mine (or vote) on different chains at the same time. Every unit of energy used to search for a block on one chain is energy that cannot be used to search for a block on another chain. By making one CPU equal to one vote, as long as the majority of CPUs are owned by honest participants, malicious actors cannot control the network\footnote{The case in which a majority of CPUs are owned by malicious actors is referred to as the 51\% attack. Similar to democratic systems in general, if the majority agrees, the decision is theirs.}.

\

This clever mix of PoW and the longest chain rule is known as the \textit{Nakamoto Consensus}. The chain with the most accumulated proof of work (i.e., the longest chain) is considered the valid one. As blocks are built upon and become older, they become increasingly less likely to be reverted. This concept is known as \textit{probabilistic finality}, where a finalized block is one that will permanently remain in the chain.

\section{Ethereum}
After Bitcoin showed the way, more and more blockchains began to emerge. The one we are interested in, and which is the focus of this manuscript, is Ethereum.

In 2011, a young boy named Vitalik Buterin stumbled upon Bitcoin and found the idea fascinating. He fell down the rabbit hole and started to learn more about it. Around that time, he co-founded Bitcoin Magazine. After realizing that blockchains could offer more than just digital money, he tried to convince the community to improve Bitcoin. However, Bitcoin is very conservative and slow when it comes to changes. He then embarked on a mission to convince people to help him create a new blockchain—one that wasn't centered around a single application, such as cryptocurrency, but rather for any purpose, leaving the freedom to its users to decide how to use it. After gathering like-minded individuals interested in his project, they created Ethereum \cite{buterin_next_2014} in 2014.

Unlike Bitcoin, which is designed primarily for one application—digital money—Ethereum is a general-purpose blockchain. Ethereum is an open platform that allows people to build their own applications on top of it. Anything built on Ethereum is protected and secured, with every transaction checked by the entire network of millions of computers around the world that protect and verify every transaction on the blockchain.

\

Ethereum started as a PoW blockchain, with a fork choice rule similar to that of Bitcoin. From the start, the plan was to eventually transition to a different type of consensus called \textit{Proof-of-Stake} (PoS).

\paragraph{Proof-of-Stake (PoS).} The concept of PoS actually originated from a 2011 post on the BitcoinTalk forum \cite{quantum_proof_2011}. The post included the following thought:

\begin{quote}
    ``I am wondering if as bitcoins become more widely distributed, whether a transition from a proof of work based system to a proof of stake one might happen. What I mean by proof of stake is that instead of your "vote" on the accepted transaction history being weighted by the share of computing resources you bring to the network, it's weighted by the number of bitcoins you can prove you own, using your private keys.''
\end{quote}

As explained, the idea behind PoS is that voting power is not determined by computing resources (as in PoW), but by the number of digital coins one owns. In PoW, one CPU equals one vote; in PoS, one digital coin equals one vote.
This naturally raises many questions: How are the coins initially distributed? Does this system simply make the rich richer? How are blocks proposed? Different blockchains have answered these questions in various ways, and Ethereum has developed its own unique solution.

One thing to note is that, while Bitcoin is slow and conservative in its goals and development, Ethereum's ethos is to ‘move fast and break things’. Although it took eight years to transition from PoW to PoS, many of these changes can be questioned. It seems that the Ethereum community prefers to implement changes quickly and refine them as needed, rather than engage in prolonged debates over potential risks. This dynamic approach renders Ethereum an interesting blockchain to study as it fosters the emergence of new ideas. Plus, the challenge of ideas is welcomed by the community, which can lead to direct impacts on the blockchain.

We dive into a thorough explanation of the protocol in \autoref{chap:Eth-CModel-Analysis} as this is part of our contribution.

\section{Terminology}
Throughout this manuscript, we will utilize specific blockchain terminology. To ensure clarity and consistency, this section will serve as a glossary to explain recurring terms. Additionally, we will address instances where different terms have equivalent meanings.

\paragraph*{Protocol.} The protocol of a blockchain refers to the set of rules that users must follow to communicate and act within that blockchain. We may also refer to the \textit{specifications} of a protocol, which holds the same meaning. In practice, the protocol is defined by the code that blockchain participants use to send messages, propose blocks and transactions, and perform other necessary actions.

\paragraph*{Participants.} "Participants" is a term with many synonyms. Participants can also be called agents, nodes or processes interchangeably. There are different types of participants: those who take part in the blockchain consensus, and those who do not. A participant who only sends transactions is considered a user of the blockchain; they use the blockchain but do not have a role in its consensus. In Ethereum, participants involved in the consensus process are known as \textbf{validators}, while in Bitcoin, they are called miners. The terminology varies depending on their role, which is defined by the protocol. In practice, consensus participants are machines running programs to follow and engage in the blockchain protocol.

\paragraph{Byzantine.} A Byzantine participant can deviate arbitrarily from the prescribed protocol. This appellation stems from the Byzantine General Problem introduced by Lamport, Shostak, and Pease \cite{lamport_byzantine_1982}. This problem is an analogy to simulate how reliable computer systems must handle malfunctioning components that give conflicting information to different parts of the system. The problem is stated as follow:

\begin{quote}
``A group of generals of the Byzantine army camped with their troops around an enemy city. Communicating only by messenger, the generals must agree upon a common battle plan. However, one or more of them may be traitors who will try to confuse the others. The problem is to find an algorithm to ensure that the loyal generals will reach agreement.''
\end{quote}

This thought experiment outlines the challenges of reaching consensus if some members of the group are compromised. One of the solution proposed involves making the assumption that, among the $n$ participants there are only $f$ traitors\footnote{$f$ stands for faulty, we use this term interchangeably to talk about Byzantine participants.}, 
such that $f<n/3$. 

This threshold stems from the two constraints we have in order to reach consensus: (a) the honest participants $n-f$ should be able to make a decision even if the traitors do not respond, and (b) the traitors should not be able to cause honest participants to make two different decisions. For (a), this means that $n-f$ must constitute a majority. For (b), this means that half of the honest participants $\frac{n-f}{2}$ plus the traitors $f$ must not constitute a majority.
This translates in:
\begin{equation*}
\begin{split}
    n-f &> \frac{n-f}{2}+f \\
    \Leftrightarrow \qquad 2n-4f &> n-f \\
    \Leftrightarrow \qquad \frac{n}{3} &> f .
\end{split}
\end{equation*}

This solution is fundamental in proving that reaching consensus is possible even in adversarial settings. 
The requirement for two-thirds of participants to agree in order to reach consensus is often called a supermajority.
A protocol that tolerates the presence of a Byzantine adversary while maintaining its guarantees is deemed \textit{Byzantine Fault Tolerant} (BFT).

\paragraph*{Finalization.} In blockchains, we say that a block is finalized if it permanently belongs to the chain. Non-finalized blocks are not guaranteed to always belong to the chain. 

\paragraph*{Consensus.} The consensus holds the same meaning as in everyday life, this is an agreement among a set of agents. The only difference is that in blockchain systems, agreements concern blocks, and consensus is
repeatedly achieved to agree on increasingly larger sets of blocks. In this context, our formalization of consensus in blockchain, partly based on Dolev et al. \cite{dolev_minimal_1987}, is the following.
\begin{definition}[Consensus]
    A blockchain protocol achieves consensus if it satisfies the following three properties:
    \begin{itemize}
        \item \textbf{Safety}:  No concurrent blocks can be finalized.
        \item \textbf{Liveness}: The set of finalized blocks continuously grows.
        \item \textbf{Validity}: The blocks agreed upon must have been proposed by one of the participants.
    \end{itemize}
\end{definition}

The validity property is often taken for granted in blockchain and what we are really interested in are the Safety and Liveness properties. This was part of our work to define them adequately before beginning our analysis, the formal definitions are presented in \autoref{chap:Eth-CModel-Analysis}.

The consensus in blockchain can be summarized as we want to make sure we agree on the same thing and there should never be a point after which we cannot agree anymore.
The protocol must not reach a point where finalization stops, preventing any new blocks from being finalized. In this case, the protocol is considered live.

These two properties are more thoroughly described in \autoref{chap:Eth-CModel-Analysis} and constitute part of our contribution.

\section{Contribution}
Our research focuses on the analysis of blockchain robustness. By robustness, we mean the blockchain's ability to avoid unsolvable forks (ensuring \textit{Safety}) while always maintaining the possibility to add new blocks (ensuring \textit{Liveness}). Our work ranges from distributed systems to game theoretic analysis. 
Distributed systems consider two types of agents: honest and Byzantine. Honest participants follow the prescribed protocol while Byzantine participants deviate from it arbitrarily. This binary classification overlooks the nuances introduced by rational players of game theoretic models, who act based on incentives rather than strict adherence to the protocol or malicious intent. The lack of research in distributed systems regarding blockchains' incentives that are yet paramount for participants motivated our work. 

In particular, we focused our analysis on one blockchain: Ethereum. Ethereum is the second biggest blockchain in terms of market capitalization and changed from a PoW consensus to a PoS one at the beginning of our work (2021). This transition called \textit{The Merge} brought a lot of changes with it and motivated our work. We often refer to Ethereum protocol as Ethereum PoS protocol to emphasize this change of paradigm. This resulted in an involved protocol lacking study and explanation.
By focusing on Ethereum, we saw an opportunity to contribute meaningfully to the field while closely observing the transition and its impacts on the properties of the blockchain.

\

The organization of this manuscript follows the chronological progression of our research. After establishing the necessary properties and definitions, we first analyze the protocol from a distributed systems perspective, initially without considering rewards and penalties. We then extend this analysis by incorporating penalties. Finally, the last technical chapter before the conclusion examines the protocol from a game-theoretic perspective. Here is a more detailed overview of each chapter:

\begin{itemize}
    \item \autoref{chap:background} presents the remaining terms useful throughout the manuscript and gives the essential properties and model we use for our analysis.
    \item \autoref{chap:Eth-CModel-Analysis} is the first part of our work, focusing on understanding the Ethereum PoS protocol. We started from scratch by reviewing the code to extract its properties. This led to the publication of our first paper \cite{pavloff_ethereum_2023} that we extended for a ACM DLT journal recently accepted.
    \item \autoref{chap:Eth-CModel-Penalties-Analysis} pushes the protocol analysis further by taking into account penalties present in the protocol. This is uncommon in distributed system analysis. This work led to another publication \cite{pavloff_byzantine_2024}
    \item \autoref{chap:Eth-GTModel-Analysis} is the last technical chapter and contains elements of our last paper \cite{pavloff_incentive_2024}. We model the interactions between block proposers and attesters as a game. We investigate the most profitable behaviors for the players.
    \item \autoref{chap:Conclusion} summarizes our results and opens on possible future research linked to our work.
\end{itemize}

\section{Related Work on Blockchain Analysis} \label{sec:RelatedWork}

As mentioned, our work during this thesis revolves around the case study of Ethereum Proof-of-Stake. 
We contribute to the field by formalizing the protocol, defining properties, presenting attacks against these properties, and then analyzing the strategies of rational validators. In this chapter, we present related works corresponding to each of these efforts.

\

\subsection{Blockchain Formalization}

The category of papers that aim to formalize blockchains includes all the \emph{white papers}. These are articles explaining the main features of a protocol, often written by the team behind the blockchain. This trend started in 2008 with the first blockchain, Bitcoin \cite{nakamoto_peer_2008}, and has been followed almost religiously by subsequent protocols. Ethereum is no exception, having released its first white paper in 2014 \cite{buterin_next_2014} when the Ethereum blockchain operated with a PoW consensus.

Following the release of a white paper by a protocol's team, other papers have emerged to challenge or complement the blockchain's description. For Bitcoin, this is exemplified by the work of Garay et al. \cite{garay_bitcoin_2015}, who analyzed the protocol's pseudo-code and deduced some of its properties. Similarly, Amoussou-Guenou et al. \cite{amoussou_dissecting_2019} proved the correctness of the Tendermint protocol, thereby complementing the initial white paper \cite{buchman_latest_2018}.  Alturki et al. \cite{alturki_towards_2019} used a proof assistant to prove the safety of the Algorand’s blockchain \cite{chen_algorand_2019}. Similarly, García-Pérez and Schett \cite{garcia_deconstructing_2019} provided a formal correctness proof of the Stellar Consensus Protocol (SCP). Amores-Sesar et al. \cite{amores_security_2020} study revealed that the  Ripple protocol might violate both safety and liveness, challenging the initial claims of Ripple's Byzantine fault tolerance.

The first section of our work has a similar aim. We extract the pseudo-code from the Ethereum specifications \cite{github_specs}, which is the description of how to implement the protocol, to formalize its properties. The most recent Ethereum white paper was released to explain the new protocol following the transition from PoW to PoS \cite{buterin_combining_2020}.

Outside of this line of work focusing on specific protocols, other efforts have aimed to provide formal foundations for blockchains. Anceaume et al. \cite{anceaume_abstract_2019} proposed a formalization of blockchains and their evolutions as Block Trees. The work of Anceaume et al. \cite{anceaume_finality_2021} described the different ways blockchains can ensure that blocks permanently belong to the chain, finalizing them. 
We rely on the definitions of Block Tree and finality to express the properties of the Ethereum protocol. 

\newpage

\subsection{Attacks and Vulnerabilities}
A considerable amount of work has been focused on identifying protocol vulnerabilities. The most famous example is probably the seminal work of Eyal and Sirer \cite{eyal_majority_2018}, which presents the selfish mining attack on Bitcoin. Eyal and Sirer show that in Bitcoin (and proof-of-work in general), miners can benefit from deviating from the prescribed protocol by withholding blocks for a while, to the detriment of honest miners. Many other examples exist: for instance, Amoussou-Guenou et al. \cite{amoussou_correctness_2018} pointed out a liveness vulnerability in the Tendermint protocol, and Neuder et al. \cite{neuder_defending_2020} presented an attack where nodes can reorganize Tezos' Emmy+ chain and then perform a double-spend attack. The protocol has been updated since then for a more robust solution proposed by Astefanoaei et al. \cite{astefanoaei_ternderbake_2021}.

Our work focuses on Ethereum, which is no stranger to protocol attacks. Neu et al. \cite{neu_ebb_2021} exhibited a balancing attack, highlighting the shortcomings of a consensus mechanism divided into two layers (finality gadget and fork choice rule). Mitigation against this attack was proposed, but Neu et al. \cite{neu_two_2022} overcame this mitigation with a new balancing attack. Schwarz-Schilling et al. \cite{schwarz_three_2021} presented \textit{reorg attacks}, attacks where the chain is reorganized leaving previous blocks orphaned, revealing that proposers could gain from disturbing the protocol by releasing their blocks late. 

Nakamura \cite{nakamura_prevention_2019} presented an attack called \textit{splitting attack}, in which the adversary sends messages to split the set of validators. However, Nakamura assumes that the adversary needs to control and manipulate network delays, which is a strong and potentially unrealistic assumption. More recently, Schwarz-Schilling et al. \cite{schwarz_three_2021} demonstrated through experiments that attackers can predict the proportion of validators receiving a given message within a specific timeframe with sufficient accuracy. This contradicts Nakamura's claim that the attack necessitates the adversary to control network delay. 

We contribute to this line of work by identifying flaws on the Ethereum PoS protocol. First, we outline a flaw regarding the liveness of the current Ethereum Proof-of-Stake protocol, emphasizing the importance of reconciling availability and finality. Our approach differs from Galletta et al. \cite{galletta_resilience_2023}, who aim to formally verify the \textit{Hybrid Casper} protocol \cite{buterin_casper_2017}, focusing on an outdated version. We formalize the current implementation of the protocol through pseudo-code and expose a liveness attack on the protocol.
Our work presents a form of the splitting attack where a message received by honest validators at different times is differently perceived to be on time or too late, splitting the validators into two 'views'. This different perception greatly influences which chain they consider canonical. This attack is based on the assumption that the adversary knows the network delay (in line with Schwarz-Schilling et al. \cite{schwarz_three_2021}) but does not control it. 


\subsubsection{Incentives}

Very few efforts in the literature have taken the incentive mechanism of protocols into account to evaluate how Byzantine validators could exploit it.
Initial efforts were made to intertwine the study of incentives with considerations of liveness and safety properties of the Ethereum protocol \cite{buterin_incentives_2020}. However, this early exploration discussed a preliminary version of the protocol \cite{buterin_casper_2017} and did not include an analysis of the inactivity leak. The inactivity leak is the mechanism that penalizes inactive validators by reducing their stake. The most recent version of the protocol by its founder \cite{buterin_combining_2020} does not mention this mechanism. The inactivity leak still lacks a detailed examination, and our work aims to fill this gap.

While mechanisms similar to Ethereum's inactivity leak exist elsewhere (e.g., \cite{wood_polkadot_2016, goodman_tezos_2014}), to the best of our knowledge, there has very few analysis of the risk associated with potentially draining honest stake in a Byzantine-prone environment. An investigation linking penalties with the actions of Byzantine validators is presented by Zhang et al. \cite{zhang_attestation_2023}. This work demonstrates how Byzantine validators can maliciously cause attestation penalties for honest validators. 

Our work is similar to Zhang et al. in scope, however we focus on more substantial penalties, i.e., the inactivity penalties and slashing. During the inactivity period, attestation penalties tend to be less significant. We found that the penalties could be detrimental for the protocol's safety if exploited by Byzantine participants.

\subsection{Rational Agents in Blockchain}
Incentives are not often considered in the distributed systems field. Game theory precisely addresses this gap.

Following the influential work of Eyal and Sirer \cite{eyal_majority_2014, eyal_majority_2018}, subsequent works (e.g., \cite{grunspan_selfish_2020, nayak_stubborn_2016, sapirshtein_optimal_2016, bar_efficient_2020}) use game-theoretic tools to analyze the maximum gain a rational miner can achieve by selfish mining, i.e., deviating from the proof-of-work protocols by withholding found blocks to gain an advantage in mining subsequent blocks. 
Also considering only rational participants, the work of Biais et al. \cite{biais_blockchain_2019} proves that while playing the proof-of-work game and following the longest chain’s rule is an equilibrium -where no participant can improve their outcome by changing their strategy-, multiple other equilibria exist where forks may persist. 

At the intersection of distributed systems, which model agents as either honest or Byzantine, and game theory, which models agents as rational, a mixed model was proposed: BAR (Byzantine, Altruistic, and Rational) \cite{abraham_distributed_2011}. However, the complexity of the analysis rises quickly with so many types of agents. The work of Halpern and Vilaça \cite{halpern_rational_2016} considers rational participants who can fail by crashing. They prove that in such a setting, there is no ex-post Nash equilibrium solving the fair consensus problem. Amoussou et al. \cite{amoussou_rational_2020} consider agents being either rational or Byzantine, exhibiting different equilibrium depending on the proportion of each. 

These works do not apply to Ethereum since its PoS mechanism is too different from classic PoW or classic BFT. Regarding PoS in general, Saleh \cite{saleh_blockchain_2020} showed that the \textit{nothing-at-stake}, problem in which PoS participant can extend simultaneously different fork without cost was prevented due to the value of the blockchain and thus their stake being decreased by such actions. For the Algorand blockchain Fooladgar et al. \cite{fooladgar_incentive_2020} showed that the cost and rewards of the protocol did not create an equilibrium in which participants followed the protocol. They proposed an adjustment of the rewards to entice selfish participants to cooperate.
Comparing our work to the game theory literature on the Ethereum PoS \emph{consensus}, Roughgarden \cite{roughgarden_transaction_2020} conducted a game-theoretic analysis of an Ethereum Improvement Proposal (EIP) to evaluate its impact on transaction and proposer rewards considering rational agents. Many works focus on MEV (Maximal Extractable Value), which involves taking advantage of the transaction ordering in a block. 
In contrast, we focus on game-theoretic analyses affecting the consensus directly. We differ from works like \cite{bhudia_game_2023}, which consider the possibility of making ransom demands without being detected. Schwarz-Schilling \cite{schwarz_three_2021} studies the optimal timing for proposers to propose their blocks.

Our endeavor is closest to the work of Carlsten et al. \cite{carlsten_instability_2016} and Tsabary and Eyal \cite{tsabary_eyal_2018} that study selfish behavior in Bitcoin using game theory when the \emph{only} source of rewards is transaction fees (no more coinbase transactions). Tsabary and Eyal show that the Bitcoin blockchain becomes unstable since block miners fork the Bitcoin blockchain to obtain the most lucrative transactions. 
Ethereum PoS does not reward block proposers with coinbase transactions; however, the presence of attesters and a different fork choice rule than Bitcoin's makes the analysis more complex. We focus on analyzing the behavior of block proposers and attesters in Ethereum PoS using game theory and demonstrate that the protocol tends to stabilize, even though a proposer might gain more by deviating from the prescribed strategy due to the initial asynchronous setting. Moreover, transaction fees do not play a crucial role in our analysis due to rewards stemming from attestations.

    \chapter{Background} \label{chap:background}

\minitoc

\chapterLettrine{C}{onducting} science primarily involves two approaches: empirical and theoretical.
Empirical science observes systems and studies the data and results from experiments, inferring general laws and their functioning. Theoretical science uses models to represent problems and employs these models to derive theorems and rules about the subject of study.
When creating a model, you must make assumptions and hypotheses that define the constrained reality you aim to understand. 
In this chapter, we outline models relevant to blockchain protocols and consensus.

\

\section{Distributed Computing Model}

Analyzing the Ethereum PoS, we consider participants to be \emph{validators}, forming a finite set $\Pi$. There are a total of $n$ validators. Being in PoS system, each validator owns a \emph{stake}, which refers to the amount of cryptocurrency (ETH) they possess. This stake serves as a metric of their investment and influence in the consensus protocol. Throughout this manuscript, the term "proportion" is used concerning a validator set to denote the ratio of their combined stake to the total staked.
A validator is interested in owning a stake as it comes with responsibilities that are rewarded. 
Initially capped at 32 ETH, a validator's stake has the potential to decrease. 

\subsection{Fault Model}
In distributed systems, the goal of a protocol is to guarantee certain properties despite the presence of \emph{faulty} participants. In our analysis, faulty participants will always be Byzantine ones \cite{lamport_byzantine_1982}.

Following the well-known work of Castro and Liskov \cite{castro_practical_1999}, we consider the worst-case scenario in which all faulty nodes are controlled by a single adversary. This assumption of a strong adversary is crucial for ensuring the reliability of critical distributed systems, where certain guarantees are expected even in the event of unexpected failures. However, during our analyses the adversary does not manipulate message delays between honest validators.

Unlike the dynamic adversary model considered by Chen and Micali \cite{chen_algorand_2019}, where the adversary can change the set of faulty nodes during execution, the works presented in this manuscript consider a \emph{static} adversary. This means that the set of faulty participants is determined at the start of the analysis and does not change throughout. We denote the proportion of Byzantine validators, which is the ratio of the sum of the stake of all Byzantine validators over the stake of all validators, by $\beta$ with $\beta< 1/3$. When analyzing changes in the Byzantine proportion over time, we define the \emph{initial} proportion as $\beta_0 < 1/3$.

\ 

The Ethereum PoS protocol aims to achieve Byzantine Fault Tolerance (BFT), ensuring the preservation of both Safety and Liveness properties as long as the initial proportion of Byzantine validators ($\beta_0$) remains strictly below $1/3$.

\section{Synchronization and Communication}
\subsection{Time}
Each participant maintains its own clock to keep track of time. We assume that all clocks are synchronized and run at the same pace. Any discrepancies in clock synchronization are considered as part of the message delay. 

In the Ethereum PoS protocol, time is measured in periods of 12 seconds, called \emph{slots}, with a period of 32 slots making up an \emph{epoch}, which serves as the largest time unit in the protocol. These timeframes are used to assign specific roles to validators at particular moments.

\subsection{Network}
We assume a \emph{partially synchronous model} \cite{dwork_consensus_1988}, which consists of an asynchronous period of unknown length followed by a synchronous period:
\begin{itemize}
\item During the asynchronous period, there is no bound on message delay. A message sent during this period has no guarantee of reaching its recipient before the asynchronous period ends.
\item Conversely, in the synchronous period, there is a known bound $\Delta$, ensuring that any message sent at time $t$ is received by time $t + \Delta$ at the latest.
\end{itemize}
The partially synchronous model begins with an asynchronous period that lasts until a global stabilization time (\texttt{GST}), after which the synchronous period begins.

Studying protocols under this model is common to ensure resilience under both good and bad network conditions. It is important to note that even with synchronized clocks, the presence of an asynchronous network before \texttt{GST} still qualifies the system as partially synchronous.

Validators communicate through message passing. We assume the existence of an underlying broadcast primitive, which operates as a best-effort broadcast. This means that when an honest validator broadcasts a value, all honest validators eventually receive it. Messages are signed with a digital signature, providing a mechanism for cryptographic identification and validation within the protocol.




\section{Game Model} \label{sec:gameModel}
This manuscript presents three distinct analyses: the first two from a distributed systems perspective, and the last one from a game-theoretic perspective. The shift in focus necessitates a corresponding change in the modeling approach.

\

In \autoref{chap:Eth-GTModel-Analysis}, we model the Ethereum PoS consensus protocol as a game where each player\footnote{The players are the participants of the Ethereum PoS that we also consider as a game for our analysis.} acts either as a \emph{proposer} or an \emph{attester}. In the ideal scenario, proposers propose blocks, and attesters broadcast attestations. The game unfolds over $s$ sequential \emph{slots}. There is one proposer and $a \in \mathbb{N}$ attesters per slot, leading to a total of $s$ proposers and $as$ attesters. The total number of slots $s$ is unknown to the players.

\

Similar to the approach in \cite{carlsten_instability_2016}, our game is based on the following assumptions: (i) The game occurs during a synchronous period where the network is fully synchronous, meaning there is no latency. This implies that once information (such as a block, attestation, or transaction) is broadcast, all players immediately become aware of it. (ii) The synchronous period follows an asynchronous period, during which there may have been delays in information transmission. This assumption aligns with the Ethereum protocol’s network behavior hypothesis \cite{buterin_combining_2020}. 

We model the interactions between proposers and attesters during $n$ slots in Ethereum PoS as a dynamic game in which actions occur sequentially. In each slot, the sequence of events is as follows: (i) a block is proposed at the beginning of the slot, (ii) new transactions are proposed, and (iii) all attesters for the slot simultaneously send their attestations.
The actions, rewards, and strategies of validators will be thoroughly described and analyzed in \autoref{chap:Eth-GTModel-Analysis}.

    \chapter{Ethereum PoS Analysis under the Distributed Computing Model} \label{chap:Eth-CModel-Analysis}

\minitoc

\chapterLettrine{A}{ll} of our work revolves around the Ethereum protocol. We intend to study this blockchain, which was once a PoW blockchain and is now a PoS blockchain, to develop a general understanding using a famous case study. To do so, we start with the roots of any protocol: its code. The Ethereum Foundation and the inventor of Ethereum, Vitalik Buterin, have produced a paper \cite{buterin_combining_2020} to explain the protocol and prove its properties. The issue is threefold: the paper discusses an outdated version of the protocol, not the entirety of the protocol is taken into account, and their results seem to contradict the crucial CAP theorem \cite{brewer_towards_2000}. For all these reasons combined, we begin our analysis from scratch, starting from the code. 
We begin our contribution by defining crucial properties for blockchains. We then use pseudo-code that reflects the protocol and the properties newly defined to evaluate the guarantees provided by the protocol. This first analysis does not consider the incentives part of the protocol.

\section{Safety and Liveness Properties}\label{sec:safetyAndLivenessProperties}
Once the protocol has been laid out, we can investigate its properties. Despite their names, blockchains are closer to \emph{block trees}. Forks can occur and cause the blockchain to have several branches rather than a unique chain. We adopt the formalization of Anceaume et al. \cite{anceaume_abstract_2019} of blockchain data structure as a block tree. 
Indeed, the blockchain takes the form of a tree in which every node is a block pointing to its unique parent, and the tree's root is the \emph{genesis block}. 
Among the different branches of the block tree, the protocol indicates a unique branch, or chain, to build upon with a so-called fork choice rule (e.g., the longest chain rule in Bitcoin).
The selected chain is called the \emph{canonical chain}.

\begin{definition}[\textbf{Canonical chain}]
 We call \textbf{canonical chain} the chain designated as the one to build upon by the fork choice rule.
 Considering the view of the chain of an honest validator $i$, $i$'s associated canonical chain is noted $\mathcal{C}_i$. 
\end{definition}

The blocks in the canonical chain can be finalized or not.

\begin{definition}[\textbf{Finalized block}]
A block is finalized for a validator $i$ if and only if the block cannot be revoked, i.e., it permanently belongs to the canonical chain $\mathcal{C}_i$.
\end{definition}

\emph{Note:} It stems from the definition that all the predecessors of a finalized block are finalized. 

\begin{definition}[\textbf{Finalized chain}]
The finalized chain is the chain constituted of all the finalized blocks.
\end{definition}

\emph{Note:} The finalized chain $\mathcal{C}_{fi}$ is always a prefix of any canonical chain $\mathcal{C}_i$.\\

To analyze the protocol, one needs to examine the capability of the Ethereum Proof-of-Stake protocol to construct a consistent blockchain (safety), to allow validators to add blocks despite network partitions and failures (availability), and to make progress on the finalization of new blocks (liveness). These are paramount properties characterizing blockchains.
Safety, availability, and liveness are expressed as follows:

\setcounter{property}{0}
\begin{property}[\textbf{Safety}]
A blockchain is consistent or \textbf{safe} if, for any two honest validators with a finalized chain, one chain is necessarily the prefix of the other. More formally, for two validators $i$ and $j$ with respective finalized chains $\mathcal{C}_{fi}$ and $\mathcal{C}_{fj}$, $\mathcal{C}_{fi}$ is the prefix of $\mathcal{C}_{fj}$ or vice versa.
\end{property}

\begin{property}[\textbf{Liveness}]
A blockchain is \textbf{live} if the finalized chain is ever growing.
\end{property}

\begin{property}[\textbf{Availability}]
A blockchain is \textbf{available} if the following two conditions hold: (1) any honest validator is able to append a block to its canonical chain in bounded time, regardless of the failures of other validators and the network partitions; (2) the canonical chains of all honest validators are eventually growing, i.e., given a block $b_k$ added to a canonical chain at a distance $d$ from the genesis block $b_0$, where the distance is the number of blocks separating $b_k$ from $b_0$, eventually a block $b_l$ will be added to the canonical chain at a distance $d' > d$.
\end{property}


The key difference between the finalized chain and the canonical chain is that blocks in the finalized chain are permanent and cannot be revoked. In contrast, the canonical chain can switch branches, meaning blocks in the previously chosen branch can potentially be revoked. 
Availability, on the other hand, guarantees that adding blocks to the canonical chain is a wait-free operation whose time to complete does not depend on network failures or Byzantine behaviors. 
Availability also implies that blocks are constantly added in such a way that the height of the canonical chain eventually grows. This property avoids the pathological scenario in which all the blocks are added to the genesis block to form a star.

As in any distributed system, blockchains are faced with the dilemma brought by the CAP Theorem. This theorem states that no distributed system can satisfy these three properties at the same time: \emph{consistency}, \emph{availability}, and \emph{partition tolerance}. Indeed, if network partitions occur, either the system remains available at the expense of consistency, or it stops making progress until the network partition is resolved to guarantee consistency. 
This means that no blockchain can simultaneously be available and consistent. However, by maintaining the canonical and the finalized chain simultaneously, Ethereum Proof-of-Stake aims to offer both safety and availability. The canonical chain aims to be available but without guaranteeing consistency all the time, while the finalized chain falls on the other side of the spectrum, guaranteeing consistency without availability. Therefore, the finalized chain will finalize blocks only when it is safe to do so, whereas the canonical chain will still be available during network partitions (caused by network failures or attacks). 
The only caveat here is that the finalized chain grows by finalizing blocks of the canonical chain, which means that the properties of the two chains are interdependent. In particular, to assure liveness, it is necessary that the canonical chain steadily grows. 
This interdependence is a source of vulnerability as we show in the remainder of our analysis.



\section{Ethereum PoS protocol}\label{sec:EthPoSprotocol}

\subsection{Overview}
\label{subsec:overview}
The Ethereum Proof-of-Stake (PoS) protocol design is quite involved. We identify, similarly to \cite{neu_ebb_2021}, the objectives underlying its design as follows: (i) finalizing blocks and 
(ii) having an \emph{available} canonical chain that does not rely on block finality to grow. 
To this end, the Ethereum PoS protocol combines two blockchain designs: a Nakamoto-style protocol to build the tree of blocks containing the transactions and a BFT finalization protocol to progressively finalize blocks in the tree. The objective is to keep the blockchain creation process always available while guaranteeing the finalization of blocks through Byzantine-tolerant voting mechanisms. The finalization mechanism is a \emph{Finality Gadget} called \emph{Casper FFG}, and the fork choice rule to select canonical chains is \emph{LMD GHOST}.

Before introducing how the fork choice rule and the finality gadget work together, we will introduce the following basic concepts: (i) slots, epochs, and checkpoints, which set the pace of the protocol allowing validators to synchronize together on the different steps, (ii) committees formation and assignment of roles to validators as proposers and voters for each slot, and (iii) the different types of votes the validators must send in order to grow and maintain the canonical chain as well as the finalized chain.

In this section, we focus on providing a formalized version of the protocol through pseudo-code, following the specification given by the Ethereum Foundation \cite{github_specs}. Note that a description of an initial plan of the protocol was proposed by Buterin et al. in \cite{buterin_combining_2020}. We describe and formalize the current implementation of the protocol \cite{github_specs}.

\subsubsection{Slots, Epochs \& Checkpoints}
\label{subsubsec:time}
In proof-of-work protocols, such as originally described in \cite{nakamoto_peer_2008}, the average frequency of block creation is predetermined in the protocol, and the mining difficulty changes to maintain that pace. In contrast, in Ethereum PoS, it is assumed that validators have synchronized clocks to propose blocks at regular intervals. 
More specifically, the protocol uses time frames called \emph{slots} (12 seconds) and \emph{epochs} (6 minutes and 24 seconds). 
Each epoch is comprised of $32$ slots. Epoch $0$ contains slot $0$ to slot $31$, then epoch $1$ slot $32$ to slot $63$, and so on.
These slots and epochs allow associating the validators' roles to the corresponding time frame. 

An essential feature of epochs is the \emph{checkpoint}. A checkpoint is a pair (block, epoch) $(b,e)$ where $b$ is the block of the first slot\footnote{In the event of an epoch without a block for the first slot, the block used for the checkpoint is the last block in the canonical chain, belonging to a previous epoch. On the contrary, if the proposer of the first slot proposes multiple blocks, this will create multiple checkpoints for the other validators to choose from using the fork choice rule.} of epoch $e$. \autoref{fig:slotEpoch} represents the structure of an epoch in Ethereum PoS, with the checkpoint being represented by the hexagonal blue shape.

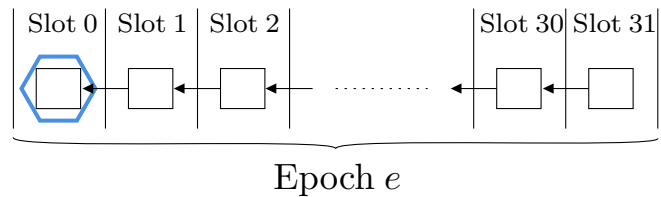
\begin{figure}
    \centering
    \resizebox{0.6\linewidth}{!}{
        \begin{tikzpicture}[x=0.75pt,y=0.75pt,yscale=-1,xscale=1]

\draw   (87.05,80) -- (109.05,80) -- (109.05,100) -- (87.05,100) -- cycle ;
\draw   (132.69,80) -- (154.69,80) -- (154.69,100) -- (132.69,100) -- cycle ;
\draw    (132.69,90) -- (112.5,90) ;
\draw [shift={(109.5,90)}, rotate = 360] [fill={rgb, 255:red, 0; green, 0; blue, 0 }  ][line width=0.08]  [draw opacity=0] (5.36,-2.57) -- (0,0) -- (5.36,2.57) -- cycle    ;
\draw   (269.62,80) -- (291.62,80) -- (291.62,100) -- (269.62,100) -- cycle ;
\draw   (315.26,80) -- (337.26,80) -- (337.26,100) -- (315.26,100) -- cycle ;
\draw    (269.62,90) -- (249.8,90) ;
\draw [shift={(246.8,90)}, rotate = 360] [fill={rgb, 255:red, 0; green, 0; blue, 0 }  ][line width=0.08]  [draw opacity=0] (5.36,-2.57) -- (0,0) -- (5.36,2.57) -- cycle    ;
\draw    (315.26,90) -- (295.5,90) ;
\draw [shift={(292.5,90)}, rotate = 360] [fill={rgb, 255:red, 0; green, 0; blue, 0 }  ][line width=0.08]  [draw opacity=0] (5.36,-2.57) -- (0,0) -- (5.36,2.57) -- cycle    ;
\draw  [dash pattern={on 0.84pt off 2.51pt}]  (189.75,90) -- (235.39,90) ;
\draw    (121.28,50) -- (121.28,110) ;
\draw    (166.93,50) -- (166.93,110) ;
\draw    (258.21,50) -- (258.21,110) ;
\draw    (303.85,50) -- (303.85,110) ;
\draw    (349.49,50) -- (349.49,110) ;
\draw    (29.96,114.4) .. controls (28.56,121.11) and (189.33,114.28) .. (189.87,120) ;
\draw    (189.87,120) .. controls (189.36,113.71) and (350.89,120.31) .. (349.49,114.88) ;
\draw    (178.34,90) -- (158,90) ;
\draw [shift={(155,90)}, rotate = 360] [fill={rgb, 255:red, 0; green, 0; blue, 0 }  ][line width=0.08]  [draw opacity=0] (5.36,-2.57) -- (0,0) -- (5.36,2.57) -- cycle    ;
\draw  [color={rgb, 255:red, 74; green, 144; blue, 226 }  ,draw opacity=1 ][line width=1.5]  (70.64,90.02) -- (61.48,106.02) -- (43.16,106.02) -- (34,90.02) -- (43.16,74.03) -- (61.48,74.03) -- cycle ;
\draw    (87.05,90) -- (67.23,90) ;
\draw [shift={(64.23,90)}, rotate = 360] [fill={rgb, 255:red, 0; green, 0; blue, 0 }  ][line width=0.08]  [draw opacity=0] (5.36,-2.57) -- (0,0) -- (5.36,2.57) -- cycle    ;
\draw    (75.64,50) -- (75.64,110) ;
\draw    (30,50) -- (30,110) ;
\draw   (41.32,80.02) -- (63.32,80.02) -- (63.32,100.02) -- (41.32,100.02) -- cycle ;

\draw (54.5,57.14) node  [font=\scriptsize,xscale=1.25,yscale=1.25]  {$\mathrm{Slot} \ 0$};
\draw (190.6,135.55) node  [font=\small,xscale=1.25,yscale=1.25]  {$\mathrm{Epoch} \ e$};
\draw (99,57.14) node  [font=\scriptsize,xscale=1.25,yscale=1.25]  {$\mathrm{Slot} \ 1$};
\draw (144.65,57.14) node  [font=\scriptsize,xscale=1.25,yscale=1.25]  {$\mathrm{Slot} \ 2$};
\draw (282.11,57.14) node  [font=\scriptsize,xscale=1.25,yscale=1.25]  {$\mathrm{Slot} \ 30$};
\draw (328.32,57.14) node  [font=\scriptsize,xscale=1.25,yscale=1.25]  {${\textstyle \mathrm{Slot} \ 31}$};

\end{tikzpicture}
    }
    \caption{Ethereum protocol Structure}
    \label{fig:slotEpoch}
\end{figure}

\subsubsection{Validators \& Committees}
\label{subsubsec:validator&committee}
Validators have two main roles: \emph{proposer} and \emph{attester}. The proposer's role consists of proposing a block during a specific slot\footnote{The current protocol specifications \cite{github_specs} indicate that honest validators should send their block proposition during the first third of their designated slot.}. This role is pseudo-randomly\footnote{Detailed explanation in \autoref{ssec:pseudo-randomness}.} assigned to 32 validators per epoch (one for each slot). The attester's role consists of producing an attestation sharing the validator's view of the chain. This role is assigned once per epoch to each validator.

In each epoch, a validator is assigned to exactly one committee (of attesters). A committee $\mathtt{C}_j$ is a subset of the whole set of validators. Each validator belongs to exactly one committee, i.e., $\forall j \neq k,  \mathtt{C}_j \bigcap \mathtt{C}_k = \emptyset $ and for each epoch $\bigcup_{i} \mathtt{C}_i = \Pi $. Each committee is associated with a slot. During this slot, each member of the committee will have to cast an \emph{attestation} to indicate its view of the chain.

In short, during an epoch, validators are all attesters once and have a small probability of being proposers ($32/n$). The roles of proposer and attester are entirely distinct, i.e., the proposer of a slot is not necessarily an attester of that slot.

\subsubsection{Vote \& Attestation}
\label{subsubsec:vote}
There are two types of votes in Ethereum PoS: the \emph{block vote}\footnote{Also called GHOST vote in \cite{buterin_combining_2020} and in the specifications \cite{github_specs}.} and the \emph{checkpoint vote}\footnote{Also called FFG vote in \cite{buterin_combining_2020} and in the specifications \cite{github_specs}.}. The message containing these two votes is called an \emph{attestation}. 
During an epoch, each validator must make one attestation. The attestation should be sent during a specific time slot, which is determined by the validator's committee. 
The two types of votes, checkpoint vote and block vote, have very distinct purposes. The checkpoint vote is used to finalize blocks and grow the finalized chain, while the block vote is used to determine the canonical chain. 
Although validators cast their two types of votes in one attestation, an important distinction must be made between the two. Indeed, the two types of votes do not require the same conditions to be taken into account. The checkpoint vote of an attestation is only considered when the attestation is included in a block. In contrast, the block vote is considered one slot after its emission, whether it is included in a block or not.

The code associated with the production of attestations is described in \autoref{algo:prepareAttestation} at \autoref{subsec:code}. We then describe in \autoref{algo:syncAttestation} how the reception of attestations is handled.

\subsubsection{Finality Gadget}
\label{subsubsec:finalityGadget}
The finality gadget is the mechanism that aims at finalizing blocks. The finality gadget grows the finalized chain independently of block production. This decoupling of the finality mechanism from block production permits block availability even when the finalizing process is slowed down. This differs from protocols like Tendermint \cite{buchman_latest_2018}, where a new block can be added to the chain only after being finalized.

The finality gadget works at the level of epochs. Instead of finalizing blocks one by one, the protocol uses checkpoint votes to finalize entire epochs. We now present in more detail how the finality gadget of Ethereum PoS grows the finalized chain.

Recall that to be taken into account, a checkpoint vote needs to be included in a block. The vote will then influence the behavior of validators regarding this particular branch. Thus, in \autoref{algo:Casper} of \autoref{subsec:code}, the function \texttt{countMatchingCheckpointVote} only counts the matching checkpoint votes of attestations included in a block.

\paragraph*{Justification} \label{paragraph:Justification}
The justification process is a step towards achieving finalization\footnote{The genesis checkpoint (i.e., the checkpoint of the first epoch) is the exception to this rule: it is justified and finalized by definition.}. It operates on checkpoints at the level of epochs. Justification occurs thanks to checkpoint votes. The checkpoint vote contains a pair of checkpoints: the checkpoint \emph{source} and the checkpoint \emph{target}. We can count with \texttt{countMatchingCheckpointVote} the sum of balances of the validators' checkpoint votes with the same source and target. If validators controlling more than two-thirds of the stake make the same checkpoint vote, then we say there is a \emph{supermajority link} from the checkpoint source to the checkpoint target. The checkpoint target of a supermajority link is said to be \emph{justified}.

More formally, a checkpoint vote is in the form of a pair of checkpoints: $\big((a,e_a),(b,e_b)\big)$, also noted $(a,e_a) \xrightarrow{} (b,e_b)$. For the checkpoint vote $(a,e_a) \xrightarrow{} (b,e_b)$, we call $(a,e_a)$ the checkpoint source and $(b,e_b)$ the checkpoint target. The checkpoint source is necessarily from an earlier epoch than the checkpoint target, i.e., $e_a<e_b$. In line with \cite{buterin_combining_2020}, we say there is a \emph{supermajority link} from checkpoint $(a,e_a)$ to checkpoint $(b,e_b)$ if validators controlling more than two-thirds of the stake cast an attestation with the checkpoint vote $(a,e_a) \xrightarrow{} (b,e_b)$. In this case, we write $(a,e_a) \xrightarrow{\texttt{J}} (b,e_b)$, and the checkpoint $(b,e_b)$ is justified.

\paragraph*{Finalization} \label{paragraph:Finalization}
The finalization process aims at finalizing checkpoints, thus growing the finalized chain. Checkpoints need to be justified before being finalized. Let us illustrate the finalization process with the two scenarios that can lead to finalization. The first case presents the main scenario in the synchronous setting. It shows how a checkpoint can be finalized in two epochs, the minimum number of epochs needed for finalization.

\subparagraph*{Case 1:} The scenario is depicted in \autoref{fig:finalizationCase1}.
\begin{enumerate}
    \item Let $A=(a,e)$ and $B=(b,e+1)$ be checkpoints of two consecutive epochs such that $A=(a,e)$ is justified.
    \item A supermajority link occurs between checkpoints $A$ and $B$ where $A$ is the source and $B$ the target. This justifies checkpoint $B$. Hence, we can write: $(a, e) \xrightarrow{\texttt{J}} (b, e+1)$ or equivalently $A \xrightarrow{\texttt{J}} B$.
    \item This leads to $A$ being finalized.
\end{enumerate}

\newcommand{\sizeFig}{0.5}
\begin{figure}
    \centering
    \resizebox{\sizeFig\linewidth}{!}{
    \begin{tikzpicture}[node distance = 1.5cm, 
        finalizednode/.style={regular polygon, regular polygon sides=6,, draw=black, fill=green, minimum size=7mm},        finalizednode2/.style={regular polygon, regular polygon sides=6,, draw=black, minimum size=5.5mm},
        justifiednode/.style={regular polygon, regular polygon sides=6,, draw=black, minimum size=7mm},         
        justifiednode2/.style={regular polygon, regular polygon sides=6,, draw=black, minimum size=7mm},
        checkpointnode/.style={regular polygon, regular polygon sides=6,, draw=black, minimum size=7mm},
        squarednode/.style={regular polygon, regular polygon sides=6,, draw=gray!60, fill=gray!5, minimum size=5mm},
        ]
        \node (Dots1)      {$\cdots$};
        \node[justifiednode]      (A)       [right of= Dots1]              {$A$};
        \node[justifiednode2]      (A2)       [right of= Dots1]              {};
        \node[checkpointnode]        (B)       [right of= A] {$B$};
         \node (Dots2)   [right of= B]   {$\cdots$};
        
        \draw[-] (A.east) -- (B.west);
        \draw[-] (Dots1.east) -- (A.west);
        \draw[-] (B.east) -- (Dots2.west);
        
\end{tikzpicture}  
    }\hspace{0cm}
    \resizebox{\sizeFig\linewidth}{!}{
    \begin{tikzpicture}[node distance = 1.5cm, 
        finalizednode/.style={regular polygon, regular polygon sides=6,, draw=black, fill=green, minimum size=7mm},        finalizednode2/.style={regular polygon, regular polygon sides=6,, draw=black, minimum size=5.5mm},
        justifiednode/.style={regular polygon, regular polygon sides=6,, draw=black, minimum size=7mm},      justifiednode2/.style={regular polygon, regular polygon sides=6,, draw=black, minimum size=7mm},
        checkpointnode/.style={regular polygon, regular polygon sides=6,, draw=black, minimum size=7mm},
        squarednode/.style={regular polygon, regular polygon sides=6,, draw=gray!60, fill=gray!5, minimum size=5mm},
        ]
        \node (Dots1)      {$\cdots$};
        \node[justifiednode]      (A)       [right of= Dots1]              {$A$};
        \node[justifiednode2]      (A2)       [right of= Dots1]              {};
        \node[justifiednode]        (B)       [right of= A] {$B$};
        \node[justifiednode2]        (B2)       [right of= A] {};
         \node (Dots2)   [right of= B]   {$\cdots$};
        
        \draw[-] (A.east) -- (B.west);
        \draw[-] (Dots1.east) -- (A.west);
        \draw[-] (B.east) -- (Dots2.west);
        \draw[->,double] (A.north) to[out=45]  (B.north);
        
\end{tikzpicture}   
    }
    \hspace{0cm}
    \resizebox{\sizeFig\linewidth}{!}{
    \begin{tikzpicture}[node distance = 1.5cm, 
        finalizednode/.style={regular polygon, regular polygon sides=6,, draw=black, fill=green, minimum size=7mm},
        finalizednode2/.style={regular polygon, regular polygon sides=6,, draw=black, minimum size=7mm},
        justifiednode/.style={regular polygon, regular polygon sides=6,, draw=black, minimum size=7mm},      justifiednode2/.style={regular polygon, regular polygon sides=6,, draw=black, minimum size=7mm},
        checkpointnode/.style={regular polygon, regular polygon sides=6,, draw=black, minimum size=7mm},
        ]
        \node (Dots1)      {$\cdots$};
        \node[finalizednode]      (A)       [right of= Dots1]              {$A$};
        \node[finalizednode2]      (A2)       [right of= Dots1]              {};
        \node[justifiednode]        (B)       [right of= A] {$B$};
        \node[justifiednode2]        (B2)       [right of= A] {};
         \node (Dots2)   [right of= B]   {$\cdots$};
        
        \draw[-] (A.east) -- (B.west);
        \draw[-] (Dots1.east) -- (A.west);
        \draw[-] (B.east) -- (Dots2.west);
        \draw[->,double] (A.north) to[out=45]  (B.north);
        
        \end{tikzpicture} 
    }
    \caption{The figure depicts the finalization scenario of \textbf{Case 1} with the 3 steps from top to bottom. We represent a checkpoint with a hexagon, a justified checkpoint with a double hexagon, and a finalized checkpoint with a colored double hexagon. The arrow between two checkpoints indicates a supermajority link.}
    \label{fig:finalizationCase1}
\end{figure}

The second case illustrates the scenario in which two consecutive checkpoints are justified but not finalized. This means that the current highest justified checkpoint (e.g., $B$ in \autoref{fig:finalizationCase2}) was not justified with a supermajority link having the previous checkpoint $A$ as its source. Then, a new justification occurs with the source and target being at the maximum distance (2 epochs) for the source to become finalized. It is important to note that there is no limit on the distance between two checkpoints for justification to be possible. This limit only exists for finalization.

\subparagraph*{Case 2:} The scenario is depicted in \autoref{fig:finalizationCase2}.
\begin{enumerate}
    \item Let $A=(a,e)$, $B=(b,e+1)$, and $C=(c,e+2)$ be checkpoints of consecutive epochs such that $A$ and $B$ are justified. Since there is no supermajority link between $A$ and $B$, $A$ cannot be finalized as in Case 1.
    \item Now, a supermajority link occurs between checkpoints $A$ and $C$ where $A$ is the source and $C$ the target. This justifies checkpoint $C$, i.e., $A \xrightarrow{\texttt{J}} C$.
    \item This leads to $A$ being finalized.
\end{enumerate}

\begin{figure}
    \centering
    \resizebox{\sizeFig\linewidth}{!}{
    \begin{tikzpicture}[node distance = 1.5cm, 
    finalizednode/.style={regular polygon, regular polygon sides=6,, draw=black, fill=green, minimum size=7mm},        finalizednode2/.style={regular polygon, regular polygon sides=6,, draw=black, minimum size=7mm},
    justifiednode/.style={regular polygon, regular polygon sides=6,, draw=black, minimum size=7mm},         justifiednode2/.style={regular polygon, regular polygon sides=6,, draw=black, minimum size=7mm},
    checkpointnode/.style={regular polygon, regular polygon sides=6,, draw=black, minimum size=7mm},
    squarednode/.style={regular polygon, regular polygon sides=6,, draw=gray!60, fill=gray!5, minimum size=5mm},
    ]
    \node (Dots1)      {$\cdots$};
    \node[justifiednode]      (A)       [right of= Dots1]              {$A$};
    \node[justifiednode2]      (A2)       [right of= Dots1]              {};
    \node[justifiednode]        (B)       [right of= A] {$B$};
    \node[justifiednode2]        (B2)       [right of= A] {};
    \node[checkpointnode]        (C)       [right of= B] {$C$};
     \node (Dots2)   [right of= C]   {$\cdots$};
    
    \draw[-] (A.east) -- (B.west);
    \draw[-] (Dots1.east) -- (A.west);
    \draw[-] (B.east) -- (C.west);
    \draw[-] (C.east) -- (Dots2.west);
    
    \end{tikzpicture}  
    }
    \hspace{0cm}
    \resizebox{\sizeFig\linewidth}{!}{
    \begin{tikzpicture}[node distance = 1.5cm,
    finalizednode/.style={regular polygon, regular polygon sides=6,, draw=black, fill=green, minimum size=7mm},        finalizednode2/.style={regular polygon, regular polygon sides=6,, draw=black, minimum size=7mm},
    justifiednode/.style={regular polygon, regular polygon sides=6,, draw=black, minimum size=7mm},         justifiednode2/.style={regular polygon, regular polygon sides=6,, draw=black, minimum size=7mm},
    checkpointnode/.style={regular polygon, regular polygon sides=6,, draw=black, minimum size=7mm},
    squarednode/.style={regular polygon, regular polygon sides=6,, draw=gray!60, fill=gray!5, minimum size=5mm},
    ]
    \node (Dots1)      {$\cdots$};
    \node[justifiednode]      (A)       [right of= Dots1]              {$A$};
    \node[justifiednode]        (B)       [right of= A] {$B$};
    \node[justifiednode]        (C)       [right of= B] {$C$};
    \node[justifiednode2]      (A2)       [right of= Dots1]              {};
    \node[justifiednode2]        (B2)       [right of= A] {};
    \node[justifiednode2]        (C2)       [right of= B] {};
     \node (Dots2)   [right of= C]   {$\cdots$};
    
    \draw[-] (A.east) -- (B.west);
    \draw[-] (Dots1.east) -- (A.west);
    \draw[-] (B.east) -- (C.west);
    \draw[-] (C.east) -- (Dots2.west);
    \draw[->,double] (A.north) to[out=45]  (C.north);
    \end{tikzpicture} 
    }
    \hspace{0cm}
    \resizebox{\sizeFig\linewidth}{!}{
    \begin{tikzpicture}[node distance = 1.5cm, 
    finalizednode/.style={regular polygon, regular polygon sides=6,, draw=black, fill=green, minimum size=7mm},        finalizednode2/.style={regular polygon, regular polygon sides=6,, draw=black, minimum size=7mm},
    justifiednode/.style={regular polygon, regular polygon sides=6,, draw=black, minimum size=7mm},         justifiednode2/.style={regular polygon, regular polygon sides=6,, draw=black, minimum size=7mm},
    checkpointnode/.style={regular polygon, regular polygon sides=6,, draw=black, minimum size=7mm},
    squarednode/.style={regular polygon, regular polygon sides=6,, draw=gray!60, fill=gray!5, minimum size=5mm},
    ]
    \node (Dots1)      {$\cdots$};
    \node[finalizednode]      (A)       [right of= Dots1]              {$A$};
    \node[finalizednode2]      (A2)       [right of= Dots1]              {};
    \node[justifiednode]        (B)       [right of= A] {$B$};
    \node[justifiednode]        (C)       [right of= B] {$C$};
    \node[justifiednode2]        (B2)       [right of= A] {};
    \node[justifiednode2]        (C2)       [right of= B] {};
     \node (Dots2)   [right of= C]   {$\cdots$};
    
    \draw[-] (A.east) -- (B.west);
    \draw[-] (Dots1.east) -- (A.west);
    \draw[-] (B.east) -- (C.west);
    \draw[-] (C.east) -- (Dots2.west);
    \draw[->,double] (A.north) to[out=45]  (C.north);
    \end{tikzpicture} 
    }
    \caption{The figure depicts the finalization scenario of \textbf{Case 2} with the 3 steps from top to bottom. We represent a checkpoint with a hexagon, a justified checkpoint with a double hexagon, and a finalized checkpoint with double hexagon coloured. 
    The arrow between two checkpoints indicates a supermajority link.}
  \label{fig:finalizationCase2}
\end{figure}
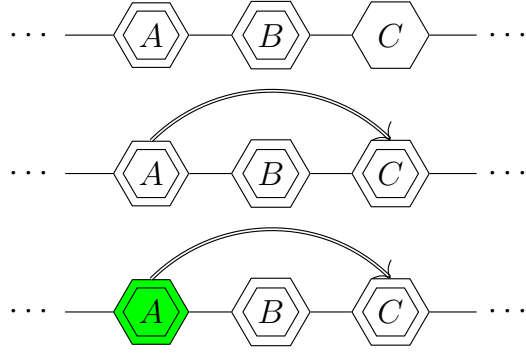

These two cases illustrate the fact that for a checkpoint to become finalized, it needs to be the source of a supermajority link between justified checkpoints. Once a checkpoint is finalized, all the blocks leading to it (including the block in the pair constituting the checkpoint) become finalized. 
We now describe the conditions for a checkpoint to be finalized more formally. Let $(a,e_a)$ and $(b,e_b)$ be two checkpoints such that $e_a < e_b$.
The checkpoint $(a,e_a)$ is finalized if the following conditions are respected: 
\begin{itemize}
    \item \textbf{Source justified:} The checkpoint $(a,e_a)$ is justified.
    \item \textbf{Supermajority link:} There exists a supermajority link $(a,e_a) \xrightarrow{\texttt{J}} (b,e_b)$.
    \item \textbf{Maximal gap:} $e_b - e_a \leq 2$.\footnote{This last condition necessitating the two checkpoints to be at most 2 epochs away from each other is also called \emph{2-finality} \cite{buterin_combining_2020}.} Moreover, if $e_b - e_a = 2$, then the checkpoint in between at epoch $e_a + 1$ ($= e_b - 1$) must be justified.
\end{itemize}
The importance of the last condition is illustrated by \autoref{fig:notFinalizationCase}. 
In practice, these three conditions are only applied to the last four epochs. As mentioned in \cite{buterin_combining_2020}, at the implementation level, checkpoints more than 4 epochs old are not considered for finalization. All the conditions for finalization are illustrated by the last 4 conditions of \autoref{algo:Casper} in \autoref{subsec:code}.

\begin{figure}
    \centering
    \resizebox{\sizeFig\linewidth}{!}{
    \begin{tikzpicture}[node distance = 1.5cm, 
    finalizednode/.style={regular polygon, regular polygon sides=6,, draw=black, fill=green, minimum size=7mm},        finalizednode2/.style={regular polygon, regular polygon sides=6,, draw=black, minimum size=7mm},
    justifiednode/.style={regular polygon, regular polygon sides=6, draw=black, minimum size=7mm},
    justifiednode2/.style={regular polygon, regular polygon sides=6,, draw=black, minimum size=7mm},
    checkpointnode/.style={regular polygon, regular polygon sides=6, draw=black, minimum size=7mm},
    squarednode/.style={regular polygon, regular polygon sides=6, draw=gray!60, fill=gray!5, minimum size=5mm},
    ]
    \node (Dots1)      {$\cdots$};
    \node[justifiednode]      (A)       [right of= Dots1]              {$A$};
    \node[justifiednode2]      (A2)       [right of= Dots1]              {};
    \node[checkpointnode]        (B)       [right of= A] {$B$};
    \node[checkpointnode]        (C)       [right of= B] {$C$};
     \node (Dots2)   [right of= C]   {$\cdots$};
    
    \draw[-] (A.east) -- (B.west);
    \draw[-] (Dots1.east) -- (A.west);
    \draw[-] (B.east) -- (C.west);
    \draw[-] (C.east) -- (Dots2.west);
    
    \end{tikzpicture}    
    }
    \hspace{0cm}
    \resizebox{\sizeFig\linewidth}{!}{
    \begin{tikzpicture}[node distance = 1.5cm,
    finalizednode/.style={regular polygon, regular polygon sides=6,, draw=black, fill=green, minimum size=7mm},        finalizednode2/.style={regular polygon, regular polygon sides=6,, draw=black, minimum size=7mm},
    justifiednode/.style={regular polygon, regular polygon sides=6,, draw=black, minimum size=7mm},         
    justifiednode2/.style={regular polygon, regular polygon sides=6,, draw=black, minimum size=7mm},
    checkpointnode/.style={regular polygon, regular polygon sides=6,, draw=black, minimum size=7mm},
    squarednode/.style={regular polygon, regular polygon sides=6,, draw=gray!60, fill=gray!5, minimum size=5mm},
    ]
    \node (Dots1)      {$\cdots$};
    \node[justifiednode]      (A)       [right of= Dots1]              {$A$};
    \node[justifiednode2]      (A2)       [right of= Dots1]              {};
    \node[checkpointnode]        (B)       [right of= A] {$B$};
    \node[justifiednode]        (C)       [right of= B] {$C$};
    \node[justifiednode2]        (C2)       [right of= B] {};
     \node (Dots2)   [right of= C]   {$\cdots$};
    \draw[-] (A.east) -- (B.west);
    \draw[-] (Dots1.east) -- (A.west);
    \draw[-] (B.east) -- (C.west);
    \draw[-] (C.east) -- (Dots2.west);
    \draw[->,double] (A.north) to[out=45]  (C.north);
    \end{tikzpicture} 
    }
    \caption{This figure illustrates the case of two checkpoints $A$ and $C$ respecting all the conditions for finalization but the one that stipulates that a checkpoint $B$ in-between must be justified for $A$ to be finalized.}
    \label{fig:notFinalizationCase}
\end{figure}
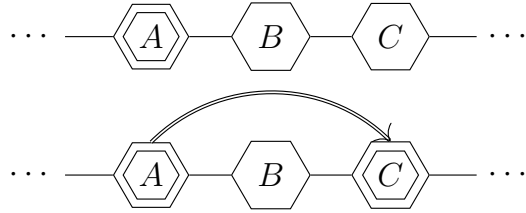

\subsubsection{Fork choice rule \& Block proposition}
\label{subsubsec:forkChoiceRule&blockProposition}
The fork choice rule is the mechanism that allows each validator to determine the canonical chain depending on their view of the BlockTree and the state of checkpoints. The Ethereum PoS fork choice rule is LMD GHOST. The LMD GHOST fork choice rule stems from the Greedy Heaviest-Observed Sub-Tree (GHOST) rule \cite{sompolinsky_secure_2015}, which considers only each participant’s most recent vote (Latest Message Driven). 
During an epoch, each validator must make one \emph{block vote} on the block considered as the head of the canonical chain according to its view.

To determine the head of the canonical chain, the fork choice rule does the following:
\begin{enumerate}
    \item Go through the list of validators and check the last block vote of each.
    \item For each block vote, add a weight to each block of the chain that has the block voted as a descendant. The weight added is proportional to the stake of the corresponding validator.
    \item Start from the block of the justified checkpoint with the highest epoch and continue the chain by following the block with the highest weight at each connection. Return the block without any child block. This block is the head of the canonical chain.
\end{enumerate}

The actual implementation is presented in \autoref{algo:GHOST} in \autoref{subsec:code}. This algorithm is similar to the one already presented in \cite{buterin_combining_2020}.
Albeit each \emph{block vote} being for a specific block, the fork choice rule considers all the chains leading to that block. This reflects the fact that a vote for a block is a vote for the chain leading to that block. \autoref{fig:forkChoiceRuleExample} offers an explanation with a visualization of how attestations influence the fork choice rule. At each chain intersection, the fork choice rule favors the chain with the most attestations.

\begin{figure}
    \centering
    \resizebox{\columnwidth}{!}{%
    \input{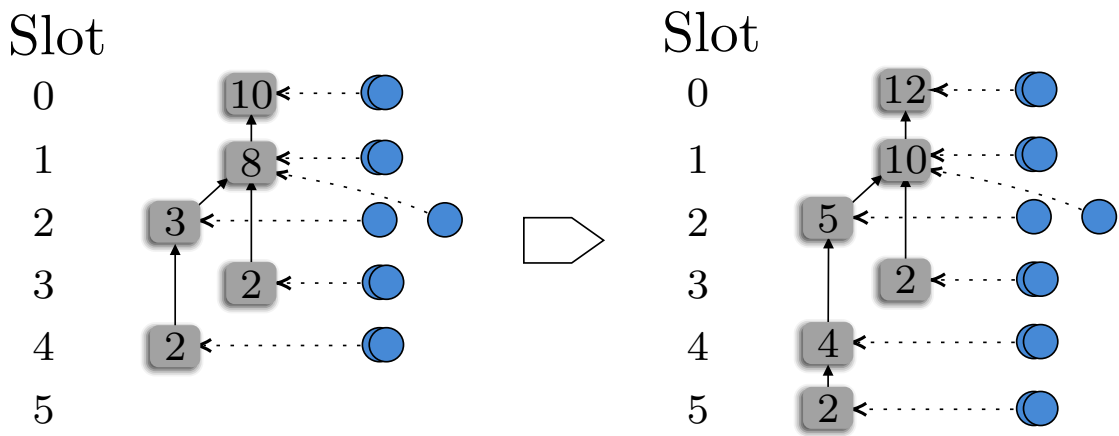}
    }
    \caption{Fork choice rule example observed from a validator $i$'s point of view. We represent block votes with blue circles. 
    Block votes point to specific blocks indicating the block considered as the \emph{headblock} of the candidate chain at the moment of the vote.
    Each block has a number representing the value attributed by the fork choice rule algorithm (cf. \autoref{algo:GHOST}) to determine the candidate chain - we assume for this example that each validator has the same stake of 1. 
    On the left we represent the chain at the end of slot 4, and on the right at the end of slot 5.
    On the left, $i$'s fork choice rule gives the block of slot 4 as $C_i$'s head.
    On the right, the fork choice rule designates the block of slot 5 as the head of the candidate chain.}
    \label{fig:forkChoiceRuleExample}
\end{figure}

\subsubsection{Pseudo-Randomness}\label{ssec:pseudo-randomness}

Ethereum's solution to incorporate randomness in the consensus is called \textit{RANDAO}. RANDAO is a mechanism that creates pseudo-random numbers in a decentralized fashion. It works by aggregating different pseudo-random sources and mixing them.

\paragraph{Seed creation.} Each epoch produces a seed. This seed is created with the help of the block proposers of the said epoch. Each valid block contains a field called \texttt{randao\_reveal}\footnote{See \autoref{subsec:code} for a detailed explanation of its use.}. The seed is the hash of an \texttt{XOR} of all the \texttt{randao\_reveal} values of an epoch plus the epoch number.

Each block's \texttt{randao\_reveal} must be the signature of specific data to prevent manipulation in the seed creation. The data to sign is the current epoch number. Anyone can then check that this signature is from the block proposer and for the correct data.

\paragraph{Seed utilization.}
The algorithm using the seed is called \texttt{compute\_shuffled\_index} (cf. \autoref{algo:computeShuffledIndex}). This algorithm stems from the \textit{swap-or-not} algorithm introduced by \cite{hoang_enciphering_2012}.
\texttt{compute\_shuffled\_index} shuffles the validators list and assigns new roles depending on their shuffled index. This pseudo-random shuffling function is used two times in the Ethereum PoS consensus algorithm: for the proposer selection and the committee selection. The proposer selection is described in \autoref{algo:getProposerIndex} and the committee selection in \autoref{algo:computeCommittee}.

\subsection{Pseudo Code}
\label{subsec:code}

This subsection can be skipped if the reader does not value the understanding on how we dissected the protocol. This step is as crucial as tedious, we outline the protocol based on the code that we formalized to be readable in the form of a pseudo-code.

In this section, we dive into a practical understanding of the mechanism behind the Ethereum PoS protocol. According to the specifications \cite{github_specs} and various implementations (such as Prysm \cite{prysm_code} and Teku \cite{teku_code}), we formalize the main functions of the protocol through pseudo-code for better understanding and analysis purposes.

Each validator $p$ runs an instance of this particular pseudo-code. For instance, when a validator $p$ proposes a block, they broadcast the following message: $\langle PROPOSE, ( slot,$ $\texttt{hash}(headBlock_p) ,$ $content) \rangle$, where $slot$ is the slot at which the proposer proposes the block, the hash of the $headBlock_p$ is the hash of the block considered to be the head of the canonical chain according to the fork choice rule (see \autoref{algo:GHOST}), and $content$ contains data used for pseudo-randomness, among other things that we will not detail here. We instead focus on the consensus protocol.

We describe in the following paragraphs the variables and functions used in the pseudo-code and the goal of these functions.

\paragraph*{Variables.} During the computation, each variable takes a value that is subjective and may depend on the validator. We indicate with $p$ that the value of variables depends on each process. The variable $tree_p$ is considered to be a graph of blocks with each block linked to its predecessor, representing the view of the blockchain (more precisely, $tree_p$ represents the view of all blocks received by the validator since the genesis of the system).
Each $tree_p$ starts with the genesis block. 
$role_p$ corresponds to the different roles a validator can have, which can be none (i.e., for each slot, the validator can be proposer, attester, or have no role). $role_p$ is a list containing the role(s) of the validator for the current slot.
The $slot_p$ is a measure of time. In particular, a slot corresponds to 12 seconds. $slot_p \in \mathbb{N}$. Slot 0 begins at the time of the genesis block and is incremented every 12 seconds. $headBlock_p$ is the head of the canonical chain according to $p$'s local view and the fork choice rule.

A checkpoint $C$ is a pair block-epoch that is used for finalization. $C$ has two attributes, $justified$ and $finalized$, which can be true or false (e.g., if $C$ is only justified, then $C.justified = \texttt{true}$ and $C.finalized = \texttt{false}$). $lastJustifiedCheckpoint_p$ is the \textit{justified} checkpoint with the highest epoch. $currentCheckpoint_p$ is the checkpoint of the current epoch. The list $attestation_p$ is a list of size $n$ (i.e., the total number of validators). This list is updated only to contain the latest messages of validators (of at most one epoch old). 
$CheckpointVote_p$ is a pair of checkpoints, so a pair of pairs, used to make a checkpoint vote. Let us stress the fact that all these variables are local, and at any time, two different validators may have different valuations of those variables.

\paragraph*{Functions.}
We describe the main functions of the protocol succinctly before providing the associated pseudo-code and a more detailed explanation:
\begin{itemize}
    \item \hyperref[algo:validatorMain]{\texttt{validatorMain}} is the primary function of the validator, which launches the execution of all subsidiary functions.
    \item \hyperref[algo:sync]{\texttt{sync}} is a function that runs in parallel with the \texttt{validatorMain} function and ensures the synchronization of the validator. It updates the slot, the role(s), and processes justification and finalization at the end of the epoch and when a new validator joins the system.
    \item \hyperref[algo:GHOST]{\texttt{getHeadBlock}} applies the fork choice rule. This function indicates the head block of the canonical chain.
    \item \hyperref[algo:Casper]{\texttt{justificationFinalization}} is the function that handles the justification and finalization of checkpoints.
\end{itemize}

We depict in \autoref{algo:validatorMain} the main procedure of the validator. This procedure initializes all the values necessary to run a validator. We consider the selection of validators already made to focus on the description of the consensus algorithm itself. The main procedure starts a routine called \texttt{sync} to run in parallel. Then there is an infinite loop that handles the call to an appropriate function when a validator needs to take action for its role(s).

\begin{algorithm}[th]
\caption{Main code for a validator $p$}
\label{algo:validatorMain}
\begin{algorithmic}[1]
\setlength{\lineskip}{3pt}
\Procedure{validatorMain}{\ }
\State $tree_p \gets nil$ \Comment{The tree represents the linked received blocks} 
\State $role_p \gets [\ ]$ \Comment{$role_p $ can be ROLE\_PROPOSER and/or ROLE\_ATTESTER when it is not empty}
\State $slot_p \gets 0$ \Comment{$slot_p \in \mathbb{N}$}
\State $lastJustifiedCheckpoint_p \gets (0, genesisBlock)$ \Comment{A checkpoint is a tuple (epoch, block)}
\State $attestation_p \gets [\ ]$ \Comment{List of latest attestations received for each validator}
\State $validatorIndex_p \gets $ index of the validator \Comment{Each validator has a unique index} 
\State $listValidator \gets [p_0,p_1,\dots,p_{N-1} ]$ \Comment{A list of the validators index}
\State $balances \gets [\ ]$ \Comment{A list of the balances of the validators, their stake}
\State 
\State \textbf{start} \hyperref[algo:sync]{\texttt{sync}($tree_p,$ $slot_p,$  $attestation_p,$  $lastJustifiedCheckpoint_p,$ $role_p,$ 
$balances$)} 
\State
\While{\texttt{true}}
\If{$role_p\neq \emptyset$}
\If{ROLE\_PROPOSER $\in role_p$ } 
\State \hyperref[algo:prepareBlock]{\texttt{prepareBlock()}}
\EndIf
\If{ROLE\_ATTESTER  $\in role_p$} 
\State \hyperref[algo:prepareAttestation]{\texttt{prepareAttestation()}}
\EndIf
\State $role_p \gets [\ ] $
\Else
\State no role assigned \Comment{No action required}
\EndIf
\EndWhile
\EndProcedure
\end{algorithmic}
\end{algorithm}

The roles performed by the validator when acting as proposer or attester are defined in \autoref{algo:prepareBlock} and \autoref{algo:prepareAttestation}, respectively.
The proposer of a block performs the following three tasks:
\begin{enumerate}
    \item Get the head of its canonical chain to have a block to build upon.
    \item Sign a predefined pair to participate in the process of pseudo-randomness.
    \item Broadcast a new block built on top of the head of the canonical chain.
\end{enumerate}

The attestation is composed of three parts: the slot, the block vote, and the checkpoint vote. The validator uses the fork choice rule presented in \autoref{algo:GHOST} to obtain the block chosen for the block vote. \autoref{algo:GHOST} and the one stemming from it, \autoref{algo:weight}, have already been defined in \cite{buterin_combining_2020}. We restate them here for the sake of completeness. For the checkpoint vote, an honest validator should always vote for the current epoch as the target and take the justified checkpoint with the highest epoch (i.e., $lastJustifiedCheckpoint$) as the source.

In order to broadcast this attestation, the attester must wait for one of two things: either a block has been proposed for this slot, or $1/3$ of the slot (i.e., $4$ seconds) has elapsed. This is ensured by the function \texttt{waitForBlockOrOneThird}.

\begin{algorithm}[th]
\caption{broadcast block}
\label{algo:prepareBlock}
\begin{algorithmic}[1]
\setlength{\lineskip}{3pt}
\Procedure{prepareBlock}{ } 
    \State $headBlock_p \gets$ 
    \hyperref[algo:GHOST]{\texttt{getHeadBlock()}} 
    \State $randaoReveal \gets $ \texttt{sign}( \texttt{epochOf($slot$)})
    \State \textbf{broadcast} $\langle PROPOSE, ( slot, \texttt{hash}(headBlock_p), randaoReveal_p,$ $ content) \rangle$
\EndProcedure
\end{algorithmic}
\end{algorithm}

\begin{algorithm}[th]
\caption{Broadcast Attestation}
\label{algo:prepareAttestation}
\begin{algorithmic}[1]
\setlength{\lineskip}{3pt}
\Procedure{prepareAttestation}{ } 
    \State \texttt{waitForBlockOrOneThird()}  \Comment{wait for a new block in this slot or $\frac{1}{3}$ of the slot
    }
    \State $headBlock_p \gets$ \hyperref[algo:GHOST]{\texttt{getHeadBlock()}}
    \State $currentCheckpoint_p \gets $ (first block of the epoch, \texttt{epochOf}($slot$) )
    \State $CheckpointVote_p \gets \big(lastJustifiedCheckpoint_p, currentCheckpoint_p \big) $
    \State \textbf{broadcast} $\langle ATTEST, (slot_p, \underbrace{\texttt{hash}(headBlock_p)}_\textrm{block vote}, \underbrace{CheckpointVote_p}_\textrm{checkpoint vote}) \rangle$
\EndProcedure
\end{algorithmic}
\end{algorithm}

The synchronization of the validator $p$ is handled by the function \texttt{sync} described in \autoref{algo:sync}. This algorithm allows the validator to update its view of the blockchain, particularly the current slot, the list of attestations, the last justified checkpoint, the validator's role, and the balances of all validators. To determine its role(s), the validator verifies the index of the designated validator for the current slot and the set of indexes forming the committee of the current slot.

In more detail, two conditions assign a role to a validator for the current slot. The first condition calls \autoref{algo:getProposerIndex} and assigns the validator $p$ the role of proposer if its index matches that of the current proposer. The second condition checks whether $p$ belongs to the committee of the current slot (see \autoref{algo:computeCommittee}). The roles of proposer and attester are entirely distinct, i.e., the proposer of a slot is not necessarily an attester.

The synchronization function also starts two other routines, \texttt{syncBlock} and \texttt{syncAttestation}, corresponding to \autoref{algo:syncBlock} and \autoref{algo:syncAttestation}, respectively. These routines are used to handle the broadcasts from proposers and attesters. In both functions, upon receiving a block or an attestation, the validator $p$ verifies its validity using the \texttt{isValid} function. It is important to note that upon receiving a block, a validator can update the last justified checkpoint only if the current epoch has not started more than $8$ slots ago. This particular condition is what the patch has introduced to prevent a liveness attack (see \autoref{subsec:probabilisticLiveness}).

\begin{algorithm}[th]
\caption{Sync}
\label{algo:sync}
\begin{algorithmic}[1]
\setlength{\lineskip}{3pt}
\Procedure{sync}{$tree, slot, attestation,$ $role,$ $lastJustifiedCheckpoint,$} 
    \State \textbf{start} \hyperref[algo:syncBlock]{\texttt{syncBlock($slot, tree$)}}
    \State \textbf{start} \hyperref[algo:syncAttestation]{\texttt{syncAttestation($attestation$)}}
    \Repeat
        \State $previousSlot \gets slot$
        \State $slot \gets \lfloor$ time in seconds since genesis block / 12 $\rfloor$
        \If{$previousSlot \neq slot$} \Comment{If we start a new slot}
            \State \textit{roleSlotDone} $ \gets $ false
            \If{$validatorIndex_p$ = \hyperref[algo:getProposerIndex]{\texttt{getProposerIndex}}(\hyperref[algo:getSeed]{\texttt{getSeed}}(current epoch), $slot$) }
            \State \texttt{append} ROLE\_PROPOSER to $role_p$
            \EndIf
            \If{$validatorIndex_p \in$ \hyperref[algo:computeCommittee]{\texttt{computeCommittee}}(\hyperref[algo:getSeed]{\texttt{getSeed}}(current epoch), $slot$)}
            \State \texttt{append} ROLE\_ATTESTER to $role_p$
            \EndIf
        \EndIf
        \If{$slot \pmod{32} = 0$} \Comment{First slot of an epoch}
            \State \hyperref[algo:Casper]{\texttt{jutificationFinalization}($tree,$ $lastJustifiedCheckpoint$)}
        \EndIf
    \Until{{validator exit}}
\EndProcedure
\end{algorithmic}
\end{algorithm}

\begin{algorithm}[th]
\caption{Sync Block}
\label{algo:syncBlock}
\begin{algorithmic}[1]
\setlength{\lineskip}{3pt}
\Procedure{syncBlock}{$slot, tree$}
    \Upon{$\langle PROPOSE, ( slot_i, \texttt{hash}(headBlock_i) , randaoReveal_i,$ $ content_i) \rangle$}{validator $i$}
        \State $block \gets \langle PROPOSE, ( slot_i, \texttt{hash}(headBlock_i) , $ $ randaoReveal_i,$ $ content_i) \rangle$
        \If{\texttt{isValid}($block$)} 
        \If{$slot \pmod{32} \leq 8$ }
            \State update justified checkpoint if necessary
        \EndIf
        \EndIf
    \EndUpon
\EndProcedure
\end{algorithmic}
\end{algorithm}

\begin{algorithm}[th]
\caption{Sync Attestation}
\label{algo:syncAttestation}
\begin{algorithmic}[1]
\setlength{\lineskip}{3pt}
\Procedure{syncAttestation}{$attestation$} 
    \Upon{$\langle ATTEST, (slot_i, headBlock_i, checkpointEdge_i) \rangle$}{validator $i$}
        \State $attestation_i \gets \langle ATTEST, (slot_i, headBlock_i, checkpointEdge_i) \rangle$
        \If{\texttt{isValid}($attestation_i$)}
            \State $attestation[i] \gets attestation_i$ 
        \EndIf
    \EndUpon
\EndProcedure
\end{algorithmic}
\end{algorithm}

\begin{algorithm}[th]
\caption{Get Head Block}
\label{algo:GHOST}
\begin{algorithmic}[1]
\setlength{\lineskip}{3pt}
\Procedure{getHeadBlock}{ }
\State $block \gets$ block of the justified checkpoint with the highest epoch
\While{$block$ has at least one child}
\State $block \gets \underset{b'\text{ child of }block}{\argmax}$ \hyperref[algo:weight]{\texttt{weight}}($tree, Attestation, b'$)  
\State (ties are broken by hash of the block header)
\EndWhile
\State \Return $block$
\EndProcedure
\end{algorithmic}
\end{algorithm}

\begin{algorithm}[th]
\caption{Weight}
\label{algo:weight}
\begin{algorithmic}[1]
\setlength{\lineskip}{3pt}
\Procedure{weight}{$tree, Attestation,block$}
\State $w \gets 0$
\For{every validator $v_i$}
\If{$\exists a \in \ Attestation$ an attestation of $v_i$ for $block$ or a descendant of $block$}
\State $w \gets w + $ stake of $v_i$
\EndIf 
\EndFor
\State \Return $w$
\EndProcedure
\end{algorithmic}
\end{algorithm}

\autoref{algo:Casper} can be considered the most intricate. This algorithm is responsible for justifying or finalizing the checkpoints at the end of each epoch. To do so, it counts the number of checkpoint votes with the same source and target. If this number corresponds to more than 2/3 of the stake of all validators, then the target is considered justified for the validator running this algorithm. The last four conditions concern finalization. They verify among the last four checkpoints which one fulfills the conditions to become finalized. The conditions to become finalized are formally described in \autoref{subsec:overview} and can be summarized as follows: the checkpoint must be the source of a supermajority link, and all the checkpoints between the source and target, inclusive, must be justified.

\begin{algorithm}[th]
\caption{Justification and Finalization}
\label{algo:Casper}
\begin{algorithmic}[1]
\setlength{\lineskip}{3pt}
\Procedure{jutificationFinalization}{$tree,$ $ lastJustifiedCheckpoint$}
\State $source \gets lastJustifiedCheckpoint$ 
\State $target \gets $ the current checkpoint
\State $nbCheckpointVote \gets $ \texttt{countMatchingCheckpointVote}($source,$ $target$)
\\
\hskip\algorithmicindent {\color{gray} $\triangleright$ \emph{justification process}:}
\If{ $nbCheckpointVote \geq \frac{2}{3} *$ total balance of validators }
\State $target.justified \gets $ \texttt{true}
\State $lastJustifiedCheckpoint \gets target$
\EndIf
\\
\hskip\algorithmicindent {\color{gray} $\triangleright$ \emph{finalization process}:}
\State $A,B,C,D \gets $ the last 4 checkpoints \Comment{With $D$ being the current checkpoint.}
\If{$A.justified$ $\land$ $B.justified$ $\land$ ($A \xrightarrow{\texttt{J}} C$)}
\State $A.finalized \gets $ \texttt{true} \Comment{Finalization of $A$}
\EndIf 
\If{$B.justified$ $\land$ ($B \xrightarrow{\texttt{J}} C$)}
\State $B.finalized \gets $ \texttt{true} \Comment{Finalization of $B$}
\EndIf 
\If{$B.justified$ $\land$ $C.justified$ $\land$ ($B \xrightarrow{\texttt{J}} D$)}
\State $B.finalized \gets $ \texttt{true} \Comment{Finalization of $B$}
\EndIf 
\If{$C.justified$ $\land$ ($C \xrightarrow{\texttt{J}} D$)} 
\State  $C.finalized \gets $ \texttt{true} \Comment{Finalization of $C$}
\EndIf 
\EndProcedure
\end{algorithmic}
\end{algorithm}


\begin{algorithm}
\caption{Get randao mix}
\label{algo:getRandaoMix}
\begin{algorithmic}[1]
\setlength{\lineskip}{3pt}
\Procedure{getRandaoMix}{$epoch$}
\State $mix \gets 0$
\State $headBlock \gets $ \hyperref[algo:GHOST]{\texttt{getHeadBlock}}()
\For{\textbf{each} $block$ parent of $headBlock$ and belonging to $epoch$}
\State $mix \gets mix \oplus \texttt{hash}(block.$randaoReveal)
\Comment{$\oplus$ is a bit-wise \texttt{XOR} operator}
\EndFor
\State \Return $mix$
\EndProcedure
\end{algorithmic}
\end{algorithm}

\begin{algorithm}
\caption{Get seed}
\label{algo:getSeed}
\begin{algorithmic}[1]
\setlength{\lineskip}{3pt}
\Procedure{getSeed}{$epoch$}
\State $mix \gets$ \hyperref[algo:getRandaoMix]{\texttt{getRandaoMix}}($epoch-2$) \Comment{The seed of an epoch $i$ is based on the randao mix of epoch $i-2$ }
\State \Return \texttt{hash}($epoch + mix$)
\EndProcedure
\end{algorithmic}
\end{algorithm}

The pseudo-randomness requires a different seed for each epoch to yield different results. This is ensured by hashing the RANDAO mix and the epoch number, as shown in \autoref{algo:getSeed}. Adding the epoch number is helpful if no block is proposed during an entire epoch. This corner case would always result in the same seed if it were not for the epoch number.

The RANDAO mix is computed in \autoref{algo:getRandaoMix}. The computation of the RANDAO mix for a given epoch consists of \textit{XORing} all the \textit{randaoReveal} values of the blocks in that particular epoch. We consider only the blocks of that particular epoch that belong to the canonical chain.

The RANDAO mix of epoch $e-2$ determines the role of validators in epoch $e$. Hence, with \autoref{algo:computeCommittee}, as soon as epoch $e-2$ is over, validators can know to which committee they belong at epoch $e$. \texttt{computeCommittee} (\autoref{algo:computeCommittee}) is the function that, given a seed and an epoch, returns the list of validator indices corresponding to the committee for the specified slot. The number of validators in each committee\footnote{In the actual implementation, committees have a maximum size of 2048 \cite{github_specs}.} is computed to be less than $N/32$ (with $n$ being the total number of validators). Then, using the shuffled index computed with \autoref{algo:computeShuffledIndex}, a committee of the given size is drawn according to the slot in question. All committee validators will have to perform the role of attester during this slot.

Since the balance can change until the previous epoch, block proposers are known at the end of epoch $e-1$ for epoch $e$. \autoref{algo:getProposerIndex} handles the selection of a proposer for a designated slot. It starts by creating a seed specifically for the slot in question. Then, a loop starts with a pseudo-random selection of the validator's index. The loop stops only when a validator meets the condition criteria. This condition is equivalent to being selected with a probability depending on the balance. Thus, the validator with index \textit{proposerIndex} is selected with probability $\frac{\textit{effectiveBalance}}{32}$, with \textit{effectiveBalance} being the stake of \textit{proposerIndex} capped at $32$, i.e., $\min(\textit{balance}, 32)$.

Both Algorithm \ref{algo:getProposerIndex} and \ref{algo:computeCommittee} use \autoref{algo:computeShuffledIndex} to imbue randomness in the proposer and committee selection. As mentioned in \autoref{sec:EthPoSprotocol}, \autoref{algo:computeShuffledIndex} stems from the algorithm \textit{swap-or-not} \cite{hoang_enciphering_2012}. Its name helps us understand the principle behind the algorithm: select a validator and its opposite (based on a pivot) and swap them or not. The selection of the validator and the swap depend on the value of a hash. An essential aspect of this algorithm is that it can get the index of validators in the shuffled list without having to compute the shuffling of the whole list of validators. This reduces unnecessary computation.

\begin{algorithm}
\caption{Compute shuffled index}
\label{algo:computeShuffledIndex}
\begin{algorithmic}[1]
\setlength{\lineskip}{3pt}
\Procedure{computeShuffledIndex}{$index, seed, nbValidators$}
\For{$i = 0 $ \textbf{to} $90$}
\State $pivot \gets $ \texttt{hash}($seed + i$) $(mod\;$ \textit{nbValidators}$)$
\State $ flip \gets pivot + nbValidators - index \; (mod\;$ \textit{nbValidators}$)$
\State $position \gets \max(index,\; flip)$
\State $bit \gets $ \texttt{hash}($seed+i+position$)$\pmod{2}$
\If{$bit = 0 \; (mod\;$ \textit{nbValidators}$)$}
\State $index \gets flip$
\EndIf
\EndFor
\State \Return $index$
\EndProcedure
\end{algorithmic}
\end{algorithm}

\begin{algorithm}
\caption{Get proposer index}
\label{algo:getProposerIndex}
\begin{algorithmic}[1]
\setlength{\lineskip}{3pt}
\Procedure{getProposerIndex}{$seed$, $slot$}
\State MAX\_RANDOM\_BYTE $\gets 2^8-1$
\State $i \gets 0$
\State $proposerSeed \gets $ \texttt{hash}($seed$+$slot$)
\State $nbValidators \gets $ \texttt{length}($listValidator$)
\While{true}
\State $proposerIndex \gets  listValidator$\hyperref[algo:computeShuffledIndex]{[\texttt{computeShuffledIndex}($i, seed, nbValidators$)]}
\State $randomByte \gets$ first byte of \texttt{hash}($proposerSeed + i \pmod{nbValidators}$) 
\State \textit{effectiveBalance} $\gets listValidators[proposerIndex]$.effectiveBalance
\If{\textit{effectiveBalance} $*$ MAX\_RANDOM\_BYTE $\geq $ MAX\_EFFECTIVE\_BALANCE $* randomByte$  }
\State \Return $proposerIndex$
\EndIf
\State $i \gets i + 1$
\EndWhile
\EndProcedure
\end{algorithmic}
\end{algorithm}

\begin{algorithm}
\caption{Compute Committee}
\label{algo:computeCommittee}
\begin{algorithmic}[1]
\setlength{\lineskip}{3pt}
\Procedure{computeCommittee}{$seed, slot$}
\State $committee \gets [\ ]$
\State \textit{nbValidatorByCommittee} $\gets \lceil \texttt{lenght}(listValidator) / 32 \rceil$
\For{$i = (slot \pmod{32}) *$\textit{nbValidatorByCommittee}  \textbf{to} $(slot+1 \pmod{32}) *$\textit{nbValidatorByCommittee} $-1$ }
\State $committee$.\texttt{append}($listValidator$[\hyperref[algo:computeShuffledIndex]{\texttt{computeShuffledIndex}($i, seed, nbValidators$)}])
\EndFor
\State \Return $committee$
\EndProcedure
\end{algorithmic}
\end{algorithm}

\section{Robustness Analysis}\label{sec:correctnessAnalysis}
We now have formalized the protocol and the blockchain properties necessary for our analysis. We will start by analyzing if the protocol is safe, and then if is live.

\subsection{Safety}\label{subsec:safety}

In order to prove the safety of the protocol, we begin by presenting lemmas concerning the justification of checkpoints. The first lemma rules out the possibility of two different justified checkpoints having the same epoch. New validators that want to join the set of validators must send the amount they wish to stake to a specific smart contract\footnote{Currently, 32 ETH is needed to become a validator.}. This transaction triggers the process for a validator to join the set of validators. The last step required for the activation of a validator (allowing it to send attestations and propose blocks) requires that the block adding the validator to the validator set gets finalized\footnote{The exact process involves placing the validator in the activation queue to be finalized. See more at \href{https://github.com/ethereum/consensus-specs/blob/80ba16283c9447db8aa04eeaf4a3940b56480758/specs/phase0/beacon-chain.md}{https://github.com/ethereum/consensus-specs/blob/dev/specs/phase0/beacon-chain.md}.}.
This means that the set of validators is fixed between two finalized checkpoints.

\begin{lemma} \label{firstLemma}
If checkpoints $C$ and $C'$ of the same epoch $e$ are justified, it must necessarily be that $C = C'$.
\end{lemma}

\begin{proof}
By hypothesis, we know that Byzantine validators are at most $f < n/3$. For the sake of contradiction, let us assume that $C$ and $C'$ are different checkpoints. Let $V$ be the set of at least $2n/3 - f$ honest validators that cast a checkpoint vote for checkpoint $C$ in epoch $e$, and $V'$ be the set of at least $2n/3 - f$ honest validators that cast a checkpoint vote for checkpoint $C'$ in epoch $e$. The intersection of the two sets of honest validators is $|V \cap V'| \geq (2n/3 - f) + (2n/3 - f) - (n - f) = (n/3 - f) > 0$. $|V \cap V'| > 0$ implies that at least one honest validator voted for both checkpoint $C$ and checkpoint $C'$ in epoch $e$. This is a contradiction since, according to the protocol specification\footnote{This is specified in the specs \url{https://github.com/ethereum/consensus-specs/blob/dev/specs/phase0/validator.md\#attester-slashing}, and implemented in the actual client Prysm \url{https://github.com/prysmaticlabs/prysm/blob/0fd52539153e32cfbd0a27ee51f253f8f6bb71c4/validator/client/attest.go\#L140}. This corresponds to the only attestation done by an honest validator during an epoch, see \autoref{algo:sync}.}, an honest process signs at most one unique block per epoch. Therefore, $C = C'$. This proves there cannot be more than one justified checkpoint per epoch.
\end{proof}

The following lemma explains why the finalization of a checkpoint necessarily means that a checkpoint cannot be justified on a different chain afterward.

\begin{lemma}\label{secondLemma}
If a checkpoint $C$ of epoch $e$ is finalized on chain $c$, and a checkpoint $C'$ of epoch $e'$ is justified on chain $c'$ with $e' > e$, it necessarily means that $c$ and $c'$ have a common prefix until epoch $e$. 
\end{lemma}

\begin{proof}
$C'$ being justified on chain $c'$ means that at least $2n/3 - f$ honest validators must have cast a checkpoint vote with $C'$ as the checkpoint target for epoch $e'$.

For the sake of contradiction, let us say that $c$ and $c'$ have a common prefix until epoch $e-1$ at most. For a checkpoint to be justified on chain $c'$ at an epoch strictly superior to $e$, it implies that a set $V'$ of at least $2n/3 - f$ honest validators must have cast a checkpoint vote with a checkpoint target on chain $c'$ and a checkpoint source with an epoch less than $e-1$.

Checkpoint $C$ of epoch $e$ being finalized on chain $c$, we have two possibilities. Either the checkpoint at epoch $e+1$ on chain $c$ has been justified with checkpoint $C$ as the source, or the checkpoint at epoch $e+2$ on chain $c$ has been justified with checkpoint $C$ as the source, and the checkpoint at epoch $e+1$ is justified. Either way, a justification occurred on chain $c$ with checkpoint $C$ as the source, and no justification occurred on a different chain before its finalization.

Hence, we know that a set $V$ of at least $2n/3 - f$ honest validators have cast a checkpoint vote with $C$ as the checkpoint source before a justification on any other chain.

Seeing that $|V \cap V'| > 0$, at least one honest validator has cast a checkpoint vote with $C$ as the checkpoint source and then a checkpoint vote with a checkpoint source of at most epoch $e-1$ and a target on chain $c'$.

Therefore, at least one honest validator has cast a checkpoint vote with a checkpoint source from an epoch less than $e-1$ after seeing checkpoint $C$ at epoch $e$ justified. However, the fork choice rule of the protocol (cf. \autoref{algo:GHOST}) requires honest validators to vote on the chain with the highest justified checkpoint. This contradiction proves the lemma.
\end{proof}

We saw with \autoref{firstLemma} that two checkpoints of the same epoch could not be justified, hence finalized. We then showed with \autoref{secondLemma} that after a finalization on one chain, no checkpoints could become justified on any other chain. These are the conditions required to have safety, as we prove now.

\begin{theorem}[Safety]
There cannot be two finalized checkpoints on different chains in Ethereum PoS.
\end{theorem}

\begin{proof}
Thanks to \autoref{firstLemma}, we know that two different checkpoints $C$ and $C'$ of the same epoch cannot be justified, hence finalized. 

For the sake of contradiction, let us assume that two checkpoints $C$ and $C'$ are finalized on different chains $c$ and $c'$ at epochs $e$ and $e'$, respectively. We assume without loss of generality that $e < e'$. 
$C$ being finalized, we know thanks to \autoref{secondLemma} that $C'$ cannot be justified on a different chain $c'$, let alone be finalized.
\end{proof}

The blockchain preserves the property of safety at all times. The Ethereum PoS is safe.

\subsection{Probabilistic Liveness}\label{subsec:probabilisticLiveness}

Spoiler alert: the protocol is not guaranteed to be live; rather, it is probabilistically live. This means that as time goes on, the probability of it being live approaches 1. However, the probability that it is not live is not zero, although it remains very small. To prove this point, we will explain an attack that targets the protocol's liveness.

In order to explain this attack which is by no means simple we start by describing a simpler liveness attack called the \emph{bouncing attack}. This attack delays finality in a partially synchronous network after \texttt{GST}. Previous works also exhibit liveness attacks against the protocol using the intertwining of the fork choice rule and the finality gadget \cite{nakamura_analysis_2019, neu_ebb_2021}. To prevent this attack, the protocol now contains a "patch" \cite{pullRequest_bouncing_2022} suggested on the Ethereum research forum \cite{nakamura_prevention_2019}. We show that the implemented patch is insufficient, and this attack is still possible if certain conditions are met. This is a probabilistic liveness attack against the Ethereum Proof-of-Stake protocol. Our attack can happen with less than 1/3 of Byzantine validators, as discussed in \autoref{subsubsec:probabilisticBouncingAttack}. We also consider the adversary to be static because Byzantine validators are chosen before the computation.

\subsubsection{Bouncing Attack}
\label{subsubsec:bouncingAttack}

\begin{figure}
    \centering
    \begin{tikzpicture}[x=0.75pt,y=0.75pt,yscale=-1,xscale=1]

\draw  [draw opacity=0][fill={rgb, 255:red, 126; green, 211; blue, 33 }  ,fill opacity=1 ] (221.29,15.74) -- (221.29,29) -- (217.29,29) -- (217.29,15.74) -- cycle ;
\draw  [fill={rgb, 255:red, 132; green, 255; blue, 0 }  ,fill opacity=1 ] (52.64,96.08) -- (46.53,106.67) -- (34.31,106.67) -- (28.2,96.08) -- (34.31,85.5) -- (46.53,85.5) -- cycle ;
\draw   (50.05,96.27) -- (45.14,104.77) -- (35.32,104.77) -- (30.41,96.27) -- (35.32,87.76) -- (45.14,87.76) -- cycle ;
\draw [color={rgb, 255:red, 251; green, 193; blue, 96 }  ,draw opacity=1 ][line width=1.5]  [dash pattern={on 1.69pt off 2.76pt}]  (220,108) .. controls (220.47,103.35) and (220.07,107.78) .. (220.01,91.94) ;
\draw [shift={(220,88)}, rotate = 90] [fill={rgb, 255:red, 251; green, 193; blue, 96 }  ,fill opacity=1 ][line width=0.08]  [draw opacity=0] (10.92,-2.73) -- (0,0) -- (10.92,2.73) -- cycle    ;
\draw   (52.64,56.8) -- (46.53,67.38) -- (34.31,67.38) -- (28.2,56.8) -- (34.31,46.22) -- (46.53,46.22) -- cycle ;
\draw   (50.19,56.68) -- (45.28,65.19) -- (35.46,65.19) -- (30.56,56.68) -- (35.46,48.18) -- (45.28,48.18) -- cycle ;
\draw   (52.6,20.3) -- (46.49,30.88) -- (34.26,30.88) -- (28.15,20.3) -- (34.26,9.72) -- (46.49,9.72) -- cycle ;
\draw  [draw opacity=0][fill={rgb, 255:red, 208; green, 2; blue, 27 }  ,fill opacity=1 ] (222.92,52.24) -- (222.92,61.36) -- (218.92,61.36) -- (218.92,52.24) -- cycle ;
\draw  [draw opacity=0][fill={rgb, 255:red, 117; green, 214; blue, 11 }  ,fill opacity=1 ] (222.92,61.36) -- (221.98,61.36) -- (218.92,56.57) -- (218.92,54.25) -- (222.92,60.51) -- (222.92,61.36) -- cycle ;
\draw  [draw opacity=0][fill={rgb, 255:red, 117; green, 214; blue, 11 }  ,fill opacity=1 ] (218.92,61.36) -- (218.92,59.03) -- (220.41,61.36) -- (218.92,61.36) -- cycle ;
\draw  [draw opacity=0][fill={rgb, 255:red, 117; green, 214; blue, 11 }  ,fill opacity=1 ] (222.92,58.51) -- (218.92,52.24) -- (220.41,52.24) -- (222.92,56.19) -- (222.92,58.51) -- cycle ;
\draw  [draw opacity=0][fill={rgb, 255:red, 117; green, 214; blue, 11 }  ,fill opacity=1 ] (222.92,52.24) -- (222.92,54.3) -- (221.61,52.24) -- (222.92,52.24) -- cycle ;

\draw   (21.5,3) -- (455,3) -- (455,121) -- (21.5,121) -- cycle ;

\draw (68.42,20.3) node [anchor=west] [inner sep=0.75pt]   [align=left] {{\footnotesize checkpoint}};
\draw (68.42,56.8) node [anchor=west] [inner sep=0.75pt]   [align=left] {{\footnotesize justified checkpoint}};
\draw (68.42,96.08) node [anchor=west] [inner sep=0.75pt]   [align=left] {{\footnotesize finalizied checkpoint}};
\draw (231.42,20.3) node [anchor=west] [inner sep=0.75pt]  [font=\footnotesize] [align=left] {attestation from honest validateur};
\draw (232.92,56.8) node [anchor=west] [inner sep=0.75pt]  [font=\footnotesize] [align=left] {attestation from Byantine validateur};
\draw (234.92,98.77) node [anchor=west] [inner sep=0.75pt]  [font=\footnotesize] [align=left] {designate which checkpoint is the \\target checkpoint for an attestation};

\end{tikzpicture}
    \caption{This figure serves as a summary of the signification of the main diagrams of other figures.}
    \label{fig:diagramSummary}
\end{figure}
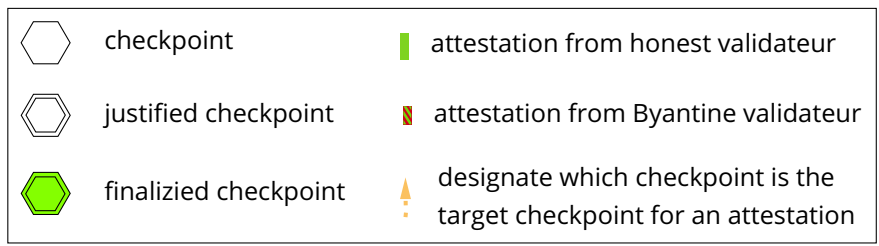

The \emph{Bouncing Attack} \cite{nakamura_analysis_2019} describes a liveness attack where the suffix of the chain changes repetitively between two canonical chains, thus preventing the chain from finalizing any checkpoint. The Bouncing Attack exploits the fact that the canonical chains should start from the justified checkpoint with the highest epoch. It is possible for Byzantine validators to divide honest validators' opinions by justifying a new checkpoint once some honest validators have already cast their vote (made an attestation) during the asynchronous period before \texttt{GST}.

The bouncing attack becomes possible once there is a \emph{justifiable} checkpoint in a different branch from the one designated by the fork choice rule with a higher epoch than the current highest justified checkpoint. A \emph{justifiable checkpoint} is a checkpoint that can become justified only by adding the checkpoint votes of Byzantine validators. If this setup occurs, the Byzantine validators could make honest validators start voting for a different checkpoint on a different chain, leaving a justifiable checkpoint again for them to repeat their attack and thus making validators \emph{bounce} between two different chains and not finalizing any checkpoint. Hence the name Bouncing attack.

Let us illustrate the attack with a concrete case. In \autoref{fig:simplifiedBouncingAttack}, we show an oversimplified case with only 10 validators, among which 3 are Byzantine. To occur, the attack needs to have a justifiable checkpoint with a higher epoch than the last justified checkpoint. We reach this situation before \texttt{GST}, which is presented in the left part of the figure. After reaching \texttt{GST}, Byzantine validators wait for honest validators to make a new checkpoint justifiable. When a new checkpoint is justifiable, the Byzantine validators cast their votes to justify another checkpoint, as shown in the right part of the figure. This will lead honest validators to vote for the left branch, thus reaching a situation similar to the first step, allowing the bouncing attack to continue. The repetition of this behavior is the bouncing attack. We emphasize this example in more detail in \autoref{fig:complexBouncingAttack} by detailing the sequence of votes allowing a "bounce" to occur and leaving a justifiable checkpoint on the other branch.

\begin{figure}
    \centering
    \input{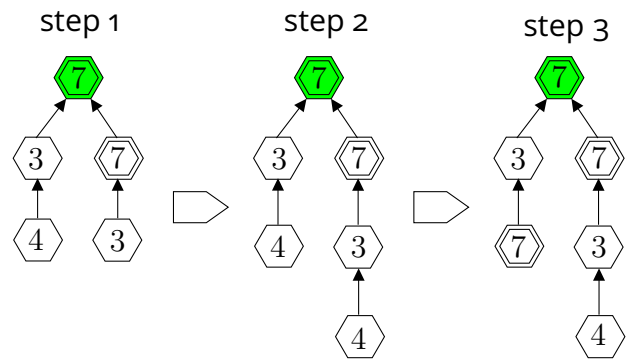}
    \caption{A bouncing attack presented in 3 steps. We have 10 validators, of which 3 are Byzantines. 
    The number inside each hexagon corresponds to the number of validators who made a checkpoint vote with this checkpoint as target.
    \textbf{1st step:} We start in a situation where there is a fork. A checkpoint is justified on one of the chains and a checkpoint of a higher epoch is \textit{justifiable} on the other. We are at the end of the third epoch in which honest validators have divided their vote on each side. \textbf{2nd step:} We have reached \texttt{GST} at the beginning of the fourth epoch and 4 honest validators have already voted (rightfully so). \textbf{3rd step:} Here is the moment Byzantine validators take action and release their checkpoint vote for the concurrent chain, thus justifying the previously forsaken checkpoint and thereby changing the highest justifying checkpoint. By repeating this process, the bouncing attack can continue indefinitely.}
    \label{fig:simplifiedBouncingAttack}
\end{figure}

\begin{figure}
    \centering
    \input{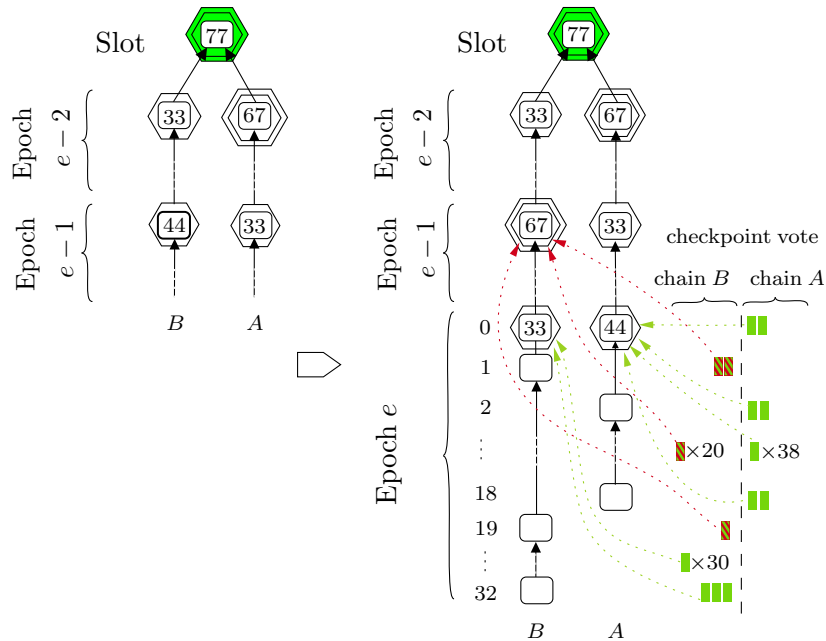}
    \caption{This figure presents a detailed version of the bouncing attack. In this example, we have a total of 100 validators, of which 23 are Byzantines.
    A block in a checkpoint corresponds to the block associated with that checkpoint.
    The number inside each hexagon (hovering a block) corresponds to the number of validators who made a checkpoint vote with this checkpoint as target. We distinguish between two sorts of checkpoint votes, the Byzantine ones, which are bi-color rectangles, 
    and the honest ones, which are uni-color rectangles.
    We compile the 3 steps of \autoref{fig:simplifiedBouncingAttack} in 2 with more information on how justification's turning point is accomplished because of the Byzantine agents. \textbf{First step:} We begin from a situation where epoch $e-1$ just ended and we now reach \texttt{GST}. Notice that the canonical chain is chain $A$ because the checkpoint with the highest epoch is on chain $A$ but not chain $B$. \textbf{Second step:} In this step, the checkpoint vote released during epoch $e$ can change the last justified checkpoint to change the canonical chain for chain $A$ to chain $B$. Byzantine validators released their checkpoint vote from the previous epoch during epoch $e$. They send their last checkpoint vote at slot 23 once the checkpoint of epoch $e$ on chain $A$ has reached 44, thus becoming justifiable (i.e., not yet justified but with enough votes so that Byzantine validators can justify it). This triggers the canonical chain to change from chain $A$ to chain $B$ starting the \emph{bounce}.}
    \label{fig:complexBouncingAttack}
\end{figure}

\subsection{Implemented Patch} \label{subsec:implementedPatch}
The explanation of the patch is described for the first time on the Ethereum research forum \cite{nakamura_prevention_2019}. The solution found to mitigate the bouncing attack is to engrave in the protocol the fact that validators cannot change their minds regarding justified checkpoints after a part of the epoch has passed.

The goal of the proposed solution is to prevent the possibility of justifiable checkpoints being left out by honest validators. To prevent honest validators from leaving a justifiable checkpoint, the patch must stop validators from changing their view of checkpoints before more than 1/3 of validators have cast their checkpoint vote. This condition stems from the fact that we reckon the proportion of Byzantine validators to be at most $1/3 - \epsilon$.
To apply this condition, the patch designates a number of slots after which honest validators cannot change their view of checkpoints. Since validators are scattered equally among the different slots to cast their vote (in attestations) within a specific time frame, stopping validators from changing their view after a certain number of slots is equivalent to stopping them from changing their view after a certain proportion of validators have voted. This does appear to be a solution to prevent Byzantine validators from influencing honest validators into forsaking a checkpoint that is now \emph{justifiable} for them.

To enforce this behavior, called the "fixation of view," the protocol has a constant $j$ called \texttt{SAFE\_SLOTS\_TO\_UPDATE\_JUSTIFIED} in the code (cf. \autoref{algo:syncBlock} in \autoref{subsec:code}). This constant is the number of slots\footnote{At the time of writing this manuscript, $j=8$ \cite{github_specs}.} during which validators can change their view of the justified checkpoints.
The patch introducing this constant $j$ mentions a possible attack called the \emph{splitting attack}. As they point out, the splitting attack relies on a "last minute delivery" strategy whereby releasing a message late enough causes some validators to consider it too late while others do not. This could split the validators into two different chains, unable to reconcile their views before the end of the epoch. After the beginning of the next epoch views can be reconciled during $j$ slots however the split can occur once again by another last minute delivery.
They consider the assumption that attackers can send a message at the right time to split honest validators too strong. In \autoref{subsubsec:probabilisticBouncingAttack}, we present a new attack inspired by the splitting attack with more realistic assumptions.

\subsubsection{Probabilistic Bouncing attack - why the patch is not enough}
\label{subsubsec:probabilisticBouncingAttack}

In this part, we present our novel attack against the protocol of Ethereum Proof-of-Stake. The attack is visually explained in \autoref{fig:probabilisticBouncingAttack3steps}.

\begin{figure}
    \centering
    \resizebox{\linewidth}{!}{
    \input{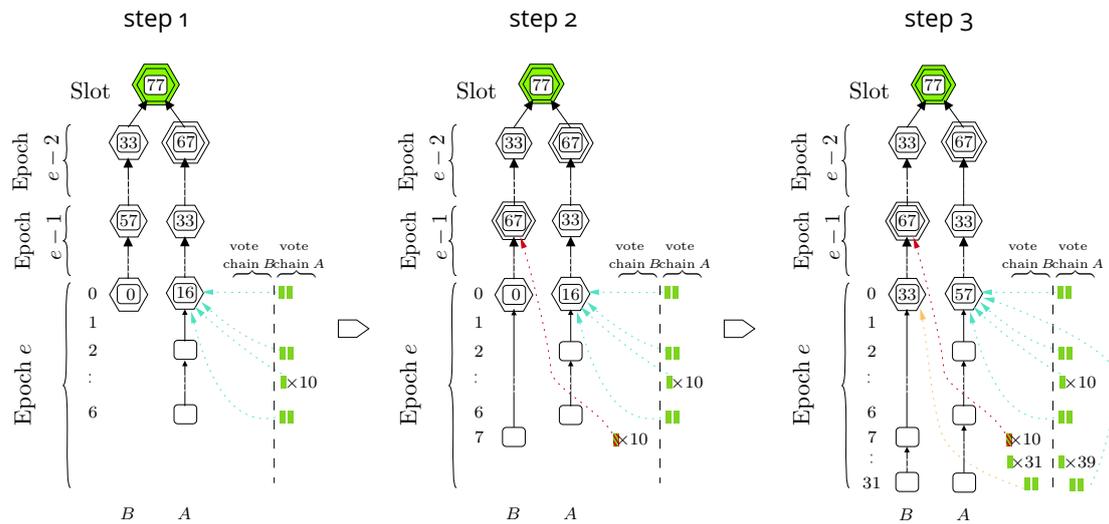}
    }
    \caption{This figure presents the probabilistic bouncing attack. In this example, we consider 100 validators, of which 10 are Byzantine. A block in a checkpoint corresponds to the block associated with that checkpoint. The number inside each hexagon (hovering a block) corresponds to the number of validators who made a checkpoint vote with this checkpoint as the target. The example starts at the slot before the attack in step 1. \texttt{GST} has been reached in epoch $e$ and honest validators have started to vote on chain $A$. This is the correct action because the justified checkpoint with the highest epoch is on chain $A$ (at epoch $e-2$). During the next slot in step 2, before reaching $limitSlot$, a Byzantine validator sends a block with withheld votes for the checkpoint at epoch $e-1$ on chain $B$. It is released just in time for a set of honest validators to consider it and too late for the remaining validators. The honest validators that see the block in time will update their view of the justified checkpoint with the highest epoch and consider chain $B$ as the canonical chain. We now show how the epoch continues with step 3. The block produced by a Byzantine, having been released just in time, causes $(1/3)$ of honest validators to change their view. This results in a situation where Byzantine validators can perform the same attack during the next epoch provided that at least one Byzantine validator is selected to be block proposer on chain $A$ for one of the first 8 slots.}
    \label{fig:probabilisticBouncingAttack3steps}
\end{figure}

\paragraph{Attack Condition.}

Our attack takes place during the synchronous period and uses the power of \textit{equivocation} by Byzantine processes. Equivocation is caused by a Byzantine process that sends a message only to a subset of validators at a given point in time and potentially another message or none to another subset of validators. The effect is that only a part of the validators will receive the message on time. More in detail, the bounded network delay is used by a Byzantine validator to convey a message to be read on a specific slot by some validators and read on the next slot by the other validators. Note that if a protocol is not tolerant to equivocation, then it is not BFT (Byzantine Fault Tolerant), since equivocation is the typical action possible for Byzantine validators.

\paragraph{Attack Description and Analysis}

Let $\beta \le f/n$ be the fraction of Byzantine validators in the system. The attack setup is the following. First, as in the traditional bouncing attack, we start in a situation where the network is still partially synchronous. A fork occurs and results in the highest justified checkpoint being on chain $A$ at epoch $e$, and a justifiable checkpoint at epoch $e+1$ on chain $B$. Assume now that \texttt{GST} is reached. The attack can proceed\footnote{Note that before \texttt{GST}, no algorithm can ensure liveness since communication delays may not be bounded.} as follows:
\begin{enumerate}
    \item Since \texttt{GST} is reached, the network is fully synchronous. Chain $A$ is the canonical chain for all validators.
    \item Just before validators must stop updating their view concerning the justified checkpoint (i.e., before reaching the limit of $j$ slots\footnote{At the time of writing, 8 slots.} in the epoch corresponding to the condition in line 6 of \autoref{algo:syncBlock}), a Byzantine proposer proposes a block (cf. \autoref{algo:prepareBlock}) on chain $B$. This block contains attestations with enough checkpoint votes to justify the justifiable checkpoint left by honest validators. The attestations included in the block are those of Byzantine validators that were not issued in the previous epoch when they were supposed to be. The block must be released just in time, that is, right before the end of slot $j$, so that $(1/3 - \beta)$ of the validators change their view of the canonical chain to be active on chain $B$ while the rest of the honest validators continue on chain $A$. This is possible due to the patch preventing validators from changing their mind after $j$ slots.
    \item Repeat the process.
\end{enumerate}

An important aspect to consider in the attack is the probability of Byzantine validators becoming proposers. This is crucial because, without the role of proposer, validators cannot propose blocks and add new attestations containing checkpoint votes on the concurrent chain.\footnote{Note that Byzantine validators cannot use their role as proposer during the previous epoch to release a block with the right attestations because it might not be the last block of the epoch. Indeed, because some honest validators are on the concurrent chain, they add blocks. 
The checkpoint votes contained in the Byzantine attestations must be on the same chain as the attestation to justify the justifiable checkpoint, making the checkpoint justifiable in the first place.} 
The probability of being selected to be a proposer directly impacts how long the probabilistic bouncing attack can continue. In the following theorem, we establish the probability of a probabilistic bouncing attack lasting for a specific number of epochs.

\begin{theorem}
The probabilistic bouncing attack occurs during $k$ epochs after \texttt{GST} and a favorable setting with probability:
\begin{equation}\label{eq:livenessProba}
    P(\textrm{bouncing }k \textrm{ times})=(1-\alpha^j)^k,
\end{equation}

with $\alpha \in [0,1]$ being the proportion of honest validators and $j$ the number of slots before locking a choice for justification.
\end{theorem}

\begin{proof}
We denote by $\alpha$ the proportion of honest validators and $j$ the number of slots before locking the choice for justification. We want to know the probability of delaying the finality for $k$ epochs. Once we assume a setup condition sufficient to start a probabilistic bouncing attack, the attack continues until it becomes impossible for Byzantine validators to cast a vote to justify the justifiable checkpoint. To cast their vote, Byzantine validators need one of the $j$ first slots of the concurrent chain to have a Byzantine validator as proposer. Considering the probability of choosing between each validator, the chance for a Byzantine validator to be a proposer for one of the first $j$ slots is $(1-\alpha^j)$, with $\alpha$ being the proportion of honest validators. For $k$ epochs, we take this result to the power of $k$.
\end{proof}

\begin{figure}
    \centering
    \includegraphics[scale=.6]{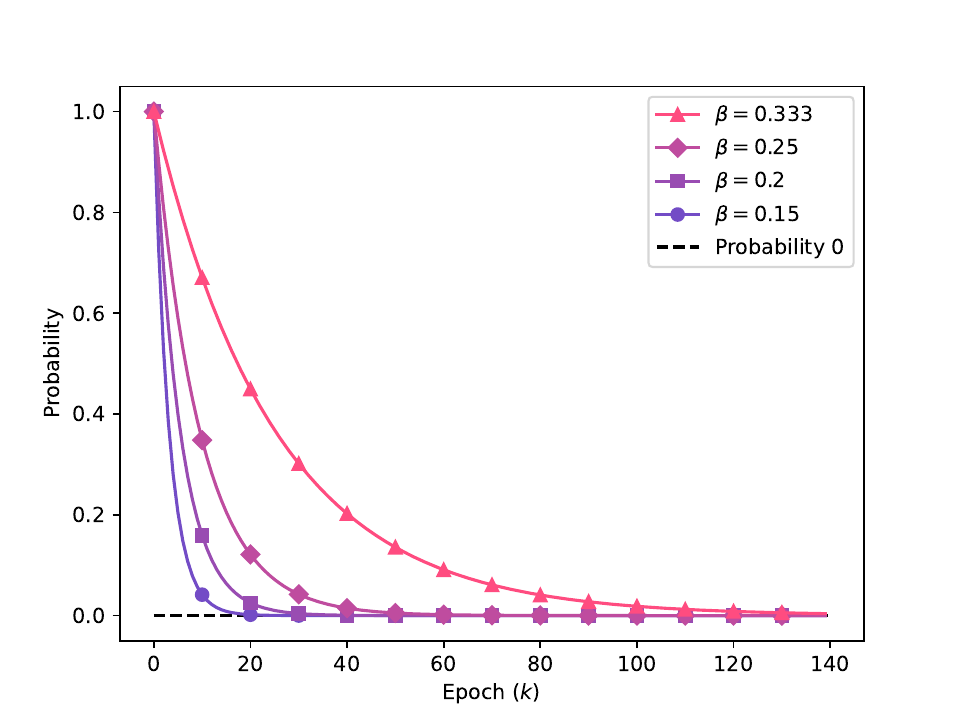}
    \caption{The figure presents the probability of the bouncing attack depending on the proportion of Byzantine validators $\beta$ and the number $k$ of epochs during which the bouncing attack lasts. The probability is computed based on \autoref{eq:livenessProba}, knowing that $\alpha = 1 - \beta$.}
    \label{fig:livenessProba}
\end{figure}

We depict the probability of the bouncing attack over time with several proportions of Byzantine validators $\beta$ in \autoref{fig:livenessProba}. The closer the proportion of Byzantine validators is to $1/3$, the higher the probability of the attack lasting for $k$ epochs (for any $k$).

The probability of the bouncing attack continuing for $k$ epochs depends on two factors: $\alpha$ ($=1-\beta$), the proportion of honest validators that cannot be controlled, and $j$, the number of slots before which validators are allowed to switch branches. Reducing $j$ to 0 would prevent the bouncing attack from happening (the probability falls to 0), but it would mean that validators are never allowed to change their view of the canonical chain. This naive solution would create irreconcilable choices among the set of validators and prevent any new checkpoint from being justified, which is a more severe threat to the liveness of Ethereum PoS.

Reducing the number of slots during which validators can change their view of the blockchain implies that different views cannot reconcile quickly. At the very least, the window of opportunity for doing so gets smaller. Theoretically, the proportion of Byzantine validators necessary to perform this attack is $1/n$. This is because we assume a favorable setup and that Byzantine validators can send messages so that only a desired portion of honest validators receive them on time. Our analysis focuses on the course of action of the attackers during the attack rather than the conditions necessary for it to occur.

This analysis highlights the delicate balance between fixing the view of validators and letting them change their view too much. There is a non-zero probability for the attack to last $k$ epochs for any $k$. However, the probability starts to plummet rapidly. With a proportion of $\beta = 0.3$, the probability of the attack lasting 100 epochs\footnote{100 epochs is about 10 hours.} is 0.02. Nonetheless, there is a non-negligible probability of delaying the finalization for several hours even with $\beta \approx 0.25$.

We conjecture that with the current design of Ethereum PoS, it is impossible to completely avoid such issues. Attempts to patch this probabilistic bouncing attack may not ensure a safe or live protocol; i.e., mitigating one attack might give rise to other vulnerabilities. For example, committee-based blockchains with single-shot finality employ complex systems to prevent conflicting and irreconcilable views \cite{amoussou_correctness_2018, amoussou_dissecting_2019, yin_hotstuff_2019}. These mechanisms typically require the exchange of messages from a quorum of validators to update one's perspective. However, such mechanisms are not feasible in Ethereum PoS since halting the blockchain's availability is not an option. More broadly, this issue aligns with existing literature \cite{neu_ebb_2021, lewispye_resource_2020} and their connections to the CAP theorem \cite{brewer_towards_2000}. Consequently, a possible straightforward solution could be to transition to a classical BFT consensus protocol with single-slot finality. Yet, such a change would fundamentally alter the protocol and should not be considered a mere mitigation.


\section{Conclusion}
In this study, we propose a novel distinction between the definitions of blockchain liveness and availability properties. This distinction is crucial for pinpointing differences between Nakamoto-style consensus and BFT (Byzantine Fault Tolerance) consensus, enabling a comparison between the two. We describe a framework for a high-level description of how the Ethereum PoS (Proof of Stake) protocol functions—a committee-based BFT protocol strongly inspired by Nakamoto-style consensus. Through this formalization, we demonstrate that the Ethereum PoS protocol satisfies the safety property and that no conflicting blocks can be finalized on different branches. We also present patches implemented to mitigate some attacks demonstrated on previous and preliminary versions of the protocol. However, we exhibit an attack on finalization, showing that the Ethereum PoS protocol exhibits probabilistic liveness. 

We supplement our analysis with an examination of the protocol's rewards and incentives, which were not considered here, in \autoref{chap:Eth-CModel-Penalties-Analysis}. We are interested in how the incentive mechanism can impact the safety of the protocol as well as the behavior of rational validators. We intend to implement the probabilistic attack and simulate its outcome while considering penalties.

Such analyses could be conducted empirically to closely monitor the actual behavior of the validators, especially in scenarios where an attack occurs.


    \chapter{Ethereum PoS Analysis under the Distributed Computing Model with Penalties}\label{chap:Eth-CModel-Penalties-Analysis}

\minitoc

\chapterLettrine{B}{eginning} with the consensus part of the protocol without the incentives, it now remains to consider them.
We knew from the beginning that taking into account the incentives would yield different results. However the complexity of the protocol was such that a first analysis focusing solely on the consensus part was necessary. 
We start this chapter by precising the network model for this analysis. It is very similar to the previous one but note that we define an \emph{initial} Byzantine proportion because it will vary during our analysis. We also redefine some properties and explain the protocol succinctly as a remainder. Then comes the analysis of the incentive mechanism called \emph{inactivity leak}.

\section{System Model}
The model for our analysis is defined in \autoref{chap:background}. For more clarity, there is a point of clarification regarding the network conditions.

During our analysis, we assume a network configuration where, during asynchronous periods, honest validators are split into two distinct partitions. Communication between these partitions is restricted, simulating a scenario where two regions are temporarily isolated from each other but maintain internal communication within each region. This setup emulates a situation where two regions of the world are temporarily unreachable from one another while maintaining unaffected communication within each region.

\section{Protocol and Properties}
In this section, we briefly remind the reader of the essential properties and definitions, as well as the necessary elements of the protocol previously described.

\paragraph*{Ethereum PoS Properties} Validators keep a local data structure in the form of a tree containing all the blocks perceived, and then a consensus protocol helps to choose a unique chain in the tree. Ethereum has a particular trait that consists of having a \emph{finalized} chain as the prefix of a chain vulnerable to forks. A metaphor for this is that the finalized chain is the trunk that possibly supports various branches, and as time passes, the trunk grows and branches are trimmed\footnote{We use the terms ``chain'' and ``branch'' interchangeably.}.

Intuitively, the Safety property of Ethereum states that the finalized chain is not forkable, while the Liveness property states that the finalized chain always grows. The nuance with respect to classical consensus protocols is the existence of an Availability property on the entire chain that guarantees constant growth of the chain despite failures and network partitions. The complete definition can be found in \autoref{sec:safetyAndLivenessProperties}.

Based on the definition of safety, we consider forks within the finalized chain as a loss of Safety. As explained in the subsequent section, forks occurring within the candidate chain suffix, which has not yet been finalized, are resolved by the fork choice rule of the protocol. This rule determines the chain upon which validators vote and build. However, this rule has not been explicitly designed to handle forks impacting the finalized chain.

The protocol is not intended to fork the finalized chain, as the finalization process depends on a super-majority vote, ensuring Safety when the Byzantine stake is less than one-third, i.e., $\beta_0 < 1/3$. We look at two types of Safety loss: (1) the finalization of two conflicting chains, and (2) the break of the Safety threshold, meaning the Byzantine stake proportion is more than one-third.







\section{Safety Attack}

The attestation contains two votes, a \textit{block vote} and a \textit{checkpoint vote}. The block vote is used in the \textit{fork choice rule}, which determines the chain to vote and build upon for validators. As its name suggests, the checkpoint vote points to checkpoints constituting the chain. It is used to justify and finalize blocks to grow the finalized chain. Justification is the step prior to finalization. If validators controlling over two-thirds of the stake make the same checkpoint vote, then the checkpoint target is justified. Finalization occurs when there are two consecutive justified checkpoints (one in epoch $e$ and the following one in epoch $e+1$).

Let us note that if justification occurs only every other epoch, finalization is not possible.

\subsection{Incentives}
The Ethereum PoS protocol provides validators with rewards and penalties to incentivize timely responses for reaching consensus. There are three different types of penalties: slashing, attestation penalties, and inactivity penalties. 

(i) \emph{Slashing penalties.} Validators face slashing if they provably violate specific protocol rules, resulting in a partial loss of their stake and expulsion from the validator set.

(ii) \emph{Attestation penalties.} To incentivize timely and correct attestations (votes), the protocol rewards validators for adhering to the protocol and penalizes those who do not. If an attestation is missing or belatedly incorporated into the chain, its validator gets penalized.

(iii) \textbf{\emph{Inactivity penalties.}} Each epoch a validator is deemed inactive, its \emph{inactivity score} increments. However, if the protocol is not in an inactivity leak, all inactivity scores are reduced.

When finalization occurs regularly, a validator that is deemed inactive only receives attestation penalties. This changes when there is no finalization for four consecutive epochs: the inactivity leak begins. During the inactivity leak, which starts when there is no finalization for four consecutive epochs, all validators will receive inactivity penalties directly linked to their stake and \emph{inactivity score}. The inactivity score varies with the validator's activity.

In addition to penalties, rewards are attributed for timely and correct attestations but not during the inactivity leak. Our analysis of the impact of the inactivity leak on the protocol takes into consideration the slashing and inactivity penalties across five different scenarios (cf. \autoref{sec:analysis}).

Having provided a comprehensive overview of the Ethereum PoS consensus mechanism, we are now well-positioned to delve into the specifics of the \textit{inactivity leak}.

\section{Inactivity Leak}\label{sec:inactivityLeak}

The Ethereum PoS blockchain strives for the continuous growth of the finalized chain. Consequently, the protocol incentivizes validators to actively finalize blocks. In the absence of finalization, validators incur penalties.

The inactivity leak, introduced in \cite{buterin_casper_2017}, serves as a mechanism to regain finality. Specifically, if a chain has not undergone finalization for four consecutive epochs, the inactivity leak is initiated. During the inactivity leak, the stakes of \textit{inactive} validators are drained until active validators amount to two-thirds of the stake. A validator is labeled as inactive for a particular epoch if it fails to send an attestation or sends one with a wrong target checkpoint.

During the inactivity leak, there are no more rewards given to attesters\footnote{Actually, the only rewards that remain are for the block producers and the sync committees.}, and additional penalties are imposed on inactive validators.

\subsection{Inactivity Score}
The \emph{inactivity score} is a dynamic variable that adjusts based on a validator's activity. The inactivity score of a validator is determined based on the attestations contained in the chain. It is important to note that if there are multiple branches, a validator's inactivity score depends on the selected branch. Within an epoch, being active on one branch implies\footnote{This is true as long as the chains differ for at least one epoch.} inactivity on another (for honest validators).

More precisely, the inactivity score is updated every epoch: if validator $i$ is active, then its inactivity score is reduced by 1; otherwise, 4 is added to it. When the inactivity leak is not in place, the inactivity scores are decreased by 16 every epoch, which often nullifies low inactivity scores.

During an inactivity leak, at epoch $t$, the inactivity score, $I_i(t)$, of validator $i$ is:
\begin{equation}
    \begin{cases}
        I_i(t) = I_i(t-1) + 4, & \text{if $i$ is inactive at epoch $t$} \\
        I_i(t) = \max(I_i(t-1) - 1, 0), & \text{otherwise.}
    \end{cases}
\end{equation}

Each attester thus has an inactivity score that fluctuates depending on its (in)activity. In the protocol, the inactivity score is always greater than zero. A validator's inactivity for epoch $t$ is determined by whether it sent an attestation for this epoch and if the sent attestation contains a correct checkpoint vote. Here "correct" implies that the target of the checkpoint vote belongs to the considered chain.

\subsection{Inactivity penalties}
Validators that are deemed inactive incur penalties. Let $s_i(t)$ represent the stake of validator $i$ at epoch $t$, and let $I_i(t)$ denote its inactivity score. The penalty at each epoch $t$ is $I_i(t-1) \cdot s_i(t-1) / 2^{26}$. Therefore, the evolution of the stake is expressed by:
\begin{equation}\label{eq:discreteStake}
s_i(t) = s_i(t-1) - \frac{I_i(t-1) \cdot s_i(t-1)}{2^{26}}.
\end{equation}

\subsection{Stake's functions during an inactivity leak}\label{subsec:stakeFunctions}

In this work, we model the stake function $s$ (see \autoref{eq:discreteStake}) as a continuous and differentiable function, yielding the following differential equation:
\begin{equation}\label{eq:stakeFunction}
s'(t) = -\frac{I(t) \cdot s(t)}{2^{26}}.
\end{equation}

We then explore three distinct validator behaviors during an inactivity leak, each influencing their inactivity score and, consequently, their stake.

\begin{enumerate}[label=(\alph*)] 
    \item Active validators: they are always active. 
    \item Semi-active validators: they are active every two epochs. 
    \item Inactive validators: they are always inactive.
\end{enumerate}

Note that, in the case of a fork, this categorization depends on the specific branch under consideration as different branches may yield different evaluations of each validator's behavior.

This categorization is orthogonal to the Byzantine-Honest categorization.
For instance, an honest validator can appear inactive in one branch due to poor connectivity or an asynchronous period (due to network partition or congestion). On the other hand, a Byzantine validator intentionally chooses one of these behaviors (e.g., being semi-active) to execute the attacks.

\begin{figure}
    \centering
    \includegraphics[width=0.8\linewidth]{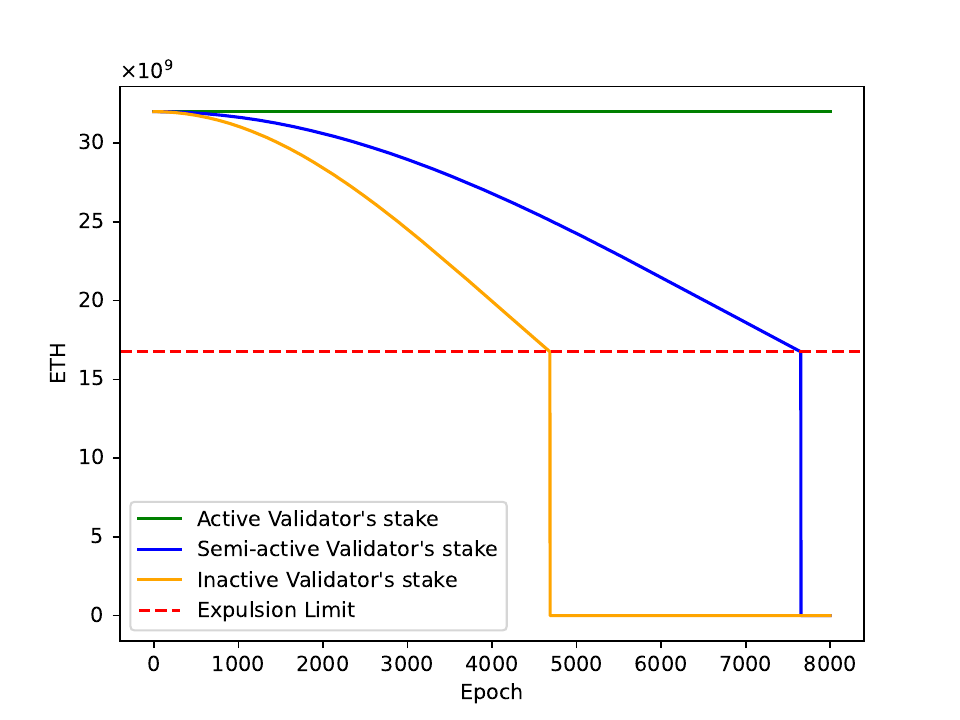}
    \caption{This figure shows three different stake trajectories in the event of an inactivity leak: the stake of a validator active every epoch, the stake of a validator active every other epoch (semi-active), and an inactive validator. The inactive validators get ejected at epoch $t=4685$. The semi-active validators get ejected at epoch $t=7652$. For reference, $5000$ epochs is about 3 weeks.}
    \label{fig:validatorsStake}
\end{figure}

We illustrate in \autoref{fig:validatorsStake} the evolution of the validators' stake depending on their behaviors. We also account for the ejection of validators with a stake lower than or equal to $16.75$.

Using these newly defined stake functions, we explore five scenarios in \autoref{sec:analysis}.

The first scenario, with only honest validators, serves as a baseline to assess the impact of Byzantine validators. Even in this seemingly straightforward setting, Safety is compromised.

In the second scenario, Byzantine validators come into play and aim to expedite the finalization of conflicting branches. They do so by performing slashable actions. Thus, they will get ejected from the set of validators once communication is restored among honest validators and evidence of their slashable offense is included in a block. We outline their impact based on their initial stake proportion. With an initial stake proportion of $\beta_0=0.2$, the finalization on conflicting chains occurs after $3107$ epochs. With $\beta_0=0.33$, the conflicting finalization occurs only after $503$ epochs.

In the third and fourth scenarios, Byzantine validators exhibit non-slashable behaviors. Specifically, Byzantine validators are semi-active, meaning they are active on both chains but in a non-slashable manner. In the third scenario, they aim to finalize conflicting branches as soon as possible, achieving conflicting finalization in $556$ epochs with an initial stake proportion of $\beta_0 = 0.33$. In the fourth scenario, their goal is to increase their stake proportion to exceed the $1/3$ threshold.

The last scenario delves into the effect of the probabilistic bouncing attack regarding the Byzantine stake proportion, considering the inactivity leak. In this attack, Byzantine validators initially aim to delay finality by being alternately active (bouncing) on both chains of a fork. This confuses honest validators, causing them to also bounce from one chain to the other. We detail how to find the distribution of honest validators' stakes in this setting, considering the inactivity penalties. We also cover how the Byzantine validators' stake proportion can exceed $1/3$ if their initial proportion is close to $1/3$.

The scenarios unfold within the context of a partially synchronous network while offering a meticulous examination of the property of Safety and the evolution of the proportion of Byzantine validators. Each scenario's initial conditions and outcomes are summarized in \autoref{tab:recapScenario}.

\begin{table}
\begin{center}
\begin{tabular}{|c|c|c||c|}
\hline
Scenario  & Outcomes \\ \hline \hline
\ref{subsec:HonDoubleFinalization} \hfill All honest \hfill\null &  2 finalized branches \\ 
\ref{subsubsec:withSlashing}  \hfill Slashable Byzantine \hfill\null & 2 finalized branches \\
\ref{subsubsec:withoutSlashing}  \hfill Non slashable Byzantine \hfill\null & 2 finalized branches \\
\ref{subsubsec:oneThirdByz}  \hfill Non slashable Byzantine \hfill\null & $\beta>1/3$ \\
\ref{subsec:revisitPBA}  \hfill Probabilistic Bouncing attack \hfill\null & $\beta>1/3$ probably \\ \hline
\end{tabular}
\end{center}
\caption{Analyzed scenarios associated with their outcomes. Initially the proportion of Byzantine's stake is smaller than $1/3$ and is zero for the first scenario. }
\label{tab:recapScenario}
\end{table}

\section{Analysis}\label{sec:analysis}

In this section, we study the robustness of the Safety property within the context of the inactivity leak. By construction, in the case of a prolonged partition, two different chains can potentially be finalized, leading to conflicting finalized blocks. We delineate scenarios that can produce such a predicament.

Considering the presence of Byzantine validators, we study how the proportion of Byzantine validators' stake evolves during an inactivity leak. Furthermore, we are interested in scenarios where the inactivity leak mechanism becomes the backbone of an attack strategy, potentially causing the proportion of Byzantine stakes to exceed the $1/3$ security threshold (cf. \autoref{subsubsec:oneThirdByz} and \autoref{subsec:revisitPBA}).

\subsection{GST upper bound for Safety}
\label{subsec:HonDoubleFinalization}
In this first subsection, we look for an upper bound on \texttt{GST} before which no finalization on conflicting chains can happen in case of a partition. We study the case of an inactivity leak with these conditions: (i) only honest validators, no Byzantine validators, and (ii) the network is asynchronous (before \texttt{GST}).

In case of catastrophic events, during an instance of a particularly disrupted network, an arbitrarily large set of honest validators might be unreachable before \texttt{GST}. During this asynchronous period, the subset of validators still communicating with each other will continue to try to finalize new blocks. We assume that, within each partition, the message delay is bounded as in the synchronous period; however, communication between partitions is not restored before \texttt{GST}. As the system model mentions, Byzantine validators can communicate between partitions without restriction but cannot manipulate the message delay between honest validators. The active validators must represent more than two-thirds of the stake to be able to finalize. After $4$ epochs without finalization, the inactivity leak starts.

All the validators deemed inactive will have their stake progressively reduced. This will continue until the active validators constitute at least two-thirds of the stake and can finalize anew.

\paragraph{Two finalized chains} A noteworthy scenario arises during asynchronous periods that can lead to a network partition and the creation of two distinct finalized chains. If this partition persists for an extended period, both chains independently drain the stakes of validators they consider inactive until they finalize again. Although the protocol permits this behavior by design, it results in finalizing two conflicting chains, thereby compromising the Safety property.

This outcome aligns with Ethereum PoS prioritizing Liveness over Safety. However, to the best of our knowledge, this corner case has not been discussed in detail.

We can theoretically assess the time required to finalize both branches of the fork. Suppose honest validators remain in their respective branches due to the partition. In this case, by understanding the distribution of these validators across the partitions, we can compute the time it takes for the proportion of active validators' stake to return to $2/3$ of the stake on each branch, permitting new finalization.

Let $n_{\rm H}$ and $n_{\rm B}$ denote the initial number of honest validators and Byzantine validators at the beginning of the inactivity leak ($n_{\rm H}+n_{\rm B}=n$). Additionally, $n_{\rm H_1}$ and $n_{\rm H_2}$ represent the number of honest validators active on branch 1 and on branch 2, respectively ($n_{\rm H_1}+n_{\rm H_2}=n_{\rm H})$.

We denote by $p_0=n_{\rm H_1}/n_{\rm H}$ the initial proportion of honest validators remaining active on branch 1, and $1-p_0=n_{\rm H_2}/n_{\rm H}$ represents the proportion of honest validators active on branch 2 (hence inactive on branch 1). In this first scenario, with only honest validators and no Byzantine validators, $p_0$ represents the proportion of all validators active on branch 1. Indeed, since $n_{\rm H}/n=1$, we have that $n_{\rm H_1}/n_{\rm H} \times n_{\rm H}/n=p_0$.

We have assessed how validators' stakes vary based on their level of activity. Consequently, we can express the ratio of active validators on branch 1 at time $t$ as:
\begin{equation}\label{eq:preliHonestActiveRatio}
\frac{n_{\rm H_1}s_{\rm H_1}(t)}{n_{\rm H_1}s_{\rm H_1}(t)+n_{\rm H_2}s_{\rm H_2}(t)} ,
\end{equation}
with $s_{\rm H_1}$ and $s_{\rm H_2}$ being the stake of honest active and inactive validators, respectively. We know the function of their stake according to time, and by dividing the numerator and the denominator by the total number of validators ($n=n_{\rm H}$), we can rewrite \autoref{eq:preliHonestActiveRatio} as:
\begin{equation}\label{eq:honestActiveRatio}
\frac{p_0}{p_0+(1-p_0)e^{-t^2/2^{25}}}.
\end{equation}
The initial stake value $s_0$ is factored out of the equation. This function is critical as the moment it reaches $2/3$ or more, finalization can occur\footnote{With a proportion of two-thirds of validators' stake active, justification and then finalization can occur in $2$ epochs.} on the branch.

To establish the upper bound on \texttt{GST} under which two conflicting branches finalize, we must find when finalization occurs on each branch for each initial proportion of active validators $p_0$ and inactive validators $1-p_0$. We simulate the evolution of the ratio of active validators (\autoref{eq:honestActiveRatio}) during an inactivity leak with different values of $p_0$ in \autoref{fig:honestActiveRatio}. The simulation starts with both active and inactive validators at $32$ ETH. At epoch $0$, the inactivity leak begins.

For $p_0=0.5$ or less, the ratio jumps to 1 at $t=4685$; this is due to the fact that validators with a stake below $16.75$ ETH are ejected from the set of validators. Conversely, for $p_0=0.6$, the proportion of active validators does not jump drastically as $2/3$ of active validators is regained before the ejection of inactive validators, permitting the active validators to finalize, hence ending the inactivity leak. Interestingly, with $p_0=0.6$ we can see that the ratio still increases several epochs after the proportion of $2/3$ of active validators' stake is reached. This is because the penalties for inactive validators take some time to return to zero.

\begin{figure}
    \centering
    \includegraphics[width=.8\linewidth]{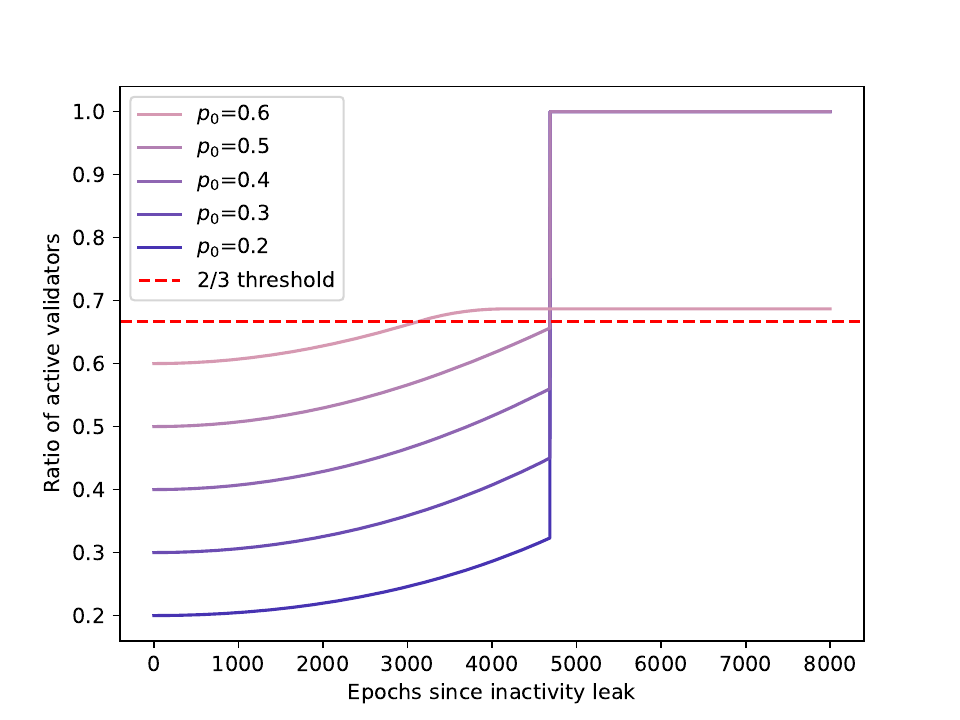}
    \caption{Evolution of the ratio of active validators depending on the proportion $p_0$ of active validators on the branch. This follows the ratio given in \autoref{eq:honestActiveRatio} before regaining $2/3$ of active validators or the expulsion of inactive validators at epoch $t=4685$.}
    \label{fig:honestActiveRatio}
\end{figure}

As expected and shown by \autoref{fig:honestActiveRatio}, a chain with more active validators will regain finality faster. To ascertain how quickly, we seek when the ratio is equal to 2/3. Taking into account the expulsion\footnote{We drew inspiration for this initial work from the insights presented in \cite{edington_technical_2023}.} of inactive validators at $t=4685$, we can find the value $t$ at which the $2/3$ threshold is reached:
\begin{equation}
      t = \min \left( \sqrt{2^{25}[\log(2(1-p_0))-\log(p_0)]}, 4685 \right).
\end{equation}

This calculation pertains to $0<p_0<2/3$ (when there are less than $2/3$ of active validators), ensuring that the epoch $t$ can be computed.

Conflicting finalization occurs once the slowest branch to finalize has regained finality. Our observation highlights that the lower the proportion of active validators, the slower the branch will regain finality. Hence, the fastest way to reach finality on both chains would be for honest validators to be evenly proportioned, with half of the validators active on one chain and the other half on the other chain ($p_0=1-p_0=0.5$). In this case, the ratio of active validators amounts to $2/3$ on both chains at $t=4685$ epochs (about 3 weeks). We can note here that even with the best configuration to finalize quickly on conflicting branches, it is impossible to be faster than $4685$ epochs. Thus, with only honest validators, whatever their proportion on each branch, the last chain to finalize will always finalize at $t=4685$.

Finality on both chains is achieved precisely at $4686$ epochs after the beginning of the inactivity leak. Adding an epoch is necessary after gaining $2/3$ of active stake to finalize the preceding justified checkpoint. This finalization ends the inactivity leak, which has lasted approximately 3 weeks. \textit{Any network partition lasting longer than $4686$ epochs will result in a loss of Safety because of conflicting finalization. This is an upper bound for Safety on the duration of the inactivity leak with only honest validators.}

\subsection{Upper bound decrease due to Byzantine validators}\label{subsec:ByzDoubleFinalization}

In a trivial setup with only honest validators, Safety does not hold if the inactivity leak is not resolved quickly. This prompts us to study the scenario in the presence of Byzantine validators to evaluate how much they will be able to hasten the conflicting finalization. We describe two possible outcomes: the first one violates Safety, but Byzantine validators get slashed; the second one violates Safety as well, but no validators get slashed. A slashing penalty entails an ejection from the validator set as well as a loss of part of the validator's stake. Both scenarios expedite the time $t$ at which Safety is breached, with different velocities depending on the chosen method.

We study the inactivity leak under these conditions: (i) at the beginning, less than one-third of the stake is held by Byzantine validators ($\beta_0=n_{\rm B}/n<1/3$), with the rest held by honest validators ($1-\beta_0=n_{\rm H}/n$); (ii) the network is asynchronous (before \texttt{GST}); and (iii) Byzantine validators are not affected by network partitions.\footnote{In a model without partitions, one needs to give Byzantine validators more power to recreate our scenario. They must be able to control the network delay to allow them to be active on both branches while preventing honest validators from even observing the branch on which they are not active. They can manipulate message delays between groups of honest validators to simulate a partition between them.}

The situation is as follows:
\begin{itemize}
    \item Honest validators are divided into branches 1 and 2; a proportion $p_0=n_{\rm H_1}/n_{\rm H}$ of the honest validators are active on branch 1, while a proportion $1-p_0=n_{\rm H_2}/n_{\rm H}$ are active on branch 2. This means that on branch 1, a proportion $n_{\rm H_1}/n_{\rm H} \times n_{\rm H}/n = p_0(1-\beta_0)$ are honest and active, and a proportion $n_{\rm H_2}/n_{\rm H} \times n_{\rm H}/n = (1-p_0)(1-\beta_0)$ are honest and inactive.
    \item Byzantine validators are not restricted to either partition; they are connected to both.
\end{itemize}

\subsubsection{With slashing}\label{subsubsec:withSlashing}

In the event of a fork during asynchronous times, Byzantine validators can be active on both branches (\autoref{fig:ByzSlashingSchema}). Being active on two branches means sending correct attestations on both every epoch. Such behavior is considered a slashable offense, incurring penalties, but only if detected by honest validators. The slashable offense is punished once proof of conflicting attestations during the same epoch has been included in a block. Thus, before \texttt{GST}, Byzantine validators could operate on both branches without facing punishment as long as honest validators are unaware of the conflicting attestations. Byzantine validators have control over the message delay before \texttt{GST}, making this behavior possible. They can thereby expedite the finalization on different branches.

\begin{figure}
    \centering
    \resizebox{.5\linewidth}{!}{
        \begin{tikzpicture}[x=0.75pt,y=0.75pt,yscale=-1,xscale=1]

\draw    (50,125) -- (100,100) ;
\draw    (50,125) -- (100,150) ;
\draw  [color={rgb, 255:red, 255; green, 0; blue, 0 }  ,draw opacity=1 ][fill={rgb, 255:red, 255; green, 52; blue, 52 }  ,fill opacity=0.86 ] (100,100) -- (300.8,100) -- (300.8,110) -- (100,110) -- cycle ;
\draw  [color={rgb, 255:red, 255; green, 0; blue, 0 }  ,draw opacity=1 ][fill={rgb, 255:red, 255; green, 52; blue, 52 }  ,fill opacity=0.86 ] (100,150) -- (300.8,150) -- (300.8,160) -- (100,160) -- cycle ;
\draw    (100,100) -- (300.8,100) ;
\draw    (100,150) -- (300.8,150) ;
\draw  [dash pattern={on 0.84pt off 2.51pt}]  (100,60) -- (100,170) ;
\draw  [dash pattern={on 0.84pt off 2.51pt}]  (150,60) -- (150,170) ;
\draw  [dash pattern={on 0.84pt off 2.51pt}]  (200,60) -- (200,170) ;
\draw  [dash pattern={on 0.84pt off 2.51pt}]  (250,60) -- (250,170) ;

\draw (128,70) node  [font=\small,xscale=1.25,yscale=1.25]  {$t_{1}$};
\draw (178,70) node  [font=\small,xscale=1.25,yscale=1.25]  {$t_{2}$};
\draw (228,70) node  [font=\small,xscale=1.25,yscale=1.25]  {$t_{3}$};
\draw (272,70) node  [font=\small,xscale=1.25,yscale=1.25]  {$t_{4}$};
\draw (78,70) node  [font=\small,xscale=1.25,yscale=1.25]  {$t_{0}$};

\end{tikzpicture}
    }
    \caption{Byzantine validators are active on both chains of a fork simultaneously during asynchronous times.}
    \label{fig:ByzSlashingSchema}
\end{figure}
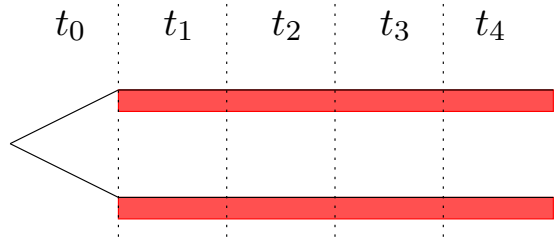

We study here the time needed for finalization to occur on conflicting branches depending on the proportion of Byzantine validators. The ratio of active validators at epoch $t$ is:
\begin{equation}
    \frac{n_{\rm H_1}s_{\rm H_1}(t)+n_{\rm B}s_{\rm B}(t)}{n_{\rm H_1}s_{\rm H_1}(t)+n_{\rm B}s_{\rm B}(t)+n_{\rm H_2}s_{\rm H_2}(t)},
\end{equation}
with $s_{\rm H_1}$, $s_{\rm B}$, and $s_{\rm H_2}$ being the stake of honest active, Byzantine active, and honest inactive validators, respectively. This can be rewritten as:
\begin{equation}
     \frac{p_0(1-\beta_0)+\beta_0}{p_0(1-\beta_0)+\beta_0+(1-p_0)(1-\beta_0)e^{-t^2/2^{25}}}.
\end{equation}
where $\beta_0$ represents the initial proportion of Byzantine validators, and $p_0$ denotes the initial proportion of honest active validators. In contrast to the analysis with only honest validators (cf. \autoref{eq:honestActiveRatio}), here, Byzantine validators are present and active on both chains. Nonetheless, as before, we can obtain the ratio of active validators on the other branch just by interchanging $p_0$ and $1-p_0$.
Finality on conflicting branches occurs when the last of the two branches finalizes. Similarly to the previous example, the branch with the fewer initial honest active validators ($p_0$) will finalize the latest. This happens $t$ epochs after the beginning of the inactivity leak, with
\begin{equation}\label{eq:byzTimeForFinalization}
      t = \min\left(\sqrt{2^{25}\left[\log(2(1-p_0))-\log(p_0+\frac{\beta_0}{1-\beta_0})\right]}, 4685 \right).
\end{equation}

Finality on conflicting branches is achieved the quickest when honest validators are evenly split between the branches of the fork, at $p_0=0.5$.

\begin{table}[htbp]
\begin{center}
\begin{tabular}{|c|c|}
\hline
\textbf{$\beta_0$} & \textbf{$t$} \\
\hline
$\mathbf{0}$ & $\mathbf{4685}$ \\
 $0.1$ & $4066$ \\ 
 $0.15$ & $3622$ \\ 
 $0.2$ & $3107$ \\
 $0.33$ & $502$ \\
\hline
\end{tabular}
\caption{Time before finalization on conflicting branches depending on the initial proportion of Byzantine validators $\beta_0$ for $p_0=0.5$ with slashing behaviour based on \autoref{eq:byzTimeForFinalization}}
\label{tab:finalizationSlashing}
\end{center}
\end{table}

\autoref{tab:finalizationSlashing} gives the epoch at which concurrent finalization occurs for $p_0=0.5$. This outlines the rapidity at which finality can be regained depending on the initial proportion $\beta_0$ of Byzantine validators' stake. The table shows that $503$ epochs (approximately 2 days) could suffice to finalize blocks on two different chains, but hypothetically it could be quicker than that. In fact, as $\beta_0$ gets closer to $1/3$, the number of epochs required before concurrent finalization occurs (\autoref{eq:byzTimeForFinalization}) approaches $0$.

The explanation is that if $\beta_0$ were to start at exactly 1/3, then with $p_0=0.5$, it would mean that on each branch we would start with $p_0(1-\beta_0)+\beta_0=2/3$ of active validators, hence finalizing immediately. This explains why if $\beta_0$ is very close to 1/3, the proportion of active validators reaches $2/3$ rapidly.
\textit{Hence, Byzantine validators can expedite the loss of Safety. If their initial proportion is $0.33$, they can make conflicting finalizations occur approximately ten times faster than scenarios involving only honest participants.}

\

One can notice that if Byzantine validators act in a slashable manner, they will be penalized after the asynchronous period ends. However, the harm is already done. Once the finalization on two branches has occurred, the branches are irreconcilable with the current protocol. Next, we demonstrate that Byzantine validators can employ more subtle strategies to break Safety without slashable actions.

\subsubsection{Without Slashing} \label{subsubsec:withoutSlashing}

Byzantine validators can hasten the violation of the Safety property without incurring a slashable offense. While not as rapid as being active on both branches simultaneously, they can be semi-active on both branches alternatively. Being semi-active on each branch means they are only active every other epoch. This approach diminishes their stake on each branch due to inactivity penalties. Nevertheless, at some point, they can finalize on two conflicting branches by being active two epochs in a row on one branch and then on the other (see \autoref{fig:ByzNoSlashingSchema}). Byzantine validators will be able to finalize when the proportion of their stake plus the proportion of the stake of honest active validators is above $2/3$ on the branch (cf. \autoref{eq:ByzSemiActiveRatio}).

\begin{figure}
    \centering
    \resizebox{.8\linewidth}{!}{
        \begin{tikzpicture}[x=0.75pt,y=0.75pt,yscale=-1,xscale=1,scale=0.6, every node/.style={scale=0.8}]

\draw    (50,125) -- (100,100) ;
\draw    (50,125) -- (100,150) ;
\draw  [color={rgb, 255:red, 255; green, 0; blue, 0 }  ,draw opacity=1 ][fill={rgb, 255:red, 255; green, 52; blue, 52 }  ,fill opacity=0.86 ] (100,100) -- (150,100) -- (150,110) -- (100,110) -- cycle ;
\draw  [color={rgb, 255:red, 255; green, 0; blue, 0 }  ,draw opacity=1 ][fill={rgb, 255:red, 255; green, 52; blue, 52 }  ,fill opacity=0.86 ] (150,150) -- (200,150) -- (200,160) -- (150,160) -- cycle ;
\draw  [color={rgb, 255:red, 255; green, 0; blue, 0 }  ,draw opacity=1 ][fill={rgb, 255:red, 255; green, 52; blue, 52 }  ,fill opacity=0.86 ] (200,100) -- (250,100) -- (250,110) -- (200,110) -- cycle ;
\draw  [color={rgb, 255:red, 255; green, 0; blue, 0 }  ,draw opacity=1 ][fill={rgb, 255:red, 255; green, 52; blue, 52 }  ,fill opacity=0.86 ] (250,150) -- (300,150) -- (300,160) -- (250,160) -- cycle ;
\draw  [dash pattern={on 0.84pt off 2.51pt}]  (100,60) -- (100,170) ;
\draw  [dash pattern={on 0.84pt off 2.51pt}]  (150,60) -- (150,170) ;
\draw  [dash pattern={on 0.84pt off 2.51pt}]  (200,60) -- (200,170) ;
\draw  [dash pattern={on 0.84pt off 2.51pt}]  (250,60) -- (250,170) ;
\draw  [color={rgb, 255:red, 255; green, 0; blue, 0 }  ,draw opacity=1 ][fill={rgb, 255:red, 255; green, 52; blue, 52 }  ,fill opacity=0.86 ] (351,100) -- (451,100) -- (451,110) -- (351,110) -- cycle ;
\draw  [color={rgb, 255:red, 255; green, 0; blue, 0 }  ,draw opacity=1 ][fill={rgb, 255:red, 255; green, 52; blue, 52 }  ,fill opacity=0.86 ] (450,150) -- (550,150) -- (550,160) -- (450,160) -- cycle ;
\draw    (100,100) -- (310,100) ;
\draw  [dash pattern={on 0.84pt off 2.51pt}]  (350,60) -- (350,170) ;
\draw  [dash pattern={on 0.84pt off 2.51pt}]  (400,60) -- (400,170) ;
\draw  [dash pattern={on 0.84pt off 2.51pt}]  (450,60) -- (450,170) ;
\draw  [dash pattern={on 0.84pt off 2.51pt}]  (500,60) -- (500,170) ;
\draw  [dash pattern={on 0.84pt off 2.51pt}]  (300,60) -- (300,170) ;
\draw    (340,100) -- (550,100) ;
\draw    (100,150) -- (310,150) ;
\draw    (340,150) -- (550,150) ;
\draw  [dash pattern={on 4.5pt off 4.5pt}]  (310,100) -- (340,100) ;
\draw  [dash pattern={on 4.5pt off 4.5pt}]  (310,150) -- (340,150) ;

\draw (128,70) node  [font=\small,xscale=1.25,yscale=1.25]  {$t_{1}$};
\draw (178,70) node  [font=\small,xscale=1.25,yscale=1.25]  {$t_{2}$};
\draw (228,70) node  [font=\small,xscale=1.25,yscale=1.25]  {$t_{3}$};
\draw (272,70) node  [font=\small,xscale=1.25,yscale=1.25]  {$t_{4}$};
\draw (78,70) node  [font=\small,xscale=1.25,yscale=1.25]  {$t_{0}$};
\draw (378,70) node  [font=\small,xscale=1.25,yscale=1.25]  {$t_{n}$};
\draw (428,70) node  [font=\small,xscale=1.25,yscale=1.25]  {$t_{n+1}$};
\draw (478,70) node  [font=\small,xscale=1.25,yscale=1.25]  {$t_{n+2}$};
\draw (522,70) node  [font=\small,xscale=1.25,yscale=1.25]  {$t_{n+3}$};
\draw (329.5,70) node  [font=\small,xscale=1.25,yscale=1.25]  {$.\ .\ .$};

\end{tikzpicture}
    }
    \caption{Byzantine validators active on both branches of a fork alternatively during asynchronous times.}
    \label{fig:ByzNoSlashingSchema}
\end{figure}
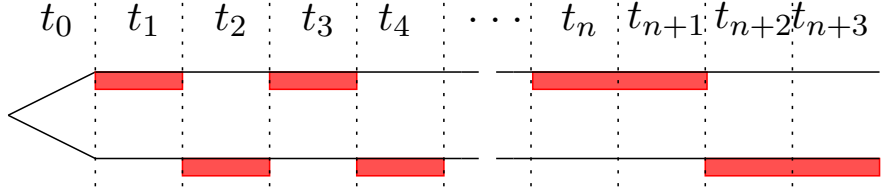

At that point, Byzantine validators must remain active for two consecutive epochs on each branch to finalize them both. If they are only semi-active, they can alternate justifications for checkpoints on each branch but will not achieve finalization. However, by maintaining activity for two consecutive epochs, first on one branch and then on the other, they ensure two sequential justifications, leading to the finalization of a checkpoint.

\

We gave the different evolution of stakes depending on the activity of validators (\autoref{subsec:stakeFunctions}). Now that Byzantine validators are semi-active, their stake follows the curve $s_0e^{-3t^2/2^{28}}$. We simplify the ratio as previously and we get that finalization occurs on the branch when the ratio
\begin{equation}\label{eq:ByzSemiActiveRatio}
    \frac{p_0(1-\beta_0)+\beta_0 e^{-3t^2/2^{28}}}{p_0(1-\beta_0)+\beta_0 e^{-3t^2/2^{28}}+(1-p_0)(1-\beta_0)e^{-t^2/2^{25}}}
\end{equation}
goes over 2/3, with $\beta_0$ and $p_0$ being the initial proportion of Byzantine validators and the proportion of honest active validators on the branch, respectively.

In contrast to the previous scenario, obtaining an analytic solution for $t$ to determine the epoch when the ratio hits $2/3$ is not straightforward. Therefore, we apply numerical methods on \autoref{eq:ByzSemiActiveRatio} with initial parameters $p_0=0.5$ and $\beta_0=0.33$, resulting in a calculated $t$ value of 555.65. This means it will take $556$ epochs to finalize, about 2 days and a half.

As previously, the proximity of $\beta_0$ to $1/3$ significantly influences the speed of finalization, as outlined in \autoref{tab:finalizationNoSlashing} and \autoref{fig:byzSlashANDbyzNoSlash}. \autoref{fig:byzSlashANDbyzNoSlash} shows how the proportion of Byzantine validators affects the time of conflicting finalization. Notice that although the acceleration is not as pronounced as in the previous scenario, it remains noteworthy that Byzantine validators still exert a substantial impact on breaching Safety, while not committing any slashable offense.

\textit{Hence, Byzantine validators can expedite the loss of Safety without committing any slashable action. If their initial proportion is $0.33$, they can make conflicting finalizations occur approximately eight times faster than scenarios involving only honest participants.}

\begin{table}[htbp]
\begin{center}
\begin{tabular}{|c|c|}
\hline
\textbf{$\beta_0$} & \textbf{$t$} \\
\hline
$\mathbf{0}$ & $\mathbf{4685}$ \\
 $0.1$ & $4221$ \\
 $0.15$ & $3819$\\
 $0.2$ & $3328$\\
 $0.33$ & $556$ \\
\hline
\end{tabular}
\caption{Time before finalization on conflicting branches depending on the initial proportion of Byzantine validators $\beta_0$ for $p_0=0.5$ without slashing behavior based on \autoref{eq:ByzSemiActiveRatio}.}
\label{tab:finalizationNoSlashing}
\end{center}
\end{table}

\begin{figure}
    \centering
    \includegraphics[width=.8\linewidth]{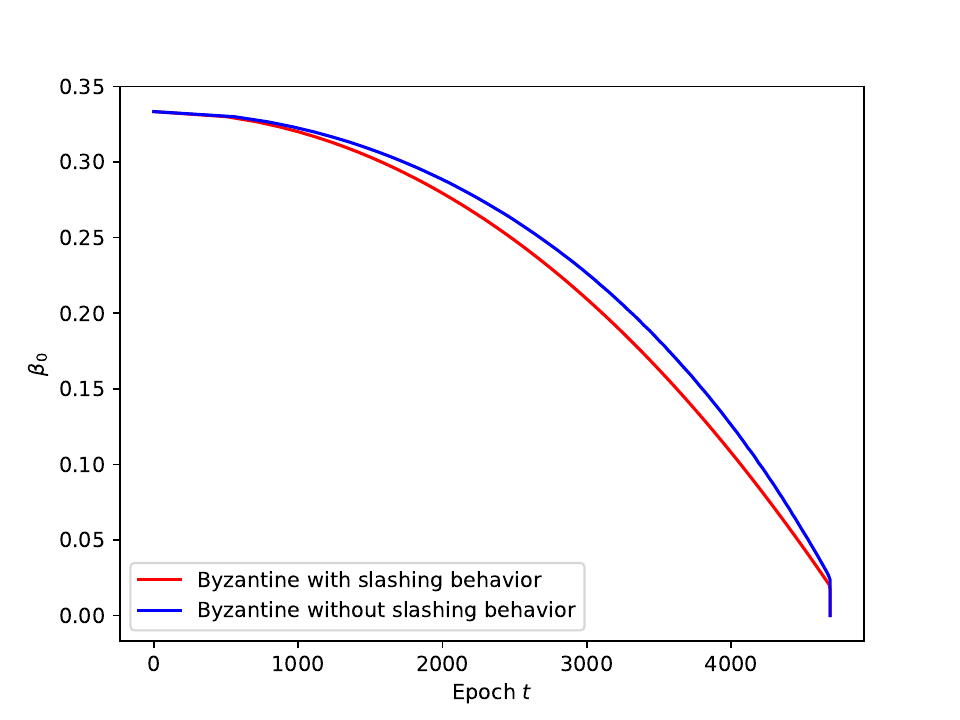}
    \caption{Time before finalization on conflicting branches, depending on the initial proportion of Byzantine validators $\beta_0$ and whether they engage in slashable actions.}
    \label{fig:byzSlashANDbyzNoSlash}
\end{figure}

Another consequence of being "semi-active" on both branches is that Byzantine validators can decide when to finalize on each branch. Indeed, even when the proportion of their stake plus the proportion of honest active validators' stake is above $2/3$, finalization only occurs when the Byzantine validators stay active for two consecutive epochs on the same chain. Being active for two epochs will justify the two consecutive epochs, thus finalizing an epoch.

There exists a scenario in which the Byzantine validators might delay finalization intentionally, aiming to increase their stake's proportion beyond the threshold of $1/3$ without incurring slashing afterward.

\subsubsection{More than one third of Byzantine validators}\label{subsubsec:oneThirdByz}

One may ask, why would Byzantine validators aim at going over the $1/3$ threshold? Indeed, we have just shown that Safety can be broken regardless of $\beta_0$; is it not the ultimate goal of Byzantine validators? It is not obvious to determine what behavior will harm the blockchain the most. We briefly discuss the impact Byzantine validators can have when they go over the $1/3$ threshold in \autoref{subsec:revisitPBA}. We now examine the necessary conditions on $\beta_0$ and $p_0$ that permit the Byzantine validators' stake to go over the one-third threshold.

The key ratio that translates into what we are looking for is the proportion of Byzantine validators' stake $\beta(t,p_0,\beta_0)$ over time:
\begin{equation}
    \frac{\beta_0 e^{-3t^2/2^{28}} }{p_0(1-\beta_0) +(1-p_0)(1-\beta_0)e^{-t^2/2^{25}} + \beta_0 e^{-3t^2/2^{28}}}
\end{equation}

As expected, at time $t=0$, $\beta(0,p_0,\beta_0)=\beta_0$. Now, let us investigate when this ratio is above the threshold of 1/3, i.e.:
\begin{equation}\label{eq:ByzOverOneThird}
\beta(t,p_0,\beta_0) \geq 1/3  
\end{equation}

The main difference with the previous scenario is that Byzantine validators seek to go over the $1/3$ threshold, not to finalize quickly. This means that even after the proportion of honest active validators' stake and semi-active Byzantine validators' stake represents more than two-thirds of the stake on the branch, they do not finalize. Byzantine validators could finalize by staying active for two epochs in a row, yet they do not do so in order to reach a higher stake proportion.

\

We construct a set containing the pairs $(p_0,\beta_0)$ that can lead $\beta$ to go over $1/3$ (\autoref{eq:ByzOverOneThird}). To do so, we take the point reached by the ratio when the validators deemed inactive are ejected. This point gives the highest value reachable\footnote{There exist more values that can lead to going over one-third when considering a special corner case. If the Byzantine validators strategically finalize just before the expulsion of honest inactive validators, the decrease in inactivity penalties might not occur quickly enough to prevent the ejection of honest inactive validators. In this particular scenario, Byzantine validators could potentially eject honest inactive participants while incurring fewer penalties themselves. This subtlety underscores the intricate dynamics at play during the inactivity leak.} for a particular $(p_0,\beta_0)$. For an intuition as to why this is the case, \autoref{fig:validatorsStake} allows us to visualize that the biggest gap between semi-active Byzantine stake and honest inactive stake is at the moment of expulsion of the honest inactive validators. We have seen that inactive validators are ejected from the chain after $4685$ epochs. We can thus evaluate the maximum ratio reachable $\beta_{\max}$ at time $t=4685$ when the inactive validators are ejected:
\begin{equation}
    \beta_{\max}(p_0,\beta_0) = \frac{\beta_0e^{-3\times (4685)^2/2^{28}}}{p_0(1-\beta_0) + \beta_0e^{-3\times (4685)^2/2^{28}}}.
\end{equation}

When this ratio is greater than 1/3, Byzantine validators have reached their goal. We show with \autoref{fig:pBetaShades} that Byzantine validators can actually go beyond the threshold of $1/3$ on both branches simultaneously. The lower bound $\beta_0$ before this becomes possible is for $p_0=0.5$ when $\beta_0=1/1+4e^{-3\times (4685)^2/2^{28}}=0.2421$.

\textit{When the initial proportion of Byzantine validators is at least $0.2421$, their proportion can eventually increase to more than $1/3$ of validators on both branches, exceeding the critical Safety threshold of voting power in each branch.}

\begin{figure}
    \centering
    \includegraphics[width=.8\linewidth]{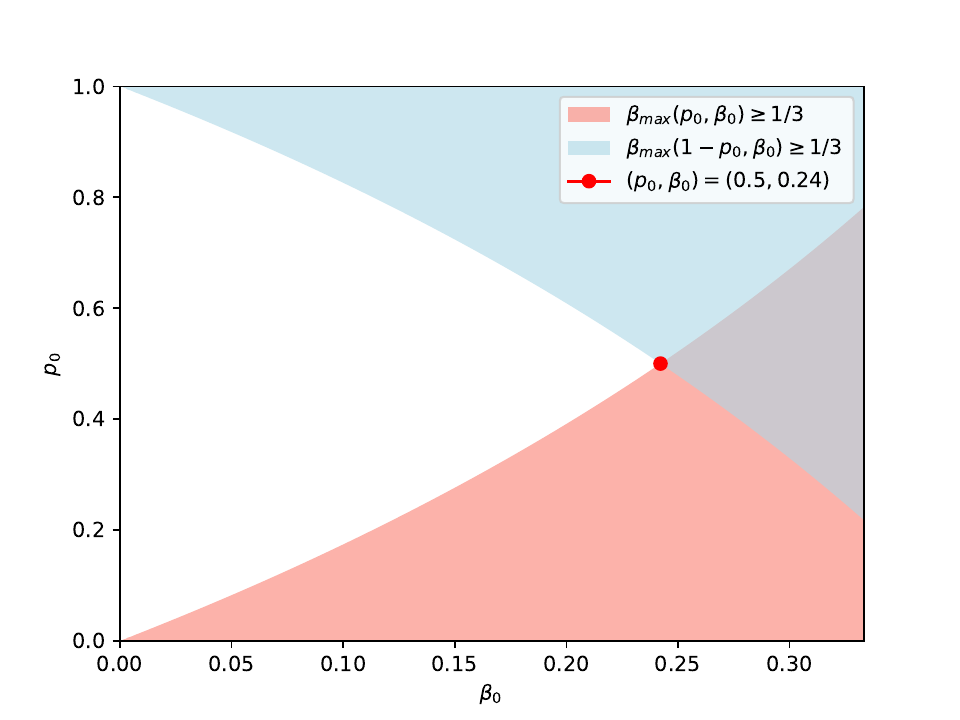}
     \caption{Pairs $(p_0,\beta_0)$ such that $\beta_{\max}(p_0,\beta_0)\geq 1/3$. This figure gives a lower bound for which $(p_0,\beta_0)$ can result in the proportion of Byzantine validators exceeding $1/3$ on both branches. } 
    \label{fig:pBetaShades} 
\end{figure}

Having explored scenarios in which protocol vulnerabilities manifest exclusively before \texttt{GST}, we now focus on potential threats posed by Byzantine validators after \texttt{GST}. Given the acknowledged impact of the \emph{Probabilistic Bouncing Attack} on Liveness (cf. \autoref{subsubsec:probabilisticBouncingAttack}), our study extends to take the inactivity leak into account.

\subsection{Revisiting the Probabilistic Bouncing Attack}\label{subsec:revisitPBA}

This subsection revisits the Probabilistic Bouncing Attack \autoref{subsubsec:probabilisticBouncingAttack}, showing that Byzantine validators could exceed the Safety threshold even during the synchronous period. Contrary to the previous scenarios, this one starts in the asynchronous period but unfolds in the synchronous period. This demonstrates that the inactivity leak poses significant challenges even within the synchronous period, revealing its broader implication for blockchain security.

As mentioned, while analyzing the probabilistic bouncing attack, we did not consider the penalties. Here, we fill this gap.

Let us note that there is no problem with conflicting finalization as the attack is progressing after \texttt{GST} in the synchronous period. In synchronous time, there is not enough delay for honest validators to miss a finalization on another branch. There would need to be more than two-thirds of the active stakes owned by Byzantine validators to break Safety in the synchronous period.

We briefly discussed the differences in gravity between conflicting finalization and having more than $1/3$ of the stake owned by Byzantine validators. We left the actual comparison and the in-depth analysis of the gravity of going beyond the infamous threshold as future work.

We primarily focus on identifying specific scenarios that would disrupt the network.
Thus, we give a detailed explanation of a scenario that could lead to Byzantine validators breaking the $1/3$ threshold even during synchronous period (after \texttt{GST}).

Let us remind how the attack takes place for self-containment.

\paragraph{Probabilistic Bouncing Attack Summary}

The attack can be summarized as follows: (1) A favorable setup that partitions honest validators into two different views of the blockchain occurs. (2) At each epoch, Byzantine validators withhold their messages from honest validators, releasing them at the opportune time to make some honest validators change their view. (3) This attack continues as long as at least one Byzantine validator is proposer in the $j^{th}$ first slots of the epoch, where $j$ is a parameter of the protocol. The probability of the attack continuing for $k$ epochs with a proportion of $(1-\beta_0)$ honest validators is $(1-(1-\beta_0)^j)^k$.

We start by analyzing the outcome of a fork where a proportion $p_0$ of the \emph{honest} validators start on chain $A$ and $1-p_0$ of the honest validators start on chain $B$.

We consider how a \textit{Probabilistic Bouncing Attack} would unfold, taking the inactivity leak into account. A probabilistic bouncing attack lasting more than $4$ epochs will necessarily cause an inactivity leak. Knowing this, we analyze the stakes of honest and Byzantine validators in this setting.

For this attack to continue, at each epoch, Byzantine validators cast their vote with a different chain as their candidate chain. They are active on each chain alternatively. Due to their inactivity every $2$ epochs, they will get ejected from the chain after a total of $7653$ epochs (4 weeks and 6 days). Byzantine validators are active on each chain to ensure that justification only happens every two epochs, preventing finalization from occurring.

For this attack to continue indefinitely, Byzantine validators must ensure honest validators are split into two branches according to two conditions:
(a) the honest validators are not enough to justify a chain on their own, and
(b) the Byzantine validators can justify it afterwards with their withheld votes.
This means that (a) $p_0$ must not represent more than $2/3$ of the stake, and (b) the proportions $p_0$ of honest validators and $\beta_0$ of Byzantine validators must represent more than two-thirds of the total stake.
The two necessary conditions are that (a) $p_0(1-\beta_0)<2/3$ and (b) $p_0(1-\beta_0)+\beta_0 > 2/3$. 
For the attack to function, we get that:
\begin{equation}\label{eq:p0conditions}
    \frac{2-3\beta_0}{3(1-\beta_0)} < p_0 < \frac{2}{3(1-\beta_0)} .
\end{equation}
We can see that the closer $\beta_0$ is to 0, the closer $p_0$ has to be from 2/3. This is to be expected as otherwise the Byzantine validators would be unable to justify the checkpoint with withheld votes.

An illustration of an ongoing attack with the probability for honest validators to be on one chain or the other is depicted in \autoref{fig:schemaAttackProba}. At each epoch, a proportion $p_0$ of honest validators is on one branch, whereas a proportion $1-p_0$ is on the other.

\begin{figure}
    \centering
    \begin{tikzpicture}[->,>=stealth,shorten >=1pt,node distance=2.5cm and 3cm]
  \tikzstyle{block} = [rectangle, draw, minimum width=1cm, minimum height=1cm]

  \node[block] (A) {$A$};
  \node[block, above right of=A] (B) {$B$};
  \node[block, below right of=A] (B') {$B'$};
  \node[block, right of= B] (C) {$C$};
  \node[block, right of= B'] (C') {$C'$};
  \node[block, right of= C] (D) {$D$};
  \node[block, right of= C'] (D') {$D'$};
  
  \draw (A) -- (B) node[midway, above, sloped]{$p_0$};
  \draw (A) -- (B') node[midway, below, sloped]{$1-p_0$}; 
  \draw (B) -- (C)  node[midway, above]{$1-p_0$} ;
  \draw (B') -- (C')  node[midway, below]{$p_0$} ;
  \draw (B) -- (C')  node[midway, below left, sloped]{$p_0$} ;
  \draw (B') -- (C)  node[midway, below right, sloped]{$1-p_0$}  ;
  
  \draw (C) -- (D)  node[midway, above]{$p_0$} ;
  \draw (C') -- (D')  node[midway, below]{$1-p_0$} ;
  \draw (C) -- (D')  node[midway, below left, sloped]{$1-p_0$} ;
  \draw (C') -- (D)  node[midway, below right, sloped]{$p_0$}  ;
  
\end{tikzpicture}
    \caption{This figure represents, using a Markov chain, the probability of an honest validator changing branches or not every epoch. During the attack, the Byzantine validators ensure that a proportion $p_0$ of honest validators remains on one branch so they can justify this branch later with their withheld votes (\autoref{eq:p0conditions}).}
    \label{fig:schemaAttackProba}
\end{figure}
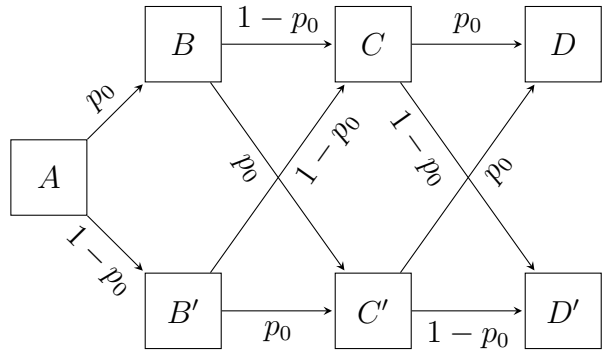

\paragraph*{Analytical Evaluation}

We are interested in the evolution of the proportion of Byzantine validators' stake $\beta$ during the attack. To examine this, we analyze the evolution of the inactivity score over time for an honest validator randomly placed at each epoch. Referring to \autoref{fig:schemaAttackProba}, we observe that after two epochs, there is a probability $p_0(1-p_0)$ of having been on branch $B$ for both epochs or on branch $A$ for both epochs. The probability of being on both branches, regardless of the order, is $p_0^2+(1-p_0)^2$. From the perspective of a chain, validators will be deemed inactive if they are active on the other chain. The probability of the inactivity score evolution after two epochs is the following:
\begin{equation}
    \left\{
    \begin{array}{ll}
        +8 :& p_0(1-p_0) \\
        +3 :& p_0^2+(1-p_0)^2 \\
        -2 :& p_0(1-p_0) \\
    \end{array}
\right.
\end{equation}

We can notice that the time-dependent probability of the inactivity score is the convolution of two random walks. The first random walk moves +4 with probability $p_0$ and -1 with probability $(1-p_0)$. The second is the opposite, moving +4 with probability $(1-p_0)$ and -1 with probability $p_0$.
We place ourselves in the continuous case to be able to continue our analysis and find the stake of validators with the inactivity score distribution over time (see \autoref{appendix:sec:discInaSco} for details on the discrete and continuous case). To do so, we use the fact that a random walk follows a Gaussian distribution when time is large, using the central limit theorem. 
The expectation of the two random walks are $(5p_0-4)t$ and $(1-5p_0)t$, respectively, with both having a standard deviation of $25p_0(1-p_0)$. We disregard here the fact that the actual inactivity score is bounded by zero for analytical tractability. Allowing for negative values in the inactivity score can result in a reward instead of a penalty, which leads to a scenario conservatively estimating the loss of stake.
The convolution of these two random walks is the probability of the inactivity score $I$:
\begin{equation}
\phi(I,t)=\frac{1}{\sqrt{4\pi Dt}}\exp \left(-\frac{(I-Vt)^2}{4Dt}\right),
\end{equation}
with $D=25p_0(1-p_0)$ and $V=3/2$.
It now remains to find the distribution function of the stake $s$.  
We rewrite here the differential equation of the stake depending on $I$ previously described in \autoref{eq:stakeDerivative}:
\begin{equation}\label{eq:stakeDerivative}
    \frac{d s}{d t} = -\frac{I(t) s}{2^{26}}.
\end{equation}

Using this (details in \autoref{appendix:sec:lognormal}) we find the distribution function of the stake $s$ to be:
\begin{equation}
    P(s,t)=\frac{2^{26}}{s\sqrt{\frac{4}{3}\pi Dt^3}}\exp\left(-\frac{(2^{26}\ln(s/32)+Vt^2/2)^2}{\frac{4}{3}Dt^3}\right),
\end{equation}
with $D$ and $V$, the diffusion and the velocity. In our case $V=3/2$ and $D=25p_0(1-p_0)$. The stake follows a log normal distribution for which the cumulative function is:
\begin{equation}
    F(s,t)=\frac{1}{2}+\frac{1}{2} \operatorname{erf}\left[\frac{2^{26}\ln (s/32)+Vt^2/2}{\sqrt{\frac{4}{3}Dt^3}}\right].
\end{equation}

Currently, the probability $P$ does not reflect the actual stake according to time since validators get ejected at $16.75$ ETH and are stuck at $32$ ETH. 
To emulate this mechanism, since we know the cumulative distribution function, we can compute the new probability law $\mathcal{P}$:
\begin{equation}    
\mathcal{P}(x,t) = \begin{cases}
F(a,t) & \text{if } x = 0, \\ 
P(x,t) & \text{if } a < x < b, \\
1-F(b,t) & \text{if } x = b,  
\end{cases}
\end{equation}
with $a=16.75$ and $b=32$. This new probability law takes into account the fact that if the stake is lower than $16.75$ ETH, it becomes 0, and it is capped at $32$ ETH.
The explicit expression of $\mathcal{P}$ reads:
\begin{equation}
\begin{split}
\mathcal{P}(x,t) = &\delta(x)\cdot F(a,t) + \delta(x-b)\cdot (1-F(b,t)) \\ 
&+ [H(x-a)\times H(b-x)] \cdot P(x,t),
\end{split}
\end{equation}
where $\delta$ is the Dirac distribution, and $H$ is the Heaviside function. 
\autoref{fig:mathcalP} shows a visual representation of the function $\mathcal{P}$.

\begin{figure}
    \centering
    \includegraphics[width=.8\linewidth]{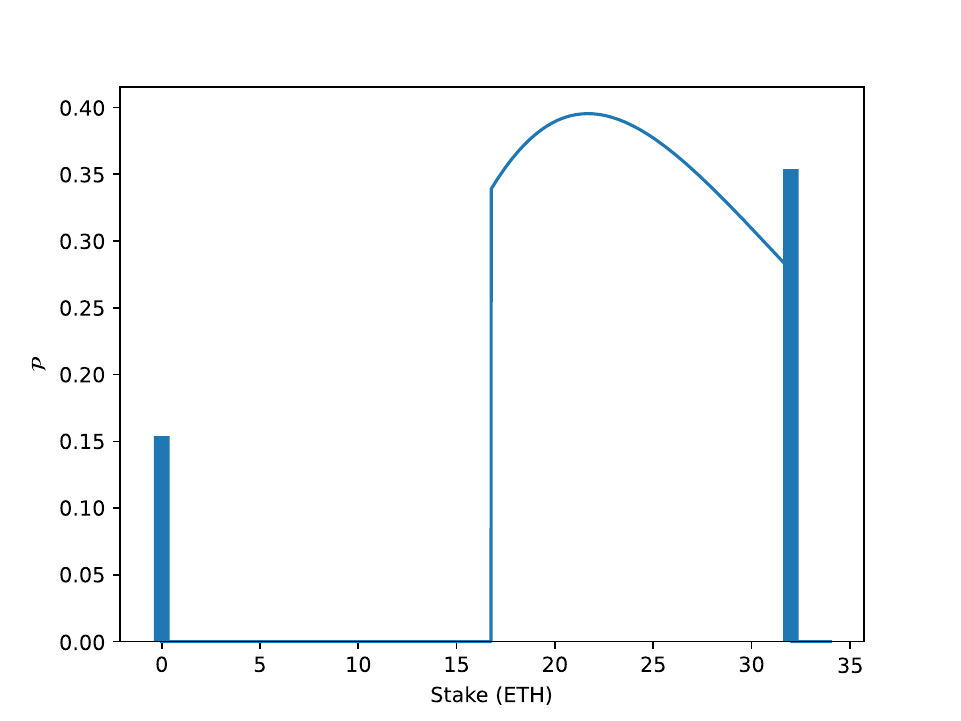}
    \caption{This is a representation of the distribution $\mathcal{P}$ at $t=4024$ with an exaggerated standard deviation to provide a better intuition of the distribution behavior.}
    \label{fig:mathcalP}
\end{figure}

The associated cumulative distribution function $\mathcal{F}$ of $\mathcal{P}$ is:
\begin{equation}
\begin{split}
    \mathcal{F}(x,t) = \; &\int_{0}^{x} \mathcal{P}(s,t) \, ds \\
    = \; &F(a,t) + H(x-a)[F(x,t)-F(a,t)] \\
    &+ H(x-b)[1-F(x,t)] .
\end{split}
\end{equation}

With this, we can evaluate the ratio of Byzantine validators and determine with what probability it will go beyond $1/3$. 
We denote by $s_{\rm B}(t)$ the stake of Byzantine validators and $s_{\rm H}(t)$ the stake of an honest validator. We are looking for the probability such that
\begin{equation}
    \beta(t) = \frac{\beta_0 s_{\rm B}(t)}{\beta_0 s_{\rm B}(t) + (1-\beta_0) s_{\rm H}(t)} > \frac{1}{3},
\end{equation}
depending on the probability of $s_{\rm H}$ that we now know. This translates into:
\begin{equation}\label{eq:CDFratioByzOverOneThird}
    \mathcal{F}\left(\frac{2\beta_0}{1-\beta_0} s_{\rm B}(t), t\right),
\end{equation}
where $s_{\rm B}(t)$, the stake of a Byzantine validator, follows the stake of a semi-active validator.

We provide a representation of \autoref{eq:CDFratioByzOverOneThird} for several values of $\beta_0$ with $p=0.5$ (note that $p_0$ has a minimal impact on the curve as it only slightly changes the variance) in \autoref{fig:CDFratioByzOverOneThird}.

The figure illustrates how the proximity of $\beta_0$ to $1/3$ can be detrimental.
This phenomenon occurs because the mean of the log-normal distribution approximates $s_{\rm B}$ when $t$ is not too large. Referring to \autoref{eq:CDFratioByzOverOneThird}, we observe that if $\beta_0=1/3$, we are examining $\mathcal{F}(s_{\rm B}(t),t)$, which explains why the probability is $0.5$.

The probability increases sharply just before the expulsion of Byzantine validators; however, it is unlikely that the probabilistic bouncing attack would persist for that long.
As an estimate, we can use the probability mentioned in the previous chapter (\autoref{eq:livenessProba}) to provide an upper bound on the probability of reaching epoch $7000$: $(1-(1-\beta_0)^8)^{7000}$ is equal to $1.01\times 10^{-121}$ for $\beta_0=1/3$. This essentially negates any strategy by Byzantine validators that would require the probabilistic bouncing attack to last that long.

However, as \autoref{fig:CDFratioByzOverOneThird} shows, with $\beta_0$ nearing 1/3, Byzantine validators realistically have a high probability of quickly exceeding $1/3$ of the stake, especially considering the significant factor of the attack occurring on two branches. This means that if a validator is active during an epoch on one branch, it is inactive on the other. Hence, the probability can be doubled for each curve. 

We can comprehend this by considering the case of $\beta_0=1/3$: after two epochs, the Byzantine validators have been active on each branch once. If one branch has more validators that have been active on it for two epochs, the other branch will have honest validators incurring, on average, more penalties than the Byzantine validators. On this latter branch, the Byzantine stake will represent more than one-third of the total stake.

\textit{These results imply that, theoretically, within the synchronous period and with a proportion of Byzantine stake sufficiently close to $1/3$ as well as a favorable initial setup, the probabilistic bouncing attack can pose a threat to the blockchain by allowing Byzantine validators to exceed the safety threshold of $1/3$.}

\begin{figure}
    \centering
    \includegraphics[width=.8\linewidth]{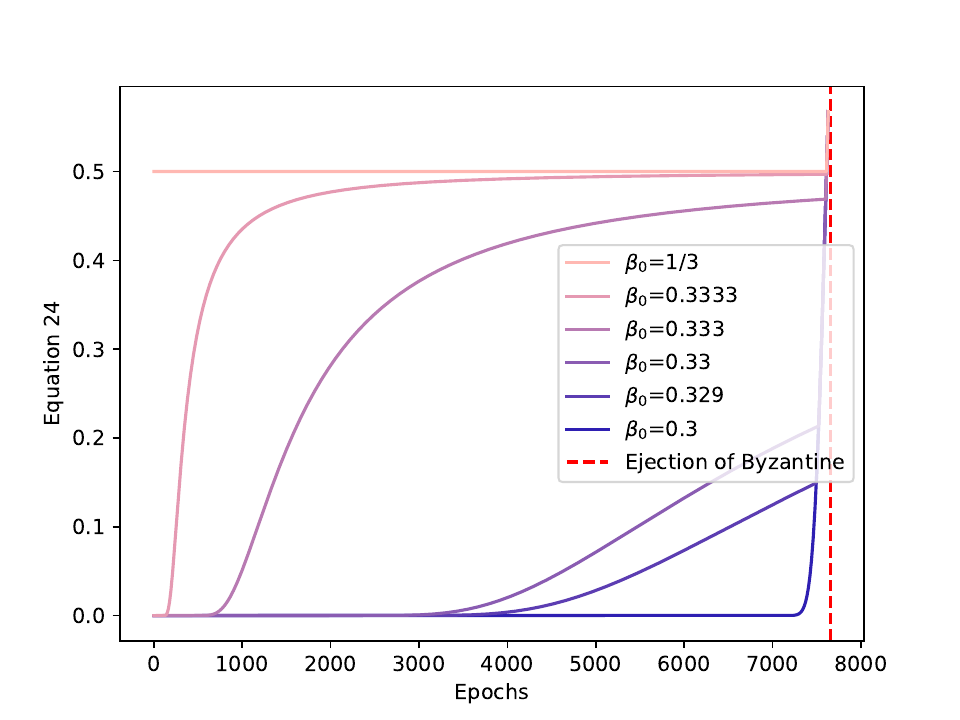}
    \caption{We represent \autoref{eq:CDFratioByzOverOneThird} according to time with various $\beta_0$.}
    \label{fig:CDFratioByzOverOneThird}
\end{figure}

\vspace{2cm}

\section{Discussion \& Conclusion}
Our work presents the first theoretical analysis of the inactivity leak, designed to restore finalization during catastrophic network failure. We highlight situations where \textit{Byzantine actions expedite the loss of Safety, either through conflicting finalization or by increasing the Byzantine proportion over the one-third Safety threshold}. Notably, we demonstrate the possibility of Byzantine validators exceeding the one-third Safety threshold even during synchronous periods.

Our findings underscore the critical role of penalty mechanisms in BFT analysis. By illuminating potential issues in the protocol design, we offer insights for future improvement and provide tools to investigate them.

    \chapter{Ethereum PoS Protocol Analysis under the Game Theoretical Model} \label{chap:Eth-GTModel-Analysis}

\minitoc

\chapterLettrine{S}{o} far, we have analyzed the protocol through the lens of distributed systems, where agents were either honest or Byzantine. Although the second analysis did take into account an important part of the incentives, the model did not consider rational agents actually driven by incentives. In this final analysis, our goal is to study the strategies and equilibrium that can arise if all agents are rational. We aim to model agents with utility functions directly linked to the protocol rewards.

We start by presenting the necessary elements of the protocol required for our study. We then proceed with the definition of our \emph{game} and its analysis.

\section{Ethereum protocol}\label{sec:ethProtocol}

We have simplified the complex functioning of the Ethereum Proof of Stake (PoS) protocol into the essential components for our analysis. At a high level, the set of participants, called validators, locally maintain a tree of blocks, denoted as $\mathcal{T}$. At any given moment, validators can evaluate the block tree to determine the branch that constitutes the current canonical chain, using a function known as the fork choice rule. The protocol requires validators to add new blocks to their local tree of blocks and broadcast these blocks to other validators.

Upon receiving a new block, validators assess whether it extends their canonical chain; if it does, a portion of them vote for it (attest to the block).

Thanks to a finalization protocol, a growing prefix of the canonical chain is maintained. This growing prefix cannot be forked, while the part of the chain beyond this prefix is forkable. In this study, we are interested in the protocol responsible for building the forkable part of the chain before finalization.

For our analysis, the elements of the protocol that we need are the following:

\begin{itemize}
    \item \textbf{Slot.} Slots are the time frames dictated by the protocol for \textit{proposers} and \textit{attesters} to perform certain actions.
    \item \textbf{Proposer.} There is one proposer per slot. The proposer's role is to propose a block during a specific slot.
    \item \textbf{Attester.} There are $a$ attesters per slot. The attester's role is to produce an attestation, which is a vote for a specific block. Attestations determine the weights of the blocks, which are used by the fork choice rule.
    \item \textbf{Fork choice rule.} The fork choice rule is the protocol's rule that determines, in case of a fork, which block is the head of the chain.
    \item \textbf{Canonical chain.} The canonical chain is the chain designated by the fork choice rule.
\end{itemize}

Proposers and attesters are assigned to slots in a deterministic and verifiable manner using a pseudo-random function included in the Ethereum protocol.

\paragraph*{Fork Choice Rule}

As mentioned, despite their name, blockchains are closer to block trees. Forks can occur, causing the blockchain to have several branches rather than a single chain. 
To address this, the protocol defines a function called the fork choice rule, $\mathcal{F}$, which indicates, at each slot $k$, on which block to build or attest based on the tree of all blocks, $\mathcal{T}_k$, and the set of attestations, $\mathcal{A}_k$. This block is called the head of the canonical chain. To determine the head of the canonical chain, the fork choice rule follows these steps:
\begin{enumerate}
    \item Traverse the set of all attestations $\mathcal{A}_k = \cup_{i=0}^k a_i$, where $a_i$ is the set of attestations sent during slot $i$. Keep only the last attestation from each attester.
    \item For each attestation, add a weight\footnote{In the protocol, the weight added is proportional to the stake of the corresponding validator. For simplicity, each validator is assumed to have the same stake in our analysis. Thus, without loss of generality, we choose the stake to be one, so that counting the attestation weight is equivalent to counting the number of attestations for this block or its descendants.} to each block attested to, as well as to all of its parents. 
    This process gives an \textit{attestation weight} to each block using the tree of blocks $\mathcal{T}_k$ and the set of attestations $\mathcal{A}_k$ at slot $k$. 
    \item Start from the genesis block and continue along the chain by following the block with the highest attestation weight at each fork. Return the block that has no children. This block is the head of the candidate chain.\footnote{We assume that the block with the highest attestation weight always starts from the genesis block. In practice, the fork choice rule begins from the justified checkpoint with the highest epoch, but we have simplified the protocol for our analysis.}
\end{enumerate}

During the execution of the protocol, it is prescribed that the block proposer of slot $k$ executes the fork choice rule $\mathcal{F}(\mathcal{T}_{k-1}, \mathcal{A}_{k-1})$ to determine the parent of its block. During the same slot, the attester should use the fork choice rule $\mathcal{F}(\mathcal{T}_k, \mathcal{A}_{k-1} + \rho a)$ to determine which block to vote for. The notation $\mathcal{A}_{s-1}+\rho a$ indicates that an additional attestation weight of $\rho a$ is added on the block of the current round. The addition of the attestation with $\rho a$ is called the proposer boost and is explained below.

It is important to note that during periods of good network conditions, all validators are likely to be aware of every block and attestation within the corresponding slot. Therefore, all validators will have the same view of the attestation weights, implying that $\mathcal{F}$ will consistently return the same head when executed by different validators. In our model, we assume 'perfect' network conditions, meaning that any message sent is assumed to be received immediately.

\paragraph*{Attestation Weight}
We previously introduced the concept of attestation weight, which is crucial throughout our analysis. Let us now briefly expand on it. The attestation weight of a block is the sum of all attestations sent for this block, as well as all attestations sent for the descendants of this block. This means that an attestation not only supports a single block but also the entire chain of blocks leading to it.

We refer to the \textit{total} attestation weight of a branch of blocks as the sum of all attestations for that branch. It should be clear that the total attestation weight of a branch is equivalent to the attestation weight of the first block in the branch.

\paragraph*{Proposer Boost}
To mitigate the balancing attack \cite{neu_ebb_2021}, in which malicious validators withhold votes and release them at an opportune time to maintain a fork indefinitely, the \textit{proposer boost} was created\footnote{See \href{https://github.com/ethereum/consensus-specs/pull/2730}{ethereum/consensus-specs/pull/2730}}. 
However, this modification of the protocol was found to be susceptible to new attacks involving equivocation \cite{neu_two_2022}, leading to an update where the fork choice rule no longer considers conflicting attestations from the same attester.

The proposer boost, denoted as $\rho \in [0,1)$, temporarily assigns $\rho a$ artificial attestations to a \emph{timely} block, where $a$ represents the total attestation weight per slot. This mechanism adds additional attestation weight to a block exclusively during the slot in which it is proposed. Specifically, if a block is received early in slot $k$ (within the first four seconds), $\rho a$ artificial attestations are temporarily added to it. 
Currently, the proposer boost is set at 0.4, effectively adding $0.4 \times a$ attestations to the block's current weight. This adjustment influences the attestation weights so that during slot $k$, the timely block $B_k$ carries additional weight, thereby affecting the fork choice rule.

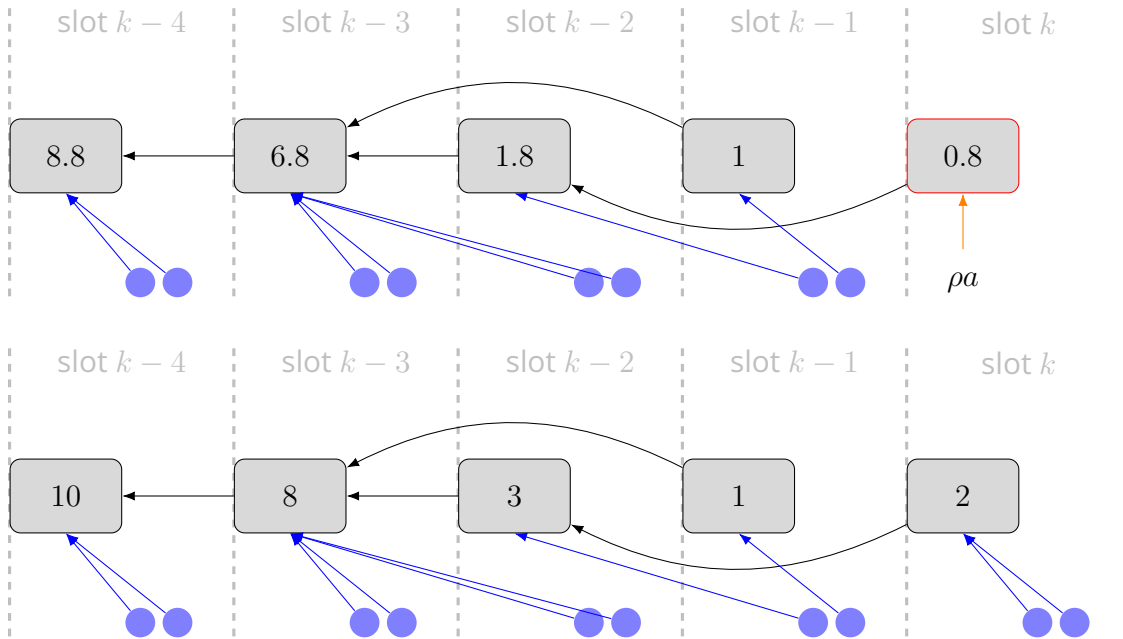
\begin{figure}[ht]
    \centering
    \resizebox{\columnwidth}{!}{%
    \begin{tikzpicture}[
    node distance=5cm and 1.5cm, 
    block/.style={rectangle, rounded corners, draw, fill=gray!30, minimum width=1.5cm, minimum height=1cm, align=center}, 
    blockOrange/.style={rectangle, rounded corners, draw, fill=orange!30, minimum width=1.5cm, minimum height=1cm, align=center},
    dashed line/.style={dashed, very thick, draw=gray!50},
    label/.style={text=gray!50, font=\normalsize}, 
    attestation/.style={Latex-, draw=blue}, 
    attester/.style={circle, fill=blue!50, minimum size=4mm, inner sep=0pt, font=\normalsize, text=white} 
]

\node[block] (block1) {$8.8$};
\node[block, right=of block1] (block2) {$6.8$};
\node[block, right=of block2] (block3) {$1.8$};
\node[block, right=of block3] (block4) {$1$};
\node[block,draw=red,right=of block4] (block5) {$0.8$};

\begin{pgfonlayer}{background}
\draw[dashed line] ($(block1.west)$) ++(0,2) -- ++(0,-4);
\draw[dashed line] ($(block2.west)$) ++(0,2) -- ++(0,-4);
\draw[dashed line] ($(block3.west)$) ++(0,2) -- ++(0,-4);
\draw[dashed line] ($(block4.west)$) ++ (0,2) -- ++(0,-4);
\draw[dashed line] ($(block5.west)$) ++ (0,2) -- ++(0,-4);
\draw[dashed line,opacity=0] ($(block5.east)$) ++ (1.5,2) -- ++(0,-4);
\end{pgfonlayer}

\node[label, above=1cm of block1, xshift=0.75cm] {slot $k-4$};
\node[label, above=1cm of block2, xshift=0.75cm] {slot $k-3$};
\node[label, above=1cm of block3, xshift=0.75cm] {slot $k-2$};
\node[label, above=1cm of block4, xshift=0.75cm] {slot $k-1$};
\node[label, above=1cm of block5, xshift=0.75cm] {slot $k$};

\draw[-Latex] (block2) -- (block1);
\draw[-Latex] (block3) -- (block2);
\draw[-Latex] (block4) to[bend right=27] (block2) ;
\draw[-Latex] (block5) to[bend left=27] (block3) ;

\node[attester, below=1cm of block1, xshift=1cm] (attest0) {};
\node[attester, below=1cm of block1, xshift=1.5cm] (attest1) {};
\node[attester, below=1cm of block2, xshift=1cm] (attest2) {};
\node[attester, below=1cm of block2, xshift=1.5cm] (attest3) {};
\node[attester, below=1cm of block3, xshift=1cm] (attest4) {};
\node[attester, below=1cm of block3, xshift=1.5cm] (attest5) {};
\node[attester, below=1cm of block4, xshift=1.5cm] (attest6) {};
\node[attester, below=1cm of block4, xshift=1cm] (attest7) {};
\node[below=1.2cm of block5, anchor=center, minimum size=4mm, circle]  (attest10) {$\rho a$};


\draw[attestation] (block1.south) -- (attest0);
\draw[attestation] (block1.south) -- (attest1);
\draw[attestation] (block2.south) -- node[sloped, above] {} (attest2) ;
\draw[attestation] (block2.south) -- node[sloped, below] {} (attest3) ;
\draw[attestation] (block2.south) -- node[sloped, below] {} (attest4) ;
\draw[attestation] (block2.south) -- node[sloped, below] {} (attest5) ;
\draw[attestation] (block4.south) -- node[sloped, below] {} (attest6) ;
\draw[attestation] (block3.south) -- node[sloped, below] {} (attest7) ;

\draw[Latex-, draw=orange] (block5.south) -- node[sloped, below] {} (attest10) ;

\end{tikzpicture}  }

\vspace{.3cm}

\resizebox{\columnwidth}{!}{%
\begin{tikzpicture}[
    node distance=5cm and 1.5cm, 
    block/.style={rectangle, rounded corners, draw, fill=gray!30, minimum width=1.5cm, minimum height=1cm, align=center}, 
    blockOrange/.style={rectangle, rounded corners, draw, fill=orange!30, minimum width=1.5cm, minimum height=1cm, align=center},
    dashed line/.style={dashed, very thick, draw=gray!50},
    label/.style={text=gray!50, font=\normalsize}, 
    attestation/.style={Latex-, draw=blue}, 
    attester/.style={circle, fill=blue!50, minimum size=4mm, inner sep=0pt, font=\normalsize, text=white} 
]

\node[block] (block1) {$10$};
\node[block, right=of block1] (block2) {$8$};
\node[block, right=of block2] (block3) {$3$};
\node[block, right=of block3] (block4) {$1$};
\node[block, right=of block4] (block5) {$2$};

\begin{pgfonlayer}{background}
\draw[dashed line] ($(block1.west)$) ++(0,2) -- ++(0,-4);
\draw[dashed line] ($(block2.west)$) ++(0,2) -- ++(0,-4);
\draw[dashed line] ($(block3.west)$) ++(0,2) -- ++(0,-4);
\draw[dashed line] ($(block4.west)$) ++ (0,2) -- ++(0,-4);
\draw[dashed line] ($(block5.west)$) ++ (0,2) -- ++(0,-4);
\draw[dashed line] ($(block5.east)$) ++ (1.5,2) -- ++(0,-4);
\end{pgfonlayer}

\node[label, above=1cm of block1, xshift=0.75cm] {slot $k-4$};
\node[label, above=1cm of block2, xshift=0.75cm] {slot $k-3$};
\node[label, above=1cm of block3, xshift=0.75cm] {slot $k-2$};
\node[label, above=1cm of block4, xshift=0.75cm] {slot $k-1$};
\node[label, above=1cm of block5, xshift=0.75cm] {slot $k$};

\draw[-Latex] (block2) -- (block1);
\draw[-Latex] (block3) -- (block2);
\draw[-Latex] (block4) to[bend right=27] (block2) ;
\draw[-Latex] (block5) to[bend left=27] (block3) ;

\node[attester, below=1cm of block1, xshift=1cm] (attest0) {};
\node[attester, below=1cm of block1, xshift=1.5cm] (attest1) {};
\node[attester, below=1cm of block2, xshift=1cm] (attest2) {};
\node[attester, below=1cm of block2, xshift=1.5cm] (attest3) {};
\node[attester, below=1cm of block3, xshift=1cm] (attest4) {};
\node[attester, below=1cm of block3, xshift=1.5cm] (attest5) {};
\node[attester, below=1cm of block4, xshift=1.5cm] (attest6) {};
\node[attester, below=1cm of block4, xshift=1cm] (attest7) {};
\node[attester, below=1cm of block5, xshift=1.5cm] (attest8) {};
\node[attester, below=1cm of block5, xshift=1cm] (attest9) {};

\draw[attestation] (block1.south) -- (attest0);
\draw[attestation] (block1.south) -- (attest1);
\draw[attestation] (block2.south) -- (attest2) ;
\draw[attestation] (block2.south) -- (attest3) ;
\draw[attestation] (block2.south) -- (attest4) ;
\draw[attestation] (block2.south) -- (attest5) ;
\draw[attestation] (block4.south) -- (attest6) ;
\draw[attestation] (block3.south) -- (attest7) ;
\draw[attestation] (block5.south) -- (attest9) ;
\draw[attestation] (block5.south) -- (attest8) ;


\end{tikzpicture} 
 }

    \caption{This figure illustrates an evaluation of the fork choice rule executed after a timely block proposal in slot $k$. In this simplified representation, there are 2 attesters per slot, so two attestations are sent per slot, each represented by a circle pointing to the block being attested. Artificial attestations are created in the current slot $k$ and add a weight of $\rho a$ for the fork choice rule.
    At the beginning of slot $k$, a timely block is proposed, and the proposer boost of $\rho a=0.8(=0.4\times 2)$ is applied. When slot $k$ ends, the proposer boost is cleared, and we observe the two attestations sent during the slot that followed the fork choice rule.
    }

    \label{fig:forkChoiceRule}
\end{figure}

An example of the fork choice rule with the proposer boost in action is illustrated in \autoref{fig:forkChoiceRule}. The figure captures the chain at two different times during slot $k$: right after the timely block proposal and at the end of the slot. Each time, we show the weight of the blocks as computed by the fork choice rule. In our study, we remain agnostic about the value of $\rho$ to examine its effects across different values.

\paragraph*{Rewards}
The rewards for proposers and attesters are computed in a verifiable manner. Based on the content of a block in the canonical chain, which includes attestations and transactions, we can determine the rewards for the attesters responsible for these attestations and for the proposer who included them. The proposer also earns rewards from transaction fees.

Validators are incentivized to participate in the process of adding new blocks and attesting them through endogenous rewards inscribed in the protocol. A crucial factor in determining rewards is identifying the canonical chain. For instance, a block proposed but not included in the canonical chain will result in zero rewards for the proposer. All rewards for block proposers and attesters can be determined by examining the content of the blocks that form the canonical chain.

Because the proposer’s rewards depend on the attesters' rewards, we first introduce the rewards for attesters.

\textbf{Attester Rewards.} An attester is rewarded for its attestation based on two factors: the \emph{timeliness} and the \emph{correctness} of its vote. \autoref{tab:attRewards} indicates the reward for an attester depending on these two factors.

\begin{table}[ht]
    \centering
    \begin{tabular}{|l|c|c|c|}
        
         \hline \text { Timeliness } & $1$ \text { slot } & $\leq 5$ \text { slots } & $\leq 64$ \text { slots } \\
\hline \text { Incorrect attestation vote} & $20x/27$ & $20x/27$ & $6x/27$ \\
\hline \text { Correct attestation vote } & $x$ & $20x/27$ & $6x/27$ \\
\hline
    \end{tabular}
    \caption{Attester's rewards based on the inclusion of the attestation in the chain and its blockvote.}
    \label{tab:attRewards}
\end{table}

Timeliness refers to the number of slots between the expected time for sending an attestation and its actual inclusion in a block. The fastest possible inclusion is 1 slot, meaning the attestation is included in the block of the subsequent slot. An attestation is considered correct if its vote points to the most recent block at the time of its slot that belongs to the canonical chain. Thus, both timeliness and correctness depend on the actions and votes of other validators. Timeliness depends on whether subsequent blocks include the attestation and eventually belong to the canonical chain. Correctness is affected by the possibility that an attester might vote for a block that is initially in the canonical chain but is later not.

As the finalized chain grows, it will eventually determine which blocks belong to the canonical chain. However, as shown in \autoref{tab:attRewards}, correctness only significantly impacts the attester's reward if the attestation is included in the following slot. This may incentivize attesters to align their votes with the proposer of the next slot, regardless of whether the block ultimately becomes part of the canonical chain.

\textbf{Proposer Rewards.} For the proposer, there are two types of rewards: rewards based on attestations and rewards based on transactions.\footnote{These rewards differ in that attestation rewards are received on the consensus layer, while transaction fees are received on the execution layer, but this distinction does not impact our analysis.} Importantly, unlike in Bitcoin, there is no coinbase transaction in each block guaranteeing a minimum reward for proposing a block. The proposer receives a proportion of the reward generated for each attestation it includes. Additionally, the proposer receives a reward for each transaction included in its block in the form of transaction fees. These rewards are formalized in \autoref{subsec:payoff}, where the utility functions are defined.

\section{Model \& Game}\label{sec:modelGame}

We model the Ethereum PoS consensus protocol as a game where each player\footnote{We use the terms \emph{players} and \emph{validators} interchangeably.} is either a \emph{proposer} or an \emph{attester}. Ideally, in the prescribed behavior, proposers propose blocks, and attesters broadcast attestations. The game evolves over $s$ sequential \emph{slots}. There is one proposer and $a \in \mathbb{N}$ attesters per slot, resulting in a total of $s$ proposers and $as$ attesters. The value of $s$ is unknown to the players.

As described in \autoref{sec:gameModel} we have two main assumptions: (i) The game occurs during a synchronous period, where the network is considered fully synchronous with no latency, and (ii) The synchronous period follows an asynchronous period, during which there may have been delays in information transmission. Therefore, our game is set in the synchronous period, but the initial state is influenced by events that occurred during the preceding asynchronous period. In this initial asynchronous period, blocks may have an uneven distribution of attestations across slots, unlike in a permanently synchronous scenario.

Under these assumptions, the broadcast of the block proposal and attestations in our game are treated as atomic events.
Thus, in each slot, there are three distinct events:
\begin{enumerate}
    \item \textit{Block proposal.} The designated proposer for the slot proposes a new block, selecting a previously existing block in the observed blockchain as its parent. 
    When a proposer prepares a block, they add all available transactions and attestations to the block—i.e., all the transactions and attestations that are not yet part of the blockchain. Once the transactions and attestations are included, the block is proposed, meaning it is sent to the network. 
    \item \textit{Generation of transactions.} Transactions are sent by users and observed by proposers and attesters. 
    \item \textit{Attestations.} After the block proposal for the slot, all attesters of the slot choose which previously proposed block to attest and send their attestations simultaneously. 
\end{enumerate}
As detailed in \autoref{sec:ethProtocol}, the protocol prescribes that proposers (and attesters) should select as the parent (or as the block to attest) the block identified by the fork choice rule, i.e., the head of the canonical chain.

These three events occur in sequence and are depicted in \autoref{fig:slotEvents}. Proposers and attesters are financially motivated to participate in the protocol. It remains to be seen whether the protocol is incentive compatible; it is the case if following the protocol maximize their gains.

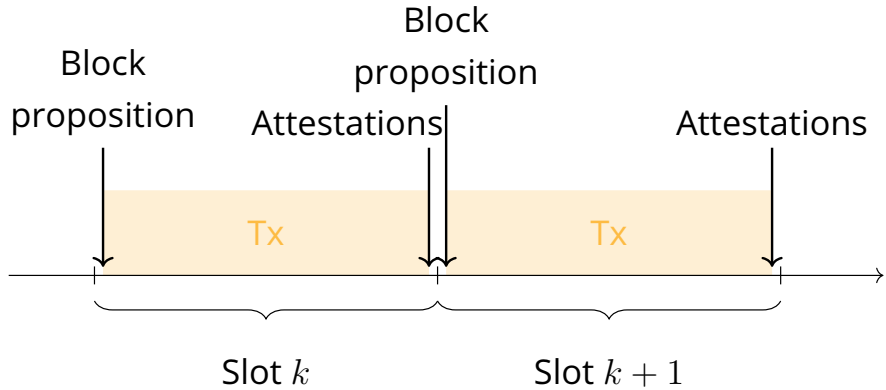
\begin{figure}[ht]
    \centering
    
    \resizebox{.8\textwidth}{!}{
\begin{tikzpicture}
    
\coordinate (O) at (-1,0); 
\coordinate (S0) at (0,-.3);
\coordinate (S1) at (4,-.3);
\coordinate (S2) at (8,-.3);
\coordinate (TX1) at (0.1,0);
\coordinate (TX2) at (3.9,0);
\coordinate (TX3) at (4.1,0);
\coordinate (TX4) at (7.9,0);
\coordinate (BP) at (.1,0);
\coordinate (BP2) at (4.1,0);
\coordinate (ATT) at (3.9,0);
\coordinate (ATT2) at (7.9,0);
\coordinate (F) at (9.2,0); 

\fill[color=Dandelion!20] rectangle (TX1) -- (TX2) -- ($(TX2)+(0,1)$) -- ($(TX1)+(0,1)$); 
\fill[color=Dandelion!20] rectangle (TX3) -- (TX4) -- ($(TX4)+(0,1)$) -- ($(TX3)+(0,1)$); 

\draw ($.5*(TX1)+.5*(TX2)+(0,0.5)$) node[Dandelion] {Tx};

\draw ($.5*(TX3)+.5*(TX4)+(0,0.5)$) node[Dandelion] {Tx};

\draw[<-,thick,color=black] ($(BP)+(0,0.1)$) -- ($(BP)+(0,1.5)$) node [above=0pt,align=center,black] {Block\\proposition};
\draw[<-,thick,color=black] ($(BP2)+(0,0.1)$) -- ($(BP2)+(0,2)$) node [above=0pt,align=center,black] {Block\\proposition};
\draw[<-,thick,color=black] ($(ATT)+(0,0.1)$) -- ($(ATT)+(0,1.5)$) node [above=0pt,align=center,black,xshift=-0.95cm] {Attestations};
\draw[<-,thick,color=black] ($(ATT2)+(0,0.1)$) -- ($(ATT2)+(0,1.5)$) node [above=0pt,align=center,black] {Attestations};
\draw [decorate,decoration={brace,amplitude=6pt,mirror}]($(S0)$) -- ($(S1)$) node [black,midway,below=15pt] {Slot $k$};
\draw [decorate,decoration={brace,amplitude=6pt,mirror}]($(S1)$) -- ($(S2)$) node [black,midway,below=15pt] {Slot $k+1$};

\draw[->] (O) -- (F);
\foreach \x in {0,4,8}
\draw(\x cm,3pt) -- (\x cm,-3pt);

\end{tikzpicture}
    }
    \caption{This schema represents the atomicity of the block proposal and the attestation broadcast. In each slot, three phases occur in order: first, the block proposal; then, the generation of transactions (which will be available for the proposer of the next slot); and finally, all attestations are sent simultaneously.}
    \label{fig:slotEvents}
\end{figure}

\subsection*{The game}\label{subsec:theGame}

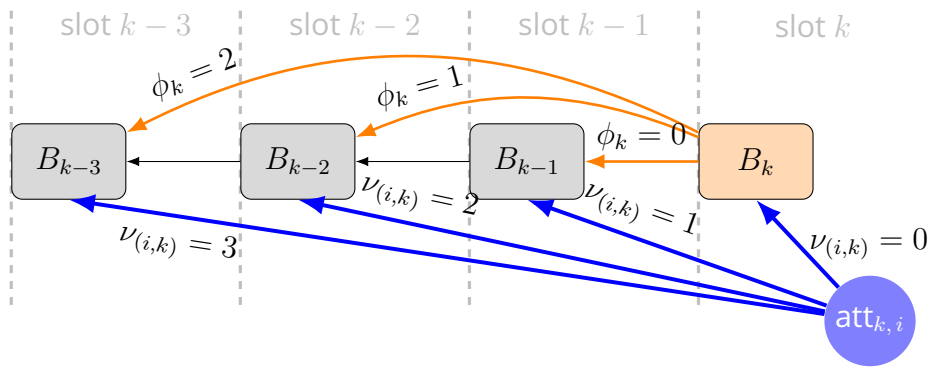
\begin{figure}[ht]
    \centering
    \begin{tikzpicture}[
    node distance=5cm and 1.5cm, 
    block/.style={rectangle, rounded corners, draw, fill=gray!30, minimum width=1.5cm, minimum height=1cm, align=center}, 
    blockOrange/.style={rectangle, rounded corners, draw, fill=orange!30, minimum width=1.5cm, minimum height=1cm, align=center},
    dashed line/.style={dashed, very thick, draw=gray!50},
    label/.style={text=gray!50, font=\normalsize}, 
    attestation/.style={Latex-, thick, draw=blue, line width=1.5pt}, 
    attester/.style={circle, fill=blue!50, minimum size=12mm, inner sep=0pt, font=\normalsize, text=white} 
]

\node[block] (block1) {$B_{k-3}$};
\node[block, right=of block1] (block2) {$B_{k-2}$};
\node[block, right=of block2] (block3) {$B_{k-1}$};
\node[blockOrange, right=of block3] (block4) {$B_{k}$};

\node[label, above=1cm of block1, xshift=0.75cm] {slot $k-3$};
\node[label, above=1cm of block2, xshift=0.75cm] {slot $k-2$};
\node[label, above=1cm of block3, xshift=0.75cm] {slot $k-1$};
\node[label, above=1cm of block4, xshift=0.75cm] {slot $k$};

\draw[-Latex] (block2) -- (block1);
\draw[-Latex] (block3) -- (block2);
\draw[-Latex, draw=orange, line width=1pt] (block4) to[bend right=27] node[sloped,above,pos=0.87] {$\phi_{k}=2$} (block1) ;
\draw[-Latex, draw=orange, line width=1pt] (block4) to[bend right=23] node[sloped,above,pos=0.8] {$\phi_{k}=1$} (block2) ;
\draw[-Latex, draw=orange, line width=1pt] (block4) -- node[above] {$\phi_{k}=0$} (block3);

\begin{pgfonlayer}{background}
\draw[dashed line] ($(block1.west)$) ++(0,2) -- ++(0,-4);
\draw[dashed line] ($(block2.west)$) ++(0,2) -- ++(0,-4);
\draw[dashed line] ($(block3.west)$) ++(0,2) -- ++(0,-4);
\draw[dashed line] ($(block4.west)$) ++ (0,2) -- ++(0,-4);
\end{pgfonlayer}

\node[attester, below=1cm of block4, xshift=1.5cm] (attest) {att\textsubscript{$k,i$}};
\draw[attestation] (block4.south) -- (attest) node[midway, right] {$\nu_{(i,k)}=0$};
\draw[attestation] (block3.south) -- node[sloped, above, pos=0.35] {$\nu_{(i,k)}=1$} (attest) ;
\draw[attestation] (block2.south) -- node[sloped, above, pos=0.22] {$\nu_{(i,k)}=2$} (attest) ;
\draw[attestation] (block1.south) -- node[sloped, below, pos=0.15] {$\nu_{(i,k)}=3$} (attest) ;

\end{tikzpicture}

    \caption{The actions available to an attester (blue) and a proposer (orange). The attester $i$ selects which block to attest with $\nu_{(i,k)}$. The proposer selects the parent of its block with $\phi_{k}$. An example is shown where $\phi_{k} = \phi_{k-1} = \phi_{k-2} = 0$.}
    \label{fig:gameActions}
\end{figure}

We denote the set of players (the validators) as $\mathcal{V} = \{P, A\}$, consisting of a set of proposers $P$ and a set of attesters $A$. Per slot, following the Ethereum protocol, there is exactly one proposer and $a \in \mathbb{N}$ attesters. Hence, the number of proposers is $|P| = s$, and the number of attesters is $|A| = as$, with $a, s \in \mathbb{N}$.

We model the interactions between proposers and attesters during $s$ slots in Ethereum PoS as a game. In each slot, the timeline of events is as follows: (i) a block is proposed at the beginning of the slot,\footnote{Every block thus receives the proposer boost in our model.} (ii) new transactions are proposed, and (iii) all the attesters of the slot send their attestations simultaneously. 
Therefore, the game is dynamic, with each stage corresponding to a slot. In each slot, the attesters play a simultaneous game following the proposal by the slot's proposer.
Our interest lies in the actions that the proposers and attesters have, which we now describe.


\textbf{Actions.}

When it is their turn (recall that each proposer/attester is uniquely assigned to a slot, and this information is verifiable, allowing players to take action only in their corresponding slot), a proposer must choose which block to extend, and an attester must select which block to attest. 
The action will take the form of a variable that indicates how many slots prior a proposer attaches its block to, or an attester attests.
More formally, the action of the proposer in slot $k$ is to assign a value to its variable $\phi_k \in \mathbb{N}$, corresponding to the difference between the current slot and the slot of the block selected as the parent. 
Similarly, after the block proposal in slot $k$, each attester $i$ of slot $k$ must assign a value to its variable $\nu_{i, k} \in \mathbb{N}$ that represents the difference between the current slot and the slot of the block being attested. 
We depict a subset of the action space in \autoref{fig:gameActions}. In more detail:
\begin{itemize}

    \item At the beginning of slot $k$, a proposer $p$ chooses the parent of its block $B_{k}$. We denote this action by $\phi_{k}\in \mathbb{N}$. $\phi_{k}=l$ means that $B_{k}$'s parent is $B_{k-1-l}$. Thus, if $B_{k}$'s parent is the block from the previous slot $k-1$, then $\phi_{k}=0$.

    The blocks contain two types of data: attestations and transactions. There is no limitation on the number of transactions and attestations a block can contain.\footnote{This simplification is similar to the one made in \cite{carlsten_instability_2016}.} A proposer always includes all available transactions and attestations. A transaction/attestation is considered available if it is not included in any of the predecessors of $B_{k}$. 

    \item After the block proposal in slot $k$, all attesters of slot $k$ simultaneously choose which block to attest. The attestation of attester $i$ in slot $k$ points to a specific block determined by $\nu_{(i,k)}\in\mathbb{N}$, i.e., the age of the block attested. $\nu_{(i,k)}=l$ means that the block $B_{k-l}$ is the one attested by attester $i$ in slot $k$. Thus, if validator $i$ attests for the block in the current slot $k$, then $\nu_{(i,k)}=0$.

\end{itemize}
These actions are repeated in each slot. Note that not proposing or attesting to a block is not an available action.

The last piece of data needed for our study is to determine whether a block $B_k$ eventually belongs to the canonical chain. In our model, this information is represented by $\chi_k \in \{0,1\}$, where $\chi_k = 1$ if the block from slot $k$ eventually belongs to the canonical chain, and $\chi_k = 0$ otherwise. This information becomes known at the end of slot $s$, which marks the conclusion of our game.

It is important to note that for any slot $k$, always assigning a value of $0$ to $\phi_k$ as a proposer (or a value of $0$ to $\nu_{(i,k)}$ as an attester) is not the prescribed action. The prescribed action is to follow the fork choice rule, as illustrated in \autoref{fig:forkChoiceRule}.

\textbf{Strategies.} 
A strategy of a player $i$ is a function $\sigma_i$, which takes as input the entire tree of blocks in the blockchain, as well as the attestations sent so far, and produces as output a number, say $s \in \mathbb{N}$. Since the only information available are the tree of blocks and the attestations, the signature of a player’s strategy is $\mathcal{T} \times \mathcal{A} \to \mathbb{N}$, where $\mathcal{T}$ is the set of blocks and $\mathcal{A}$ is the set of available attestations.

For the proposer of slot $k$, the prescribed strategy is $\sigma_{(0,k)} (\mathcal{T}_{k-1},\mathcal{A}_{k-1}) = l$, where $\mathcal{F}(\mathcal{T}_{k-1},\mathcal{A}_{k-1})=B_{k-1-l}$. The prescribed strategy for an attester $i$ at slot $k$ is $\sigma_{(i,k)} (\mathcal{T}_{k},\mathcal{A}_{k-1}) = l$, where $\mathcal{F}(\mathcal{T}_{k},\mathcal{A}_{k-1}+\rho a)=B_{k-l}$. We say that a player deviates from the prescribed protocol when their strategy produces a number different from the slot of the block resulting from the fork choice rule.

A strategy profile $\sigma = (\sigma_{0,1}, \dots, \sigma_{a,1}, \sigma_{0,2}, \dots , \sigma_{a,2}, \dots, \sigma_{0,s}, \dots, \sigma_{a, s})$ is a vector where each component is a strategy of the corresponding player. We denote by $\mathcal{S}$ the set of all strategy profiles and by $\mathcal{S}_{(i,k)}$ the set of strategies for the player of component $(i,k)$. 
In this notation, players with indices $(i,k)$ where $i=0$ are proposers, while players with indices $(i,k)$ where $1 \leq i \leq a$ are attesters. For clarity, we denote by $(\sigma_{-i},\sigma'_i)$ the strategy profile $\sigma$ where, instead of playing with strategy $\sigma_i$, player $i$ deviates and uses strategy $\sigma'_i$ instead. This applies to both attesters and proposers.

It remains to define the reward of the players at the end of the game. 
At the end of slot $s$, the payoff of all players is computed and given by the function $u: \mathcal{S} \to \mathbb{R}^{n+an}$ (defined in \autoref{subsec:payoff}). The payoff of each player is given by its component in the reward vector, which depends on its type and is determined by a reward function. In the remainder of the chapter, for clarity, for any strategy profile $\sigma$, we write $u_{i,j}(\sigma)$ instead of $u(\sigma)_{(i,k)}$, where $u_{i,j}(\sigma)$ represents the payoff of player $(i,j)$, and player $(0,j)$ is the proposer of slot $j$.

\subsection{Payoff}\label{subsec:payoff}
Attesters' rewards vary depending on when their attestations are included in a block and which block they attest to. This can incentivize them to align their attestations with the behavior of future block proposers. 
Block proposers have a clear incentive to accumulate the maximum transaction fees and lucrative attestations to maximize their rewards. 
One strategy to achieve this is to fork the chain and include in the new block all the attestations and transactions that do not belong to the new chain.
However, if the block does not end up in the canonical chain, the block proposer will not receive any rewards. This incentivizes the proposer to consider other behaviors, as we will see.

Given a strategy profile $\sigma$, the reward of attester $i$ in slot $k$, player $(i,k)$, depends on a variable $x > 0$ set by the protocol and the slot in which the attestation is subsequently included in a block: 
\begin{equation}\label{eq:utilityAttester}
    u_{(i,k)}(\sigma) = \begin{cases}
                             x & \text{if } \sigma_{(i,k)} \text{ sets } \nu_{(i,k)}=\phi_{k+1} \text{ and } \chi_{k+1}=1,\\
                             20x/27 & \text{if } \chi_{k+2 \geq \dots \geq k+5}=1 \text{, included in 2 to 5 slots following the attestation,} \\
                             6x/27=2x/9 & \text{otherwise (if } \chi_{\geq k+6}=1\text{).}
                        \end{cases}
\end{equation}

Here, $x > 0$ and $\chi_{n+1}$ denote the fact that the block of slot $s+1$ belongs to the canonical chain. 
The rewards for the attester can be understood as follows: they are maximized when the attestation votes for the parent of the block in the subsequent slot, and this block in the subsequent slot ends up in the canonical chain.

The actual rewards for an attester are detailed in \autoref{tab:attRewards}. Note that the reward is influenced by the correctness of the attestation only if it is included in the block of the next slot. Additionally, if the block of slot $k+1$ does not end up in the canonical chain, the attesters of slot $k$ can never receive the maximum reward.

For the proposer of slot $k$, the reward function is given by:

\begin{equation}\label{eq:utilityProposer}
     u_{(0,k)}(\sigma) = \chi_{k} \sum_{j=n-\phi_{k}}^n  \left(\frac{1}{7}\sum_{i=1}^a u_{(i,j)}(\sigma)  + f_{j-1}  \right),
\end{equation}

The reward is the sum of attestation rewards and transaction fees over the slots separating the block from its parent, multiplied by the factor that indicates the block belongs to the canonical chain. 
Here, $f_{n-1} > 0$ represents the random value of transaction fees generated during slot $s-1$. The transaction fees are the incentives that motivate proposers to include transactions in their block.
The proposer receives $1/7$ of what the attesters receive for their attestations being included in a block.
The factor $\chi_s$, indicating whether the block ends up in the canonical chain, applies to the entire reward since, if the block is not included in the finalized chain, it does not yield any rewards.

\section{Analysis}\label{sec:Analysis}

In this section, we explore a set of possible strategies for proposers and attesters. Each can either follow the obedient strategy or adopt a cunning strategy. The obedient strategy is the one prescribed by the protocol. In contrast, the cunning strategy may deviate from the protocol while exploiting the proposer boost as a means to remain part of the canonical chain. We begin by introducing the game-theoretic preliminaries necessary for our analysis.

\subsection{Preliminaries}\label{subsec:preliminaries}
To ensure clarity and self-containment, we redefine well-known game-theoretic concepts. These concepts are useful for categorizing strategy equilibria and exploring possible states of the game.

\begin{definition}[Best response]\label{def:BestRep}
A strategy $\sigma_i^*$ is a best response for player $i$ to the strategy profile $\sigma_{-i}$ of the other players if:
\begin{equation}
u_i(\sigma_{-i},\sigma_i^*) \geq u_i( \sigma_{-i}, \sigma_i), \quad \forall \sigma_i \in \mathcal{S}_i,
\end{equation}
where $u_i$ is the payoff function for player $i$, $\sigma_{-i}$ is the strategy profile of all other players, and $\mathcal{S}_i$ is the set of all possible strategies for player $i$.
\end{definition}


\begin{definition}[Nash equilibrium]\label{def:NashEq}
A strategy profile $\sigma^* = (\sigma_1^*, \sigma_2^*, \dots, \sigma_n^*)$ is a Nash equilibrium if each player's strategy $\sigma_i^*$ is a best response to the strategies $\sigma_{-i}^*$ of the other players. Formally,
\begin{equation}
u_i(\sigma_{-i}^*,\sigma_i^*) \geq u_i( \sigma_{-i}^*,\sigma_i), \quad \forall \sigma_i \in \mathcal{S}_i \text{ and for all players } i,
\end{equation}

where $u_i$ is the payoff function for player $i$, $\sigma_{-i}^*$ is the strategy profile of all other players in the equilibrium, and $\mathcal{S}_i$ is the set of all possible strategies for player $i$.
\end{definition}

In summary, the concept of a best response helps identify the optimal strategy for a player given the strategies of the other players. 
Nash equilibrium defines a state where each player's strategy is a best response to the strategies of the other players, ensuring no player can benefit from unilaterally changing their strategy.

\subsection{Obedient}\label{subsec:obedientValidator}
As Carlsten et al. \cite{carlsten_instability_2016}, we first describe the strategy of proposers and attesters that act as prescribed by the protocol, we refer to them as obedient. However in the case of Ethereum, the actions prescribed by the fork choice rule are more complex compared to those described by Carlsten et al.

\begin{stratProposer}
\textbf{Obedient Proposer} ($\sigma^{O}_{(0,k)}$):\\
\textit{Action:} \quad $\phi_k = l$, where $\mathcal{F}(\mathcal{T}_{k-1},\mathcal{A}_{k-1}) \rightarrow B_{k-1-l}$.\\

The strategy of an obedient proposer at slot $k$ is to propose a block $B_k$ linked to the block designated by the fork choice rule $\mathcal{F}(\mathcal{T}_{k-1},\mathcal{A}_{k-1}) \rightarrow B_{k-1-l}$.
\end{stratProposer}

\begin{stratAttester}
\textbf{Obedient Attester} ($\sigma^{O}_{(i,k)}$):\\
\textit{Action:} \quad $\nu_{(i,k)} = l$ (Block attested is $B_{k-l}$.)\\

The obedient attester strategy of attester $i$ is to attest to the block designated by the fork choice rule $\mathcal{F}(\mathcal{T}_{k},\mathcal{A}_{k-1}+\rho a) \rightarrow B_{k-l}$.
\end{stratAttester}

We denote by $\sigma_{(i,j)}^O$ the obedient strategy of player $(i,j)$ and by $\sigma^O$ the strategy profile where all players act obediently.

When proposers and attesters follow the obedient strategy, we can evaluate the rewards each of them will receive. Since they will all follow the fork choice rule and there are no delays, no forks will occur, and attesters will attest to the block of their slot. Moreover, each attestation will be included in the following slot and will be correct. For proposers and attesters following the actions prescribed by the protocol, the rewards are as follows:
\begin{itemize}
    \item For each attester $i$ following the obedient strategy, the reward is: $u_{(i,k)}(\sigma^{O}) = x$, where $\sigma^{O}$ is the strategy profile in which all proposers and attesters are obedient.
    \item For the proposer of slot $k$, the reward is:
    $u_{(0,k)}(\sigma^{O}) = \frac{ax}{7} + f_{k-1}$.
\end{itemize}
With this strategy profile, attesters obtain the maximum reward attainable (\autoref{eq:utilityAttester}). However, there is no maximum reward for a block proposer, as their rewards increase the more ancient their block's parent is.

\subsection{Cunning Strategy}\label{subsec:cunningValidator}

We now examine a strategy that could yield more rewards for validators than simply following the protocol. In some situations, deviating from the protocol can allow validators to accumulate more rewards without incurring penalties. We refer to this as the \textit{cunning} behavior. For a proposer, the strategy involves choosing a block parent for its proposal that maximizes its rewards.

As the block parent's slot is further away from the new block, the proposer can include more transactions and attestations to increase its rewards. The ideal block parent, in theory, would be the genesis block. However, for the block to actually yield rewards, it must become part of the canonical chain. The cunning proposer will always propose a block that is considered the head of the canonical chain during its slot.

For instance, a cunning proposer will not strictly follow the fork choice rule to determine its block's parent. Instead, it will subtly test whether it can choose an older block as the parent while still having its block become the head of the canonical chain. The block that maximizes rewards—typically the oldest possible block—will be selected as the parent by the cunning proposer.

\begin{stratProposer}
\textbf{Cunning Proposer} ($\sigma^{C}_{(0,k)}$):\\
\textit{Action:} \quad $\phi_k = \max \{ x \in \mathbb{N} : \mathcal{F}(\mathcal{T}_{x},\mathcal{A}_{x-1}+\rho a) = B_k \}$ \\

The cunning proposer's block extends the block that leaves the most available transactions and attestations while still being selected as the head of the canonical chain by the fork choice rule (for attesters in the same slot) due to the proposer boost $\rho a$.
\end{stratProposer}

We denote by $\sigma_{(0,k)}^C$ the cunning strategy of the proposer in slot $k$.

\begin{remark}
    The obedient and the cunning strategy can result in the same action.
\end{remark}

It is important to note that while the cunning and obedient strategies are distinct, the actions resulting from them can sometimes be identical. Indeed, the action taken by a cunning proposer is to propose a block with the oldest possible parent while still ensuring the block is designated by the fork choice rule for the attesters of the slot. However, if the oldest possible parent is the same block initially designated by the fork choice rule, the action will align with the protocol, just as it would under the obedient strategy. In this sense, we say that a cunning player \textit{acts} obediently if their action is the one prescribed by the protocol. Conversely, we say they act cunningly if the action taken differs from what is prescribed by the protocol, making the cunning strategy truly distinct from the obedient strategy.

\begin{observation}[Cunning condition] \label{obs:cunningCondition}
The divergence between cunning and obedient proposer behavior occurs when the branch containing the block designated by the fork choice rule, with a total attestation weight $w_f$, has a concurrent branch with a total attestation weight $w_g$ such that:
\begin{equation} \label{eq:2branchesClose}
    w_f - w_g \leq \rho a. 
\end{equation}
We call this inequality the \textbf{cunning condition}.
\end{observation}

First, it is clear that $w_f$ is always greater than $w_g$, as the block designated by the fork choice rule is on the branch with a total attestation weight of $w_f$. To understand the cunning condition, we consider two illustrations:
\begin{enumerate}
    \item The first, and less intuitive, case is presented in \autoref{fig:cunningProposer}. This showcases the scenario where $w_g = 0$. A branch can consist of many blocks, a single block, or, in this case, no block at all.
    
    The newly proposed block can become the head of the canonical chain by attaching itself to the first block with more than $\rho a$ attestation weight. If the block designated by the fork choice rule is on a branch with a total attestation weight less than $\rho a$, the cunning behavior differs from the obedient behavior.
    \item Another representation of the cunning condition is shown in \autoref{fig:cunningProposer2}. Here, there are two distinct branches, each with one block. 
    The branch of the block designated by the fork choice rule has an attestation weight of $w_f = 3$, while the concurrent branch has $w_g = 2$. In this case, where $\rho a = 1.2$, the condition is met, allowing the proposer to act cunningly.
\end{enumerate}

\begin{figure}[ht]
    \centering
    
\begin{tikzpicture}[
    node distance=5cm and 1.5cm, 
    block/.style={rectangle, rounded corners, draw, fill=gray!30, minimum width=1.5cm, minimum height=1cm, align=center}, 
    blockOrange/.style={rectangle, rounded corners, draw, fill=orange!30, minimum width=1.5cm, minimum height=1cm, align=center},
    dashed line/.style={dashed, very thick, draw=gray!50},
    label/.style={text=gray!50, font=\normalsize}, 
    attestation/.style={Latex-, thick, draw=blue}, 
    attester/.style={circle, fill=blue!50, minimum size=4mm, inner sep=0pt, font=\normalsize, text=white} 
]

\node[block] (block1) {$9$};
\node[block, right=of block1] (block2) {$5$};
\node[block, right=of block2] (block3) {$1$};

\begin{pgfonlayer}{background}
\draw[dashed line] ($(block1.west)$) ++(0,2) -- ++(0,-4);
\draw[dashed line] ($(block2.west)$) ++(0,2) -- ++(0,-4);
\draw[dashed line] ($(block3.west)$) ++(0,2) -- ++(0,-4);
\draw[dashed line] ($(block4.west)$) ++ (0,2) -- ++(0,-4);
\end{pgfonlayer}

\node[label, above=1cm of block1, xshift=0.75cm] {slot $k-3$};
\node[label, above=1cm of block2, xshift=0.75cm] {slot $k-2$};
\node[label, above=1cm of block3, xshift=0.75cm] {slot $k-1$};
\node[label, above=1cm of block4, xshift=0.75cm] {slot $k$};

\draw[-Latex] (block2) -- (block1);
\draw[-Latex] (block3) -- (block2);

\node[attester, below=1cm of block1, xshift=1cm] (attest0) {};
\node[attester, below=1cm of block1, xshift=1.5cm] (attest1) {};
\node[attester, below=1cm of block1, xshift=2cm] (attest2) {};
\node[attester, below=1cm of block2, xshift=1cm] (attest3) {};
\node[attester, below=1cm of block2, xshift=1.5cm] (attest4) {};
\node[attester, below=1cm of block2, xshift=2cm] (attest5) {};
\node[attester, below=1cm of block3, xshift=1cm] (attest6) {};
\node[attester, below=1cm of block3, xshift=1.5cm] (attest7) {};
\node[attester, below=1cm of block3, xshift=2cm] (attest8) {};

\draw[attestation] (block1.south) -- (attest0);
\draw[attestation] (block1.south) -- (attest1);
\draw[attestation] (block1.south) -- (attest2);
\draw[attestation] (block1.south) -- (attest3);
\draw[attestation] (block2.south) -- (attest4);
\draw[attestation] (block2.south) -- (attest5);
\draw[attestation] (block2.south) -- (attest6);
\draw[attestation] (block2.south) -- (attest7);
\draw[attestation] (block3.south) -- (attest8);
\draw[attestation] (block3.south) -- (attest8);

\end{tikzpicture} 

\vspace{.3cm}

\begin{tikzpicture}[
    node distance=5cm and 1.5cm, 
    block/.style={rectangle, rounded corners, draw, fill=gray!30, minimum width=1.5cm, minimum height=1cm, align=center}, 
    blockOrange/.style={rectangle, rounded corners, draw, fill=orange!30, minimum width=1.5cm, minimum height=1cm, align=center},
    dashed line/.style={dashed, very thick, draw=gray!50},
    label/.style={text=gray!50, font=\normalsize}, 
    attestation/.style={Latex-, thick, draw=blue}, 
    attester/.style={circle, fill=blue!50, minimum size=4mm, inner sep=0pt, font=\normalsize, text=white} 
]

\node[block] (block1) {$10.2$};
\node[block, right=of block1] (block2) {$6.2$};
\node[block, right=of block2] (block3) {$1$};
\node[blockOrange, right=of block3] (block4) {$1.2$};

\begin{pgfonlayer}{background}
\draw[dashed line] ($(block1.west)$) ++(0,2) -- ++(0,-4);
\draw[dashed line] ($(block2.west)$) ++(0,2) -- ++(0,-4);
\draw[dashed line] ($(block3.west)$) ++(0,2) -- ++(0,-4);
\draw[dashed line] ($(block4.west)$) ++ (0,2) -- ++(0,-4);
\end{pgfonlayer}

\node[label, above=1cm of block1, xshift=0.75cm] {slot $k-3$};
\node[label, above=1cm of block2, xshift=0.75cm] {slot $k-2$};
\node[label, above=1cm of block3, xshift=0.75cm] {slot $k-1$};
\node[label, above=1cm of block4, xshift=0.75cm] {slot $k$};

\draw[-Latex] (block2) -- (block1);
\draw[-Latex] (block3) -- (block2);
\draw[-Latex, draw=orange, line width=1pt] (block4) to[bend right=23] node[sloped,above,pos=0.8] {$\phi_{k}=1$} (block2) ;

\node[attester, below=1cm of block1, xshift=1cm] (attest0) {};
\node[attester, below=1cm of block1, xshift=1.5cm] (attest1) {};
\node[attester, below=1cm of block1, xshift=2cm] (attest2) {};
\node[attester, below=1cm of block2, xshift=1cm] (attest3) {};
\node[attester, below=1cm of block2, xshift=1.5cm] (attest4) {};
\node[attester, below=1cm of block2, xshift=2cm] (attest5) {};
\node[attester, below=1cm of block3, xshift=1cm] (attest6) {};
\node[attester, below=1cm of block3, xshift=1.5cm] (attest7) {};
\node[attester, below=1cm of block3, xshift=2cm] (attest8) {};
\node[below=1cm of block4, minimum size=4mm, circle] (attest9) {$\rho a$};

\draw[attestation] (block1.south) -- (attest0);
\draw[attestation] (block1.south) -- (attest1);
\draw[attestation] (block1.south) -- (attest2);
\draw[attestation] (block1.south) -- (attest3);
\draw[attestation] (block2.south) -- (attest4);
\draw[attestation] (block2.south) -- (attest5);
\draw[attestation] (block2.south) -- (attest6);
\draw[attestation] (block2.south) -- (attest7);
\draw[attestation] (block3.south) -- (attest8);
\draw[attestation] (block3.south) -- (attest8);
\draw[Latex-, draw=orange] (block4.south) -- (attest9) ;


\end{tikzpicture}
    \caption{Cunning proposer $(0,k)$ deviating from the prescribed protocol with $\rho a = 1.2 (= 0.4 \times 3)$ when the block designated by the fork choice rule has an attestation weight less than $\rho a$.}
    \label{fig:cunningProposer}
\end{figure}
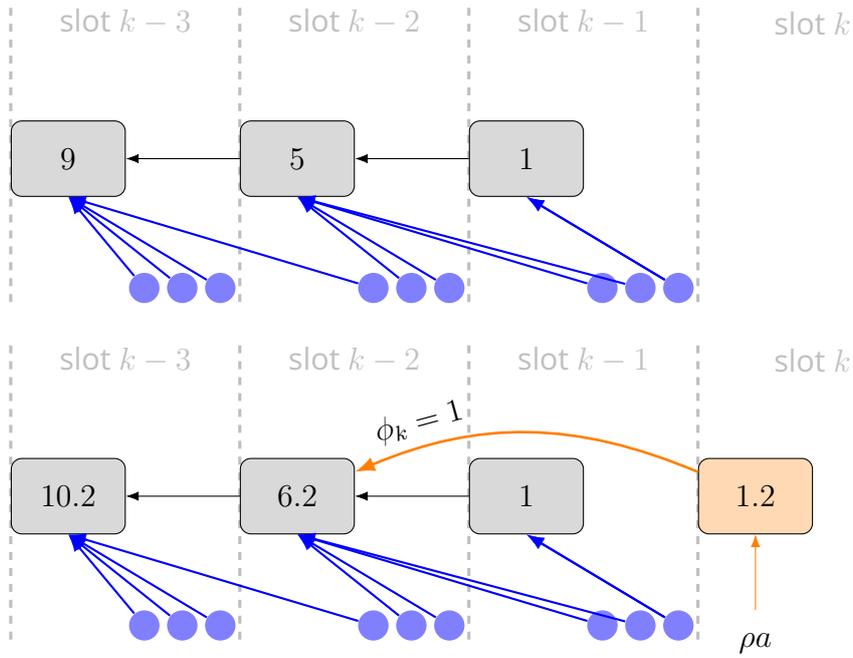

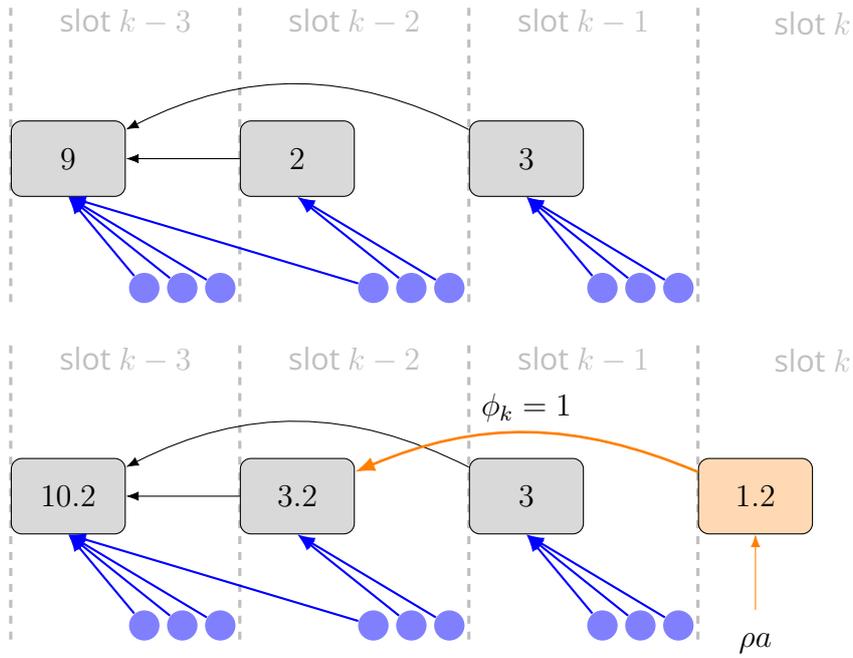
\begin{figure}[ht]
    \centering
    
\begin{tikzpicture}[
    node distance=5cm and 1.5cm, 
    block/.style={rectangle, rounded corners, draw, fill=gray!30, minimum width=1.5cm, minimum height=1cm, align=center}, 
    blockOrange/.style={rectangle, rounded corners, draw, fill=orange!30, minimum width=1.5cm, minimum height=1cm, align=center},
    dashed line/.style={dashed, very thick, draw=gray!50},
    label/.style={text=gray!50, font=\normalsize}, 
    attestation/.style={Latex-, thick, draw=blue}, 
    attester/.style={circle, fill=blue!50, minimum size=4mm, inner sep=0pt, font=\normalsize, text=white} 
]

\node[block] (block1) {$9$};
\node[block, right=of block1] (block2) {$2$};
\node[block, right=of block2] (block3) {$3$};

\begin{pgfonlayer}{background}
\draw[dashed line] ($(block1.west)$) ++(0,2) -- ++(0,-4);
\draw[dashed line] ($(block2.west)$) ++(0,2) -- ++(0,-4);
\draw[dashed line] ($(block3.west)$) ++(0,2) -- ++(0,-4);
\draw[dashed line] ($(block4.west)$) ++ (0,2) -- ++(0,-4);
\end{pgfonlayer}

\node[label, above=1cm of block1, xshift=0.75cm] {slot $k-3$};
\node[label, above=1cm of block2, xshift=0.75cm] {slot $k-2$};
\node[label, above=1cm of block3, xshift=0.75cm] {slot $k-1$};
\node[label, above=1cm of block4, xshift=0.75cm] {slot $k$};

\draw[-Latex] (block2) -- (block1);
\draw[-Latex] (block3) to[bend right=27] (block1);

\node[attester, below=1cm of block1, xshift=1cm] (attest0) {};
\node[attester, below=1cm of block1, xshift=1.5cm] (attest1) {};
\node[attester, below=1cm of block1, xshift=2cm] (attest2) {};
\node[attester, below=1cm of block2, xshift=1cm] (attest3) {};
\node[attester, below=1cm of block2, xshift=1.5cm] (attest4) {};
\node[attester, below=1cm of block2, xshift=2cm] (attest5) {};
\node[attester, below=1cm of block3, xshift=1cm] (attest6) {};
\node[attester, below=1cm of block3, xshift=1.5cm] (attest7) {};
\node[attester, below=1cm of block3, xshift=2cm] (attest8) {};

\draw[attestation] (block1.south) -- (attest0);
\draw[attestation] (block1.south) -- (attest1);
\draw[attestation] (block1.south) -- (attest2);
\draw[attestation] (block1.south) -- (attest3);
\draw[attestation] (block2.south) -- (attest4);
\draw[attestation] (block2.south) -- (attest5);
\draw[attestation] (block3.south) -- (attest6);
\draw[attestation] (block3.south) -- (attest7);
\draw[attestation] (block3.south) -- (attest8);

\end{tikzpicture} 

\vspace{.3cm}

\begin{tikzpicture}[
    node distance=5cm and 1.5cm, 
    block/.style={rectangle, rounded corners, draw, fill=gray!30, minimum width=1.5cm, minimum height=1cm, align=center}, 
    blockOrange/.style={rectangle, rounded corners, draw, fill=orange!30, minimum width=1.5cm, minimum height=1cm, align=center},
    dashed line/.style={dashed, very thick, draw=gray!50},
    label/.style={text=gray!50, font=\normalsize}, 
    attestation/.style={Latex-, thick, draw=blue}, 
    attester/.style={circle, fill=blue!50, minimum size=4mm, inner sep=0pt, font=\normalsize, text=white} 
]

\node[block] (block1) {$10.2$};
\node[block, right=of block1] (block2) {$3.2$};
\node[block, right=of block2] (block3) {$3$};
\node[blockOrange, right=of block3] (block4) {$1.2$};

\begin{pgfonlayer}{background}
\draw[dashed line] ($(block1.west)$) ++(0,2) -- ++(0,-4);
\draw[dashed line] ($(block2.west)$) ++(0,2) -- ++(0,-4);
\draw[dashed line] ($(block3.west)$) ++(0,2) -- ++(0,-4);
\draw[dashed line] ($(block4.west)$) ++ (0,2) -- ++(0,-4);
\end{pgfonlayer}

\node[label, above=1cm of block1, xshift=0.75cm] {slot $k-3$};
\node[label, above=1cm of block2, xshift=0.75cm] {slot $k-2$};
\node[label, above=1cm of block3, xshift=0.75cm] {slot $k-1$};
\node[label, above=1cm of block4, xshift=0.75cm] {slot $k$};

\draw[-Latex] (block2) -- (block1);
\draw[-Latex] (block3) to[bend right=27] (block1);
\draw[-Latex, draw=orange, line width=1pt] (block4) to[bend right=23] node[sloped,above,pos=0.5] {$\phi_{k}=1$} (block2) ;

\node[attester, below=1cm of block1, xshift=1cm] (attest0) {};
\node[attester, below=1cm of block1, xshift=1.5cm] (attest1) {};
\node[attester, below=1cm of block1, xshift=2cm] (attest2) {};
\node[attester, below=1cm of block2, xshift=1cm] (attest3) {};
\node[attester, below=1cm of block2, xshift=1.5cm] (attest4) {};
\node[attester, below=1cm of block2, xshift=2cm] (attest5) {};
\node[attester, below=1cm of block3, xshift=1cm] (attest6) {};
\node[attester, below=1cm of block3, xshift=1.5cm] (attest7) {};
\node[attester, below=1cm of block3, xshift=2cm] (attest8) {};
\node[below=1cm of block4, minimum size=4mm, circle] (attest9) {$\rho a$};

\draw[attestation] (block1.south) -- (attest0);
\draw[attestation] (block1.south) -- (attest1);
\draw[attestation] (block1.south) -- (attest2);
\draw[attestation] (block1.south) -- (attest3);
\draw[attestation] (block2.south) -- (attest4);
\draw[attestation] (block2.south) -- (attest5);
\draw[attestation] (block3.south) -- (attest6);
\draw[attestation] (block3.south) -- (attest7);
\draw[attestation] (block3.south) -- (attest8);
\draw[Latex-, draw=orange] (block4.south) -- (attest9) ;


\end{tikzpicture}
    \caption{Cunning proposer $(0,k)$ deviating from the prescribed protocol with $\rho a = 1.2 (= 0.4 \times 3)$ when two blocks have an attestation weight difference of less than $\rho a$, and one of them is designated by the canonical chain.}
    \label{fig:cunningProposer2}
\end{figure}

It should be noted that without the proposer boost, the cunning proposer strategy would never differ from the obedient proposer strategy. This strategy relies on the advantage provided by the proposer boost to ensure that its block becomes the head according to the fork choice rule. 

Conversely, and intuitively, as the proposer boost increases, the opportunity for the cunning block proposer to act cunningly arises more frequently in the case all attesters are obedient. 

\paragraph*{Best response of a proposer among $s-1$ obedient proposers and $as$ obedient attesters.}

We first study the behavior of one proposer when all others are obedient with respect to their designated slot. 
In the case in which the proposer is associated to the first slot of the game, for this proposer the cunning strategy is the best response and the proposer deviates from the protocol if the cunning condition holds (\autoref{lem:bestRespCun}). 
In the case the proposer is associated to a subsequent slot of the game, then the two strategies obedient and cunning are equivalent; this means that the proposer will follow the protocol (\autoref{lem:firstCunning}).

\begin{lemma}\label{lem:bestRespCun}
    When all other validators are obedient, the cunning proposer strategy is a \textit{best response}. 
\end{lemma}

\begin{proof}
    Let us denote by $\phi_k^C$ and $\phi_k^O$ the actions taken by proposer $(0,0)$ under the cunning strategy and the obedient strategy, respectively. The cunning proposer strategy differs from the obedient strategy when $\phi_k^C > \phi_k^O$.
    Considering that the rest of the validators follow the obedient strategy, a proposed block that becomes the head of the chain at slot $k$ will end up in the canonical chain ($\chi_k=1$). 
    By construction, $\phi_k^C \geq \phi_k^O$, and in both cases, the proposed block will be the head of the canonical chain and thus belong to the canonical chain ($\chi_k=1$). 
    Based on the definition of $u_{(0,k)}$ (cf. \autoref{eq:utilityProposer}), the reward increases as the sum increases. This implies that $u_{(0,k)}(\sigma_{-(0,k)}^O, \sigma_{(0,k)}^C) \geq u_{(0,k)}(\sigma_{-(0,k)}^O, \sigma_{(0,k)}^O)$.
\end{proof}

Let us note that when everyone else is obedient, the cunning strategy can only differ from the obedient strategy for the first proposer. All subsequent proposers will follow the protocol regardless of whether they follow the cunning or obedient strategy because the cunning condition is never satisfied. We formalize this in the following observation:

\begin{lemma}\label{lem:firstCunning}
For the strategy profile $\sigma_{-(0,i)}^O$ where all other validators follow the obedient strategy, $\phi_{i}^C \neq \phi_{i}^O$ only if $i=0$, where $\phi_{i}$ is the action of proposer $(0,i)$.
\end{lemma}

\begin{proof}
    
This can be explained because our model makes strong assumptions about synchronous network conditions. When considering that all attesters are obedient, this implies that in each slot, all attesters will send the same attestation. 
As a result, every new block will have an attestation weight that is a multiple of $a$. If the first proposer is obedient, it attaches its block to the branch with the highest attestation weight $w_f$. The obedient attesters will attest to it, adding a weight of $a$. 
This makes the cunning condition (\autoref{obs:cunningCondition}) impossible for subsequent proposers, as $w_f + a - w_g \geq a$. 
Nevertheless, the first proposer can act cunningly since the network's state before the first slot is not predetermined, leaving any arrangement of blocks and attestation weights possible.
\end{proof}

\paragraph*{Best response of a proposer among $s-1$ cunning proposers and $as$ obedient attesters.}

We study the behavior of one proposer when all other proposers are cunning and attesters are obedient. We have two cases: 
\begin{itemize}
    \item $\rho < 1/2$. If the proposer is associated with the first slot of the game, the cunning strategy is the best response, and the proposer will deviate from the protocol if the cunning condition holds (\autoref{lem:cuningBRifRHOinfHALF}).

    If the proposer is associated with any subsequent slot in the game, the cunning condition will never hold. In this case, the obedient and cunning strategies are equivalent, leading the proposer to act as prescribed by the protocol (\autoref{lem:actObedientlyAfterSlot}).

    \item $\rho \geq 1/2$. In this case, the cunning condition can apply to multiple consecutive proposers. If the cunning condition does not hold for the second proposer, then the cunning strategy is the best response for the first proposer (\autoref{lem:propBRtxfee}), causing the first proposer to deviate from the protocol.

    If the cunning condition holds for more than just the first proposer, the cunning strategy becomes the best response if, and only if, the expected rewards gained from acting cunningly and hoarding the rewards over the two previous slots exceed the rewards from the most recent slot alone. (\autoref{lem:propBRtxfee2}).
    In this scenario, each proposer deviates from the prescribed protocol.
\end{itemize}

\begin{lemma} \label{lem:cuningBRifRHOinfHALF}
    The cunning proposer strategy is a best response for a proposer when all other proposers are cunning and attesters are obedient, provided $\rho < 1/2$. 
     If the cunning condition holds, it will only do so for the first proposer, causing this proposer to deviate from the protocol.
\end{lemma}

\begin{proof}
    With $\rho < 1/2$, only the first proposer can act cunningly. Let's assume the first proposer acts cunningly, meaning that the cunning condition is satisfied. The maximum gap between $w_f$ and $w_g$ for the first proposer to act cunningly is $\rho a$, such that $w_f = w_g + \rho a$. After the attestations sent by the obedient attesters in the first slot, the attestation weight of the branch designated by the fork choice rule becomes $w_g + a$. For the second proposer to act cunningly, it must hold that $w_g + a - w_f \leq \rho a$ (cunning condition for the second proposer). Substituting $w_f$ with the maximum possible gap from $w_g$, this condition implies that the second proposer can act cunningly if and only if:
    \begin{equation}
    \begin{split}
        w_g + a - (w_g + \rho a) &\leq \rho a \\
        \frac{1}{2} &\leq \rho.
    \end{split}
    \end{equation}
    Thus, with cunning proposers and obedient attesters, only the first proposer can act cunningly, deviating from the obedient action. For the first proposer, acting cunningly will yield the maximum rewards.
\end{proof}

In fact, since only the first proposer can act cunningly when $\rho < 1/2$, the obedient strategy remains a best response for the rest of the proposers, even when all other proposers are cunning and attesters are obedient.

\begin{lemma} \label{lem:propBRtxfee}
    The cunning proposer strategy is a best response for a proposer when all other proposers are cunning and attesters are obedient, if $\rho \geq 1/2$ and the cunning condition does not hold for the second proposer.
    If the cunning condition holds, it will only do so for the first proposer, causing this proposer to deviate from the protocol.
\end{lemma}

\begin{proof}
    If the attestation weight of the branch $f$ designated by the fork choice rule, $w_f$, and the attestation weight $w_g$ of a concurrent branch $g$ are such that $w_g + a - w_f > \rho a$, the second proposer cannot act cunningly with obedient attesters (cunning condition false for the second proposer). We previously showed that if the first proposer is cunning and attaches its block to the concurrent chain $g$, the obedient attesters will follow, increasing the weight of $g$ to $w_g + a$. By ensuring that $w_g + a - w_f > \rho a$, we prevent the second proposer from changing the canonical chain with the proposer boost. Thus, this condition ensures that the first proposer can be the only one to act cunningly, and in this case, the best response is the cunning strategy.
\end{proof}

\begin{lemma} \label{lem:propBRtxfee2}
    When all proposers are cunning, attesters are obedient, and $\rho \geq 1/2$, the cunning strategy is a best response if the cunning condition holds for the second proposer and:
    \begin{equation}
        \frac{f_{k-2} - f_{k-1}}{2} \geq \frac{ax}{27},
    \end{equation}
    where $f_{k}$ denotes the transaction fees emitted at slot $k$.
\end{lemma}

\begin{proof} 
If the cunning condition is true for the second proposer ($w_g + a - w_f \leq \rho a$ ) it can attach its block to the branch with attestation weight $w_f$ since the gap with the concurrent chain of weight $w_g + a$ is less than the proposer boost $\rho a$. The obedient attesters of the second slot will add an attestation weight of $a$ to $w_f$. Following this, the gap between the attestation weights of the two concurrent branches will always remain less than $\rho a$, leading all cunning proposers to attach their blocks two slots prior.
    
We illustrate in \autoref{fig:cunningBouncing} the "bouncing" that will unfold due to cunning proposers. Their resulting reward will be affected, as the repeated bouncing of the canonical chain between the two branches will cause the blocks from each chain to become canonical with a probability of $1/2$. No attesters will receive the maximum reward since they would never attest in accordance with the following proposer. The reward of the proposer $(0,k)$ following the cunning strategy will thus be:
    \begin{equation}
    \begin{split}
        u_{(0,k)}(\sigma_{-(0,k)},\sigma_{(0,k)}^C) &= \frac{1}{2} \left(\frac{a}{7} \cdot \frac{20x}{27} + f_{k-2} + \frac{a}{7} \cdot \frac{20x}{27} + f_{k-1} \right)  \\
        &= \frac{a}{7} \cdot \frac{20x}{27} + \frac{f_{k-2} + f_{k-1}}{2},
    \end{split}
    \end{equation}
    with $\sigma_{-(0,k)}$ being the strategy profile in which every proposer is cunning and every attester is obedient. 
    
    Being cunning is a best response when $w_f + \rho a \geq a$, if and only if:
    \begin{equation}
    \begin{split}
        u_{(0,k)}(\sigma_{-(0,k)},\sigma_{(0,k)}^C) &\geq u_{(0,k)}(\sigma_{-(0,k)},\sigma_{(0,k)}^O) \\
        \Leftrightarrow \quad \frac{a}{7} \cdot \frac{20x}{27} + \frac{f_{k-2} + f_{k-1}}{2} &\geq \frac{ax}{7} + f_{k-1} \\
        \Leftrightarrow \quad \frac{f_{k-2} - f_{k-1}}{2} &\geq \frac{ax}{27},
    \end{split}
    \end{equation}
    where $\sigma_{-(0,k)}$ is the strategy profile in which every proposer is cunning and every attester is obedient, and $\sigma_{(0,k)}^O$ is the obedient strategy. Therefore, if $f_{k-2}$ is not sufficiently greater than $f_{k-1}$, the best response is the obedient strategy; otherwise, the best response is the cunning strategy.
\end{proof}

Since transaction fees are positive, they cannot continue to decrease indefinitely with each slot. This implies that when the cunning condition holds for the second proposer, eventually one proposer will follow the obedient strategy, thereby stopping the fork.

\begin{figure}[ht]
    \centering
    \resizebox{\columnwidth}{!}{%
    \begin{tikzpicture}[
    node distance=5cm and 1.5cm, 
    block/.style={rectangle, rounded corners, draw, fill=gray!30, minimum width=1.5cm, minimum height=1cm, align=center}, 
    blockOrange/.style={rectangle, rounded corners, draw, fill=orange!30, minimum width=1.5cm, minimum height=1cm, align=center},
    dashed line/.style={dashed, very thick, draw=gray!50},
    label/.style={text=gray!50, font=\normalsize}, 
    attestation/.style={Latex-, draw=blue}, 
    attester/.style={circle, fill=blue!50, minimum size=4mm, inner sep=0pt, font=\normalsize, text=white} 
]

\node[block] (block1) {$X$};
\node[block, right=of block1] (block2) {$\rho a$};
\node[block,draw=orange,right=of block2] (block3) {$\color{orange} \rho a$};

\begin{pgfonlayer}{background}
\draw[dashed line] ($(block1.west)$) ++(0,2) -- ++(0,-3);
\draw[dashed line] ($(block2.west)$) ++(0,2) -- ++(0,-3);
\draw[dashed line] ($(block3.west)$) ++(0,2) -- ++(0,-3);
\draw[dashed line, opacity=0] ($(block4.west)$) ++ (0,2) -- ++(0,-3);
\draw[dashed line,opacity=0] ($(block4.east)$) ++ (1.5,2) -- ++(0,-3);
\end{pgfonlayer}

\node[label, above=1cm of block1, xshift=0.75cm] {slot $k-2$};
\node[label, above=1cm of block2, xshift=0.75cm] {slot $k-1$};
\node[label, above=1cm of block3, xshift=0.75cm] {slot $k$};

\draw[-Latex] (block2) -- (block1);
\draw[-Latex] (block3) to[bend right=27] (block1) ;

\end{tikzpicture}  }

\vspace{.3cm}

\resizebox{\columnwidth}{!}{%
\begin{tikzpicture}[
    node distance=5cm and 1.5cm, 
    block/.style={rectangle, rounded corners, draw, fill=gray!30, minimum width=1.5cm, minimum height=1cm, align=center}, 
    blockOrange/.style={rectangle, rounded corners, draw, fill=orange!30, minimum width=1.5cm, minimum height=1cm, align=center},
    dashed line/.style={dashed, very thick, draw=gray!50},
    label/.style={text=gray!50, font=\normalsize}, 
    attestation/.style={Latex-, draw=blue}, 
    attester/.style={circle, fill=blue!50, minimum size=4mm, inner sep=0pt, font=\normalsize, text=white} 
]

\node[block] (block1) {$X$};
\node[block, right=of block1] (block2) {$\rho a \color{orange} + \rho a$};
\node[block, right=of block2] (block3) {$a$};
\node[block,draw=orange,right=of block3] (block4) {$\color{orange} \rho a$};

\begin{pgfonlayer}{background}
\draw[dashed line] ($(block1.west)$) ++(0,2) -- ++(0,-3);
\draw[dashed line] ($(block2.west)$) ++(0,2) -- ++(0,-3);
\draw[dashed line] ($(block3.west)$) ++(0,2) -- ++(0,-3);
\draw[dashed line] ($(block4.west)$) ++ (0,2) -- ++(0,-3);
\draw[dashed line,opacity=0] ($(block4.east)$) ++ (1.5,2) -- ++(0,-3);
\end{pgfonlayer}

\node[label, above=1cm of block1, xshift=0.75cm] {slot $k-2$};
\node[label, above=1cm of block2, xshift=0.75cm] {slot $k-1$};
\node[label, above=1cm of block3, xshift=0.75cm] {slot $k$};
\node[label, above=1cm of block4, xshift=0.75cm] {slot $k+1$};

\draw[-Latex] (block2) -- (block1);
\draw[-Latex] (block3) to[bend right=27] (block1) ;
\draw[-Latex] (block4) to[bend left=27] (block2) ;

\end{tikzpicture}  
 }

    \caption{$X$ indicates that the value of the block is irrelevant. As a reminder, the proposer boost is equivalent to an attestation weight of $\rho a$ for a new block (in orange). 
    In this scenario, the proposer of slot $k$ is cunning and all the attesters are obedient. The proposer of slot $k$ takes advantage of the proposer boost to become the head of the canonical chain.
    The proposer of slot $k+1$ can cunningly become the head of the canonical chain by attaching its block to the block from slot $k-1$ and can become the head of the chain only if $2\rho a \geq a$, which means $\rho \geq 1/2$. In conclusion, a proposer boost greater than $1/2$ can create a situation in which multiple forks occur in the presence of cunning proposers and obedient attesters.}
    \label{fig:cunningBouncing}
\end{figure}
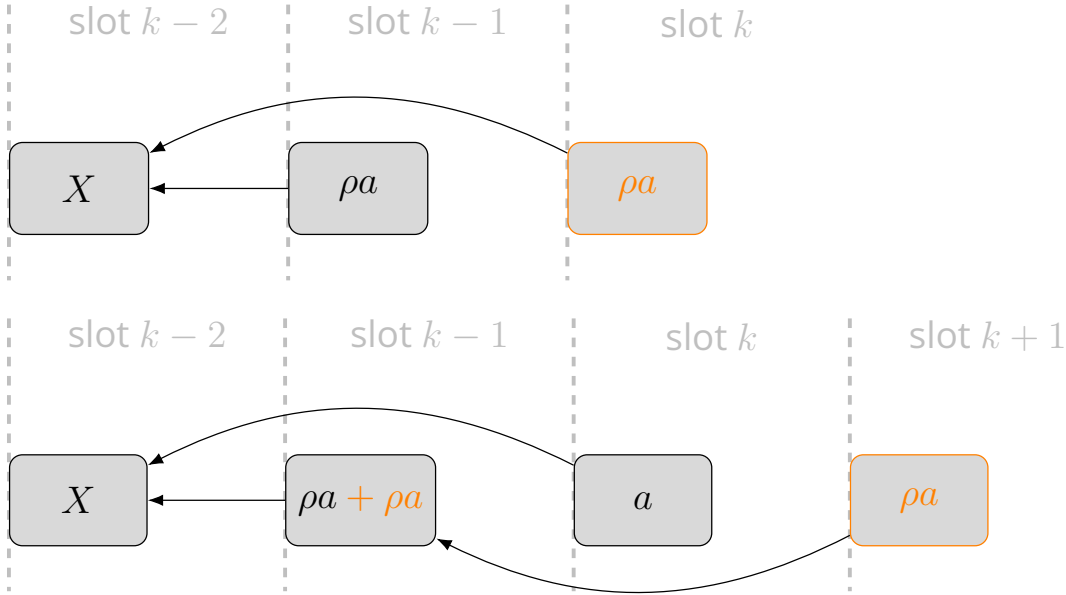

\paragraph*{Best response of an attester among $s$ cunning proposers and $an-1$ obedient attesters.}

Until now, we have described a scenario where all proposers follow the cunning strategy, and attesters follow the obedient strategy, leading to a potentially long fork in which attesters pay the price for the cunning behavior of proposers and do not receive the maximum reward. Let us now introduce the cunning attester strategy, which takes advantage of knowing when the cunning proposer strategy is the best response to act accordingly and secure a higher reward.

\begin{stratAttester}
\textbf{Cunning Attester} ($\sigma^{C}_{(i,k)}$): \\
\textit{Action:} \quad $\nu_{(i,k)} = \sigma^C_{0,k+1}(\mathcal{T}_k,\mathcal{A}_{k-1}\cup A^O_k)$. \\

The cunning attester $(i,k)$ attests to the parent of the block in slot $k+1$, assuming that all other attesters in slot $k$ will act obediently ($A^O_k$) and that the proposer of slot $k+1$ will act cunningly. The action is the same as the following cunning proposer's, i.e., $\nu_{(i,k)} = \phi^C_{k+1}$.
\end{stratAttester}

We found that when $\rho < 1/2$, the obedient attester strategy and the cunning attester strategy are equivalent. This implies that attesters will follow the protocol (\autoref{lem:obedStratIScunStrat_attester}) when $\rho<1/2$. 
Moreover, if $\rho \geq 1/2$, when all proposers are cunning and attesters are obedient, the best response is the cunning attester strategy. When the cunning condition holds for the second proposer the cunning attester strategy will deviate from the protocol. Otherwise, all attesters will follow the protocol (\autoref{lem:7}).

\begin{lemma} \label{lem:obedStratIScunStrat_attester}
    When $\rho < 1/2$, the obedient attester strategy and the cunning attester strategy are equivalent.
\end{lemma}

\begin{proof}
    This result stems from the fact that when $\rho < 1/2$, only the first proposer can act conspicuously cunningly, meaning they deviate from the protocol (cf. proof of \autoref{lem:cuningBRifRHOinfHALF}). Therefore, if all subsequent proposers act similarly to obedient proposers and follow the protocol, attesters will never have the opportunity to act cunningly and will follow the protocol as well.
\end{proof}
    
\begin{lemma} \label{lem:7}
    When $\rho \geq 1/2$ and all proposers are cunning while all other attesters are obedient, the cunning attester strategy is a best response. If the cunning condition holds for the second proposer, the cunning attester strategy will lead the attester to deviate from the protocol. Otherwise, all attesters will follow the protocol.
\end{lemma}

\begin{proof}
    For attesters to exhibit cunning behavior, more than just the first proposer must act cunningly. This occurs if and only if the cunning condition holds for the second attester.
    
    The reward for attester $(i,k)$ following the cunning attester strategy $\sigma_{(i,k)}^C$, while all other attesters are obedient and proposers are cunning $\sigma_{-(i,k)}$, is:
    \begin{equation}
        u_{(i,k)}(\sigma_{-(i,k)},\sigma_{(i,k)}^C) = \frac{47 x}{54}.
    \end{equation}
    
    Since the attester's reward depends on when their attestation is included in a block, it also depends on whether the blocks belong to the canonical chain. Each attestation is included in the next two blocks, which are on different chains, each having a 1/2 probability of being in the canonical chain. This gives the cunning attester a reward of $\frac{1}{2}(x + \frac{20x}{27})$. Following the obedient attester strategy leads to a reward of $\frac{20x}{27}$, as in both blocks, the attestation will either attest to the wrong block or be included too late. Thus, $u_{(i,k)}(\sigma_{-(i,k)},\sigma_{(i,k)}^C) \geq u_{(i,k)}(\sigma_{-(i,k)},\sigma_{(i,k)}^O)$.
\end{proof}

\paragraph*{Best response of an attester among $s$ cunning proposers and $an-1$ cunning attesters.}

We study the behavior of an attester when all proposers are cunning and other attesters cunning. In this case the best response is the cunning attester strategy. This strategy only deviates from the protocol for the attesters of the first slot if the cunning condition holds for the second proposer. However this deviation will prevent the cunning condition to hold for the third proposer, effectively making all subsequent validators to follow the protocol.

\begin{lemma} \label{lem:BRcunningAttester} 
    The cunning attester strategy is a best response for an attester when all validators are cunning.
    If the cunning condition holds for the second proposer, the cunning attester strategy will lead the attesters of the first slot to deviate from the protocol. Otherwise, all attesters will follow the protocol.
\end{lemma}

\begin{proof}  
    For attesters to exhibit cunning behavior, more than just the first proposer $(0,0)$ must act cunningly. This occurs if the cunning condition holds for the second proposer $(0,1)$.

    In this case, all attesters of the first slot will expect proposer $(0,1)$ to act cunningly and attach itself to the block designated by the fork choice rule at the beginning of the game ($\mathcal{F}(\mathcal{T}_{0},\mathcal{A}_{-1})$), leading them to attest to the head of the branch with total attestation weight $w_f$. This scenario is represented in \autoref{fig:cunningAttester}. 
    As a result, the block proposed by the first proposer $(0,0)$ will not be attested by the attesters. 
    The block proposed by $(0,1)$ will belong to the canonical chain since no other fork is possible for the subsequent proposers. The gap between $w_f + a$ and $w_g$ is too large for the proposer boost to enable further cunning actions, the cunning condition cannot hold anymore.
    The attesters $(i,0)$ of the first slot, who act in accordance with proposer $(0,1)$, receive the maximum reward. 
    Proposers $(0,2)$ and beyond will not have the opportunity to act cunningly, nor will the remaining attesters, resulting in all attesters receiving the maximum reward.
\end{proof}

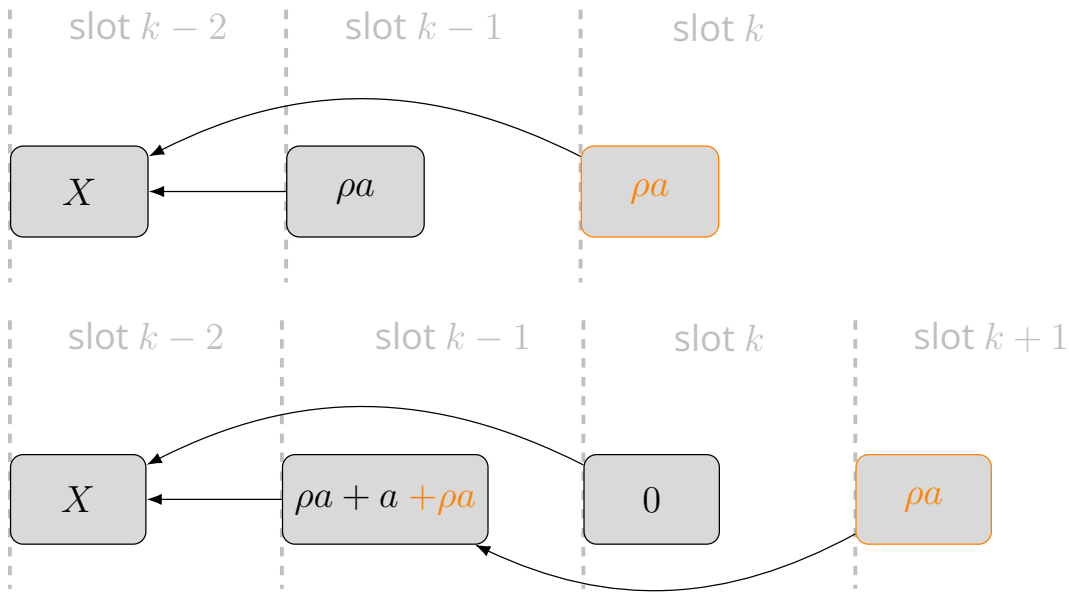
\begin{figure}[ht]
    \centering
    \resizebox{\columnwidth}{!}{%
    \begin{tikzpicture}[
    node distance=5cm and 1.5cm, 
    block/.style={rectangle, rounded corners, draw, fill=gray!30, minimum width=1.5cm, minimum height=1cm, align=center}, 
    blockOrange/.style={rectangle, rounded corners, draw, fill=orange!30, minimum width=1.5cm, minimum height=1cm, align=center},
    dashed line/.style={dashed, very thick, draw=gray!50},
    label/.style={text=gray!50, font=\normalsize}, 
    attestation/.style={Latex-, draw=blue}, 
    attester/.style={circle, fill=blue!50, minimum size=4mm, inner sep=0pt, font=\normalsize, text=white} 
]

\node[block] (block1) {$X$};
\node[block, right=of block1] (block2) {$\rho a$};
\node[block,draw=orange,right=of block2, xshift=0.2cm] (block3) {$\color{orange} \rho a$};

\begin{pgfonlayer}{background}
\draw[dashed line] ($(block1.west)$) ++(0,2) -- ++(0,-3);
\draw[dashed line] ($(block2.west)$) ++(0,2) -- ++(0,-3);
\draw[dashed line] ($(block3.west)$) ++(0,2) -- ++(0,-3);
\draw[dashed line, opacity=0] ($(block4.west)$) ++ (0,2) -- ++(0,-3);
\draw[dashed line,opacity=0] ($(block4.east)$) ++ (1.5,2) -- ++(0,-3);
\end{pgfonlayer}

\node[label, above=1cm of block1, xshift=0.75cm] {slot $k-2$};
\node[label, above=1cm of block2, xshift=0.75cm] {slot $k-1$};
\node[label, above=1cm of block3, xshift=0.75cm] {slot $k$};

\draw[-Latex] (block2) -- (block1);
\draw[-Latex] (block3) to[bend right=27] (block1) ;

\end{tikzpicture}  }

\vspace{.3cm}

\resizebox{\columnwidth}{!}{%
\begin{tikzpicture}[
    node distance=5cm and 1.5cm, 
    block/.style={rectangle, rounded corners, draw, fill=gray!30, minimum width=1.5cm, minimum height=1cm, align=center}, 
    blockOrange/.style={rectangle, rounded corners, draw, fill=orange!30, minimum width=1.5cm, minimum height=1cm, align=center},
    dashed line/.style={dashed, very thick, draw=gray!50},
    label/.style={text=gray!50, font=\normalsize}, 
    attestation/.style={Latex-, draw=blue}, 
    attester/.style={circle, fill=blue!50, minimum size=4mm, inner sep=0pt, font=\normalsize, text=white} 
]

\node[block] (block1) {$X$};
\node[block, right=of block1] (block2) {$\rho a+a$ $\color{orange} + \rho a$};
\node[block, right=of block2, xshift=-0.45cm] (block3) {$0$};
\node[block,draw=orange,right=of block3] (block4) {$\color{orange} \rho a$};

\begin{pgfonlayer}{background}
\draw[dashed line] ($(block1.west)$) ++(0,2) -- ++(0,-3);
\draw[dashed line] ($(block2.west)$) ++(0,2) -- ++(0,-3);
\draw[dashed line] ($(block3.west)$) ++(0,2) -- ++(0,-3);
\draw[dashed line] ($(block4.west)$) ++ (0,2) -- ++(0,-3);
\draw[dashed line,opacity=0] ($(block4.east)$) ++ (1.5,2) -- ++(0,-3);
\end{pgfonlayer}

\node[label, above=1cm of block1, xshift=0.75cm] {slot $k-2$};
\node[label, above=1cm of block2, xshift=0.75cm] {slot $k-1$};
\node[label, above=1cm of block3, xshift=0.75cm] {slot $k$};
\node[label, above=1cm of block4, xshift=0.75cm] {slot $k+1$};

\draw[-Latex] (block2) -- (block1);
\draw[-Latex] (block3) to[bend right=27] (block1) ;
\draw[-Latex] (block4) to[bend left=27] (block2) ;

\end{tikzpicture}  
 }

    \caption{$X$ indicates that the value of the block is irrelevant. In this scenario, all validators act cunningly. The proposer of slot $k$ attaches its block to the block from two slots prior. The cunning attesters in slot $k$ attest to the block from slot $k-1$ to align with the following proposer's strategy. The proposer of slot $k+1$ is then compelled to attach its block to the block from slot $k-1$. 
    As a reminder, the proposer boost is equivalent to an attestation weight of $\rho a$ for a new block (in orange). This results in the proposer of slot $k$ forking alone and receiving no rewards.
    }
    \label{fig:cunningAttester}
\end{figure}

\paragraph*{Best response of a proposer among $s-1$ cunning proposers and $as$ cunning attesters.}

Now that the attester can also act cunningly, let us evaluate the best response of proposers. We found that for $\rho < 1/2$, as usual, only the first proposers can deviate from the prescribed protocol. Doing so will yield more reward hence the cunning strategy is the best response for the first proposer.

When $\rho\geq 1/2$, there are two cases. If the cunning condition does not hold for the second proposer, the best response is the cunning proposer strategy and only the first proposer has the opportunity to deviate from the protocol. 
On the other hand, if the cunning condition holds for the second proposer, all attesters of the first slot will deviate from the protocol and not attest the block of the first proposer. This makes the obedient proposer strategy the best response for the first proposer as well as all other proposers that will follow the protocol regardless of the first proposer strategy when attesters are cunning.

\begin{lemma}
    If $\rho < 1/2$ and all validators are cunning, the best response for the first proposer is the cunning strategy. If the cunning condition holds, it will only do so for the first proposer, causing this proposer to deviate from the protocol.
\end{lemma}

\begin{proof}
    This follows directly from \autoref{lem:obedStratIScunStrat_attester} and \autoref{lem:cuningBRifRHOinfHALF}.
\end{proof}

\begin{lemma}
    If $\rho \geq 1/2$ and all validators are cunning, the best response for a proposer is cunning, if the cunning condition does not hold for the second proposer. Otherwise, the obedient proposer strategy is the best response.
\end{lemma}

\begin{proof}
    If the cunning condition holds for the second proposer, \autoref{lem:BRcunningAttester} describes how the scenario would unfold. A possible outcome is represented in \autoref{fig:cunningAttester}. The result is that, for the first proposer, the best response is the cunning strategy only if the second proposer cannot act cunningly.
    
    Otherwise, as described in the proof of \autoref{lem:BRcunningAttester}, if the first proposer remains cunning, they will receive zero reward. All other proposers follow the protocol regardless of the strategy of the first proposer.
\end{proof}

A counterintuitive finding is that with a higher proposer boost $\rho > 1/2$, the obedient proposer strategy can be favored. One reason is that the transaction fees gained may not be enough to compensate for the probability of not belonging to the canonical chain. Another reason is that being cunning can backfire if the attesters are also cunning, leading to the cunning block not being attested at all. The assurance of belonging to the canonical chain and the available rewards can be sufficient to make the obedient strategy the favored response.



\paragraph*{Eventual Incentive compatibility}

We have seen that the strategy profile in which all validators are obedient and follow the prescribed protocol is not a Nash equilibria, participants can gain from changing strategy. In this sense the consensus protocol and more precisely the fork choice rule is not incentive compatible.

Nonetheless, we now show that in all equilibria, there exists a slot after which all validators follow the prescribed protocol (\autoref{the:eventualNashObedient}).



\begin{lemma} \label{lem:actObedientlyAfterSlot}
    Once a proposer follows the protocol, all subsequent validators do so.
\end{lemma}

\begin{proof} 
    A proposer $(0,j)$ following the protocol implies that if all attesters of its slot also follow the protocol, the next proposer cannot deviate. This is because all obedient attesters in that slot would give an attestation weight of $a$ to the block proposed by $(0,j)$. Since it extends the branch designated by the fork choice rule, with an attestation weight $w_f \geq w_g$, where $w_g$ is the attestation weight of any concurrent branch, adding $a$ to $w_f$ ensures that no proposer deviate from the protocol (cf. \autoref{obs:cunningCondition}).

    Knowing this, the attesters $(i,j)$ of slot $j$ will follow the protocol as well. No validators can deviate from the protocol after a block is attached to the head of the canonical chain.
\end{proof}

\begin{theorem} \label{the:eventualNashObedient}
    In all Nash equilibria, there is a slot after which all validators follow the protocol.
\end{theorem}

\begin{proof}
We know that once a proposer acts obediently, all subsequent validators do (\autoref{lem:actObedientlyAfterSlot}).

If there is an equilibrium in which one proposer follows the obedient strategy and extends the head of the canonical chain, the theorem is valid. We now look at proposers all following the cunning strategy. When all proposers follow the cunning strategy, to have more than the first to effectively act cunningly we need to have $w_{g}+a-w_{f}\leq\rho a$ otherwise the second proposer will extend the head of the canonical chain, validating the theorem. Then in the case of the fork continuing with each proposer thus attaching their block two slots prior, this makes each of their blocks have an expectation of 1/2 to belong to the canonical chain ($\chi=1/2$). As computed in \autoref{lem:propBRtxfee}, their reward will thus be in the form of: $f_{k-2}-f_{k-1} \geq \alpha$, where $\alpha$ is a positive number that depends on the attestation included and their reward associated. No matter the value of $\alpha$, even taking $\alpha=0$, a proposer will have as best response to be cunning only if the transaction fees gained with a probability 1/2 by being cunning are at least superior to the transaction fees obtained with certainty otherwise.

This condition cannot be true for all proposers as the transaction fees are positive and discrete. Eventually, a proposer $(0,k)$ will see previous transaction fees where $f_{k-2}<f_{k-1}$. The best response of proposer $(0,k)$ is to act obediently.
\end{proof}

We can conclude that under perfect network conditions, regardless of the proposer boost, the obedient strategy will eventually prevail.



\section{Conclusion}

In this chapter, we have analyzed the Ethereum PoS protocol through a game-theoretic lens, particularly focusing on the incentive mechanisms that influence the behavior of proposers and attesters. Our findings reveal that the current design leads rational validators to all eventually adhere to the protocol. Specifically the proposer boost mechanism does not permit prolonged forks while having good network conditions. Surprisingly, a high proposer boost (superior to $1/2$) can even prevent cunning behavior. 

This initial analysis focused on two strategies: cunning and obedient. Future research will expand to include a broader range of strategies. Additionally, exploring different assumptions about network conditions, such as communication with more realistic delays, is expected to significantly impact the results.

    \chapter{Conclusion and Perspectives}\label{chap:Conclusion}

\chapterLettrine{D}{riven} by the need to evaluate the impact of incentive mechanisms on blockchain robustness, we embarked on this thesis.
Traditional distributed systems approaches, which consider honest and Byzantine participants, often overlook incentives, while game theory focuses on incentives but typically disregards blockchain robustness. Our work aims to bridge this gap by examining the case study of Ethereum PoS.

Ethereum PoS is unique for two reasons: it has a very active research community, and its protocol is a hybrid of Nakamoto-style and BFT-like consensus. These features make it an ideal subject for studying complex ideas that apply broadly to other blockchain systems, as most blockchains incorporate at least one of these fundamental components.

\

Our analysis began with a detailed examination of the Ethereum PoS protocol. The complexities of this protocol necessitated a dedicated paper \cite{pavloff_ethereum_2023} focused on its description. In this work, we redefined the crucial properties of safety and liveness, which are fundamental to blockchain robustness and serve as the foundation for our entire study. We formalized the protocol based on its code which we transcribed in pseudo-code. We then revealed that the protocol was safe and the liveness was probabilistic due to a possible attack we identified. 

We then shifted our focus to the incentive mechanisms within the Ethereum PoS protocol. While our first analysis was in line with traditional distributed systems approaches and did not consider incentives, the second part of our work \cite{pavloff_byzantine_2024} explored the inactivity leak, an incentive mechanism that directly impacts the protocol’s robustness. Our results show that the mechanism penalizing seemingly inactive validators to restore liveness in times of partitions could be subverted by Byzantine validators to break the safety. 

The final part of our work addresses the behavior of rational participants within the protocol, bridging the gap between distributed system and game theoretic analysis. We investigated whether proposers and attesters could financially benefit from deviating from the protocol, specifically by exploiting the fork choice rule to their advantage. The strategy exploiting the fork choice rule is called cunning while the strategy that adheres to the protocol is called obedient. Our findings state that for any equilibrium, eventually all validators will behave obediently. There are two reasons for this. First we assume perfect network conditions during which messages sent are instantaneously received by all participants. Secondly, not following the obedient strategy can be less rewarding as it can imply that blocks do not eventually end up in the canonical chain, which would yield zero rewards in this case.

\

Our contributions highlight the critical role of incentive mechanisms in blockchains. Rewards and penalties can either prevent misbehavior, as shown in our game-theoretic analysis, or be detrimental if exploited by Byzantine participants, as illustrated in our analysis of the inactivity leak. This work contributes to the effort of making blockchain protocols more comprehensible. The complexity of the Ethereum protocol is undeniable, and we hope to serve as a resource for a detailed, formalized explanation of the protocol. 
While we have striven to provide an accurate description of the consensus mechanisms, the examination of how transactions are processed and executed remains. This constitutes an entire work on its own.

\

Although our work focuses on a single blockchain, the insights gained offer valuable guidance for the design of future protocols. Addressing the intersection of distributed systems and game theory is both recent and challenging. While Abraham et al. \cite{abraham_distributed_2011} introduced the idea of combining these fields in 2011, the complexity involved has led to limited research in this area. One of our goals was to study the protocol with three types of agents: honest, Byzantine, and rational. However, the complexity of the protocol made adding this refinement to the participant model too daunting and, we believe, unfeasible.

\

Future work should aim to develop simpler protocols that remain robust in the presence of honest, Byzantine, and rational players. Our work provides a framework for evaluating protocols. However, a limitation of this thesis is the absence of a proposed solution to the problems identified.

Blockchain technology is still in its early stages, and we hope that our contribution will aid its continued development.


\bibliographystyle{alpha}
\bibliography{main}

\newcommand{\etalchar}[1]{$^{#1}$}
\begin{thebibliography}{CKWN16}

\bibitem[AAH11]{abraham_distributed_2011}
Ittai Abraham, Lorenzo Alvisi, and Joseph~Y. Halpern.
\newblock Distributed computing meets game theory: combining insights from two fields.
\newblock {\em {SIGACT} News}, 42(2):69--76, 2011.

\bibitem[ABPT20]{amoussou_rational_2020}
Yackolley Amoussou{-}Guenou, Bruno Biais, Maria Potop{-}Butucaru, and Sara {Tucci Piergiovanni}.
\newblock Rational vs byzantine players in consensus-based blockchains.
\newblock In {\em Proceedings of the 19th International Conference on Autonomous Agents and Multiagent Systems, {AAMAS} '20, Auckland, New Zealand, May 9-13, 2020}, pages 43--51. International Foundation for Autonomous Agents and Multiagent Systems, 2020.

\bibitem[ACL{\etalchar{+}}19]{alturki_towards_2019}
Musab~A. Alturki, Jing Chen, Victor Luchangco, Brandon~M. Moore, Karl Palmskog, Lucas Pe{\~{n}}a, and Grigore Rosu.
\newblock Towards a verified model of the algorand consensus protocol in coq.
\newblock In {\em Formal Methods. {FM} 2019 International Workshops - Porto, Portugal, October 7-11, 2019, Revised Selected Papers, Part {I}}, volume 12232 of {\em Lecture Notes in Computer Science}, pages 362--367. Springer, 2019.

\bibitem[ACM20]{amores_security_2020}
Ignacio Amores{-}Sesar, Christian Cachin, and Jovana Micic.
\newblock Security analysis of ripple consensus.
\newblock In {\em 24th International Conference on Principles of Distributed Systems, {OPODIS} 2020, December 14-16, 2020, Strasbourg, France (Virtual Conference)}, volume 184 of {\em LIPIcs}, pages 10:1--10:16. Schloss Dagstuhl - Leibniz-Zentrum f{\"{u}}r Informatik, 2020.

\bibitem[ACP{\etalchar{+}}21]{astefanoaei_ternderbake_2021}
Lacramioara Astefanoaei, Pierre Chambart, Antonella~Del Pozzo, Thibault Rieutord, Sara {Tucci Piergiovanni}, and Eugen Zalinescu.
\newblock Tenderbake - {A} solution to dynamic repeated consensus for blockchains.
\newblock In {\em 4th International Symposium on Foundations and Applications of Blockchain 2021, {FAB} 2021, May 7, 2021, University of California, Davis, California, {USA} (Virtual Conference)}, volume~92 of {\em OASIcs}, pages 1:1--1:23. Schloss Dagstuhl - Leibniz-Zentrum f{\"{u}}r Informatik, 2021.

\bibitem[ADL{\etalchar{+}}19]{anceaume_abstract_2019}
Emmanuelle Anceaume, Antonella {Del Pozzo}, Romaric Ludinard, Maria Potop{-}Butucaru, and Sara {Tucci{-}Piergiovanni}.
\newblock Blockchain abstract data type.
\newblock In {\em The 31st {ACM} on Symposium on Parallelism in Algorithms and Architectures, {SPAA} 2019, Phoenix, AZ, USA, June 22-24, 2019}, pages 349--358. {ACM}, 2019.

\bibitem[ADPT18]{amoussou_correctness_2018}
Yackolley Amoussou{-}Guenou, Antonella {Del Pozzo}, Maria Potop{-}Butucaru, and Sara Tucci{-}Piergiovanni.
\newblock Correctness of tendermint-core blockchains.
\newblock In {\em 22nd International Conference on Principles of Distributed Systems, {OPODIS} 2018, December 17-19, 2018, Hong Kong, China}, pages 16:1--16:16, 2018.

\bibitem[ADRT21]{anceaume_finality_2021}
Emmanuelle Anceaume, Antonella {Del Pozzo}, Thibault Rieutord, and Sara Tucci{-}Piergiovanni.
\newblock On finality in blockchains.
\newblock In {\em 25th International Conference on Principles of Distributed Systems, {OPODIS} 2021, December 13-15, 2021, Strasbourg, France}, pages 6:1--6:19, 2021.

\bibitem[APM{\etalchar{+}}23]{attiya_synchronization_2023}
Hagit Attiya, Antonella~Del Pozzo, Alessia Milani, Ulysse Pavloff, and Alexandre Rapetti.
\newblock The synchronization power of auditable registers.
\newblock In {\em 27th International Conference on Principles of Distributed Systems, {OPODIS} 2023, December 6-8, 2023, Tokyo, Japan}, volume 286 of {\em LIPIcs}, pages 4:1--4:23. Schloss Dagstuhl - Leibniz-Zentrum f{\"{u}}r Informatik, 2023.

\bibitem[APPT19]{amoussou_dissecting_2019}
Yackolley Amoussou{-}Guenou, Antonella~Del Pozzo, Maria Potop{-}Butucaru, and Sara {Tucci{-}Piergiovanni}.
\newblock Dissecting tendermint.
\newblock In {\em Networked Systems - 7th International Conference, {NETYS} 2019, Marrakech, Morocco, June 19-21, 2019, Revised Selected Papers}, volume 11704 of {\em Lecture Notes in Computer Science}, pages 166--182. Springer, 2019.

\bibitem[Bac97]{back_hashcash_1997}
Adam Back.
\newblock A partial hash collision based postage scheme, 1997.

\bibitem[Bac02]{back_hashcash_2002}
Adam Back.
\newblock Hashcash - a denial of service counter-measure, 2002.

\bibitem[BBBC19]{biais_blockchain_2019}
Bruno Biais, Christophe Bisière, Matthieu Bouvard, and Catherine Casamatta.
\newblock The blockchain folk theorem.
\newblock {\em Review of Financial Studies}, 32:1662--1715, 05 2019.

\bibitem[BCC{\etalchar{+}}23]{bhudia_game_2023}
Alpesh Bhudia, Anna Cartwright, Edward~J. Cartwright, Darren Hurley{-}Smith, and Julio~C. Hernandez{-}Castro.
\newblock Game theoretic modelling of a ransom and extortion attack on ethereum validators.
\newblock In {\em Proceedings of the 18th International Conference on Availability, Reliability and Security, {ARES} 2023, Benevento, Italy, 29 August 2023- 1 September 2023}, pages 105:1--105:11. {ACM}, 2023.

\bibitem[BG17]{buterin_casper_2017}
Vitalik Buterin and Virgil Griffith.
\newblock Casper the friendly finality gadget.
\newblock {\em CoRR}, abs/1710.09437, 2017.

\bibitem[BHK{\etalchar{+}}20]{buterin_combining_2020}
Vitalik Buterin, Diego Hernandez, Thor Kamphefner, Khiem Pham, Zhi Qiao, Danny Ryan, Juhyeok Sin, Ying Wang, and Yan~X. Zhang.
\newblock Combining {GHOST} and casper.
\newblock {\em CoRR}, abs/2003.03052, 2020.

\bibitem[BKM18]{buchman_latest_2018}
Ethan Buchman, Jae Kwon, and Zarko Milosevic.
\newblock The latest gossip on {BFT} consensus.
\newblock {\em CoRR}, abs/1807.04938, 2018.

\bibitem[Bre00]{brewer_towards_2000}
Eric~A. Brewer.
\newblock Towards robust distributed systems (abstract).
\newblock In Gil Neiger, editor, {\em Proceedings of the Nineteenth Annual {ACM} Symposium on Principles of Distributed Computing, July 16-19, 2000, Portland, Oregon, {USA}}, page~7. {ACM}, 2000.

\bibitem[BRLP20]{buterin_incentives_2020}
Vitalik Buterin, Dani{\"{e}}l Reijsbergen, Stefanos Leonardos, and Georgios Piliouras.
\newblock Incentives in ethereum's hybrid casper protocol.
\newblock {\em Int. J. Netw. Manag.}, 30(5), 2020.

\bibitem[But14]{buterin_next_2014}
Vitalik Buterin.
\newblock Ethereum: A next-generation smart contract and decentralized application platform, 2014.

\bibitem[CFN88]{chaum_untraceable_1988}
David Chaum, Amos Fiat, and Moni Naor.
\newblock Untraceable electronic cash.
\newblock In {\em Advances in Cryptology - {CRYPTO} '88, 8th Annual International Cryptology Conference, Santa Barbara, California, USA, August 21-25, 1988, Proceedings}, volume 403 of {\em Lecture Notes in Computer Science}, pages 319--327. Springer, 1988.

\bibitem[Cha82]{chaum_blind_1982}
David Chaum.
\newblock Blind signatures for untraceable payments.
\newblock In David Chaum, Ronald~L. Rivest, and Alan~T. Sherman, editors, {\em Advances in Cryptology: Proceedings of {CRYPTO} '82, Santa Barbara, California, USA, August 23-25, 1982}, pages 199--203. Plenum Press, New York, 1982.

\bibitem[CKWN16]{carlsten_instability_2016}
Miles Carlsten, Harry~A. Kalodner, S.~Matthew Weinberg, and Arvind Narayanan.
\newblock On the instability of bitcoin without the block reward.
\newblock In {\em Proceedings of the 2016 {ACM} {SIGSAC} Conference on Computer and Communications Security, Vienna, Austria, October 24-28, 2016}, pages 154--167. {ACM}, 2016.

\bibitem[CL99]{castro_practical_1999}
Miguel Castro and Barbara Liskov.
\newblock Practical byzantine fault tolerance.
\newblock In {\em Proceedings of the Third {USENIX} Symposium on Operating Systems Design and Implementation (OSDI), New Orleans, Louisiana, USA, February 22-25, 1999}, pages 173--186. {USENIX} Association, 1999.

\bibitem[CM19]{chen_algorand_2019}
Jing Chen and Silvio Micali.
\newblock Algorand: {A} secure and efficient distributed ledger.
\newblock {\em Theor. Comput. Sci.}, 777:155--183, 2019.

\bibitem[Con22]{teku_code}
Consensys.
\newblock Teku consensus client, 2022.

\bibitem[Dai98]{wei_b-money_1998}
Wei Dai.
\newblock B-money, 1998.

\bibitem[DDS87]{dolev_minimal_1987}
Danny Dolev, Cynthia Dwork, and Larry~J. Stockmeyer.
\newblock On the minimal synchronism needed for distributed consensus.
\newblock {\em J. {ACM}}, 34(1):77--97, 1987.

\bibitem[DLS88]{dwork_consensus_1988}
Cynthia Dwork, Nancy~A. Lynch, and Larry~J. Stockmeyer.
\newblock Consensus in the presence of partial synchrony.
\newblock {\em J. {ACM}}, pages 288--323, 1988.

\bibitem[DN92]{dwork_pricing_1992}
Cynthia Dwork and Moni Naor.
\newblock Pricing via processing or combatting junk mail.
\newblock In {\em Advances in Cryptology - {CRYPTO} '92, 12th Annual International Cryptology Conference, Santa Barbara, California, USA, August 16-20, 1992, Proceedings}, volume 740 of {\em Lecture Notes in Computer Science}, pages 139--147. Springer, 1992.

\bibitem[Edg3 ]{edington_technical_2023}
Ben Edgington.
\newblock {\em A technical handbook on Ethereum's move to proof of stake and beyond}.
\newblock {ETH2 Book}, 2023-.

\bibitem[ES14]{eyal_majority_2014}
Ittay Eyal and Emin~G{\"{u}}n Sirer.
\newblock Majority is not enough: Bitcoin mining is vulnerable.
\newblock In {\em Financial Cryptography and Data Security - 18th International Conference, {FC} 2014, Christ Church, Barbados, March 3-7, 2014, Revised Selected Papers}, volume 8437 of {\em Lecture Notes in Computer Science}, pages 436--454. Springer, 2014.

\bibitem[ES18]{eyal_majority_2018}
Ittay Eyal and Emin~G{\"{u}}n Sirer.
\newblock Majority is not enough: bitcoin mining is vulnerable.
\newblock {\em Commun. {ACM}}, 61(7):95--102, 2018.

\bibitem[Fin08]{finney_fork_2008}
Hal Finney.
\newblock Mail from finney to, 2008.

\bibitem[FMJR20]{fooladgar_incentive_2020}
Mehdi Fooladgar, Mohammad~Hossein Manshaei, Murtuza Jadliwala, and Mohammad~Ashiqur Rahman.
\newblock On incentive compatible role-based reward distribution in algorand.
\newblock In {\em 50th Annual {IEEE/IFIP} International Conference on Dependable Systems and Networks, {DSN} 2020, Valencia, Spain, June 29 - July 2, 2020}, pages 452--463. {IEEE}, 2020.

\bibitem[Fou24]{github_specs}
Ethereum Foundation.
\newblock Consensus specifications github.
\newblock \url{https://github.com/ethereum/consensus-specs/tree/a42d6706d89c414764eda7e2d0103e19f1e23761/specs}, 2024.

\bibitem[GKL15]{garay_bitcoin_2015}
Juan~A. Garay, Aggelos Kiayias, and Nikos Leonardos.
\newblock The bitcoin backbone protocol: Analysis and applications.
\newblock In {\em Advances in Cryptology - {EUROCRYPT} 2015 - 34th Annual International Conference on the Theory and Applications of Cryptographic Techniques, Sofia, Bulgaria, April 26-30, 2015, Proceedings, Part {II}}, volume 9057 of {\em Lecture Notes in Computer Science}, pages 281--310. Springer, 2015.

\bibitem[GLMV23]{galletta_resilience_2023}
Letterio Galletta, Cosimo Laneve, Ivan Mercanti, and Adele Veschetti.
\newblock Resilience of hybrid casper under varying values of parameters.
\newblock {\em Distributed Ledger Technol. Res. Pract.}, 2(1):5:1--5:25, 2023.

\bibitem[Goo14]{goodman_tezos_2014}
L.M. Goodman.
\newblock Tezos — a self-amending crypto-ledger, 2014.

\bibitem[GP20]{grunspan_selfish_2020}
Cyril Grunspan and Ricardo P{\'{e}}rez{-}Marco.
\newblock Selfish mining in ethereum.
\newblock In {\em 2nd International Conference on Mathematical Research for Blockchain Economy, {MARBLE} 2020, online, August 24, 2020}, Springer Proceedings in Business and Economics, pages 65--90. Springer, 2020.

\bibitem[GS19]{garcia_deconstructing_2019}
{\'{A}}lvaro Garc{\'{\i}}a{-}P{\'{e}}rez and Maria~Anna Schett.
\newblock Deconstructing stellar consensus.
\newblock In {\em 23rd International Conference on Principles of Distributed Systems, {OPODIS} 2019, December 17-19, 2019, Neuch{\^{a}}tel, Switzerland}, volume 153 of {\em LIPIcs}, pages 5:1--5:16. Schloss Dagstuhl - Leibniz-Zentrum f{\"{u}}r Informatik, 2019.

\bibitem[HMR12]{hoang_enciphering_2012}
Viet~Tung Hoang, Ben Morris, and Phillip Rogaway.
\newblock An enciphering scheme based on a card shuffle.
\newblock In {\em Advances in Cryptology -- CRYPTO 2012}. Springer Berlin Heidelberg, 2012.

\bibitem[HV16]{halpern_rational_2016}
Joseph~Y. Halpern and Xavier Vila{\c{c}}a.
\newblock Rational consensus: Extended abstract.
\newblock In {\em Proceedings of the 2016 {ACM} Symposium on Principles of Distributed Computing, {PODC} 2016, Chicago, IL, USA, July 25-28, 2016}, pages 137--146. {ACM}, 2016.

\bibitem[LPR20]{lewispye_resource_2020}
Andrew Lewis-Pye and Tim Roughgarden.
\newblock Resource pools and the cap theorem, 2020.

\bibitem[LSP82]{lamport_byzantine_1982}
Leslie Lamport, Robert~E. Shostak, and Marshall~C. Pease.
\newblock The byzantine generals problem.
\newblock {\em {ACM} Trans. Program. Lang. Syst.}, 4(3):382--401, 1982.

\bibitem[Nak08]{nakamoto_peer_2008}
Satoshi Nakamoto.
\newblock Bitcoin: A peer-to-peer electronic cash system, 2008.

\bibitem[Nak19a]{nakamura_analysis_2019}
Ryuya Nakamura.
\newblock Analysis of bouncing attack on ffg, 2019.

\bibitem[Nak19b]{nakamura_prevention_2019}
Ryuya Nakamura.
\newblock Prevention of bouncing attack on ffg, 2019.

\bibitem[NKMS16]{nayak_stubborn_2016}
Kartik Nayak, Srijan Kumar, Andrew Miller, and Elaine Shi.
\newblock Stubborn mining: Generalizing selfish mining and combining with an eclipse attack.
\newblock In {\em {IEEE} European Symposium on Security and Privacy, EuroS{\&}P 2016, Saarbr{\"{u}}cken, Germany, March 21-24, 2016}, pages 305--320. {IEEE}, 2016.

\bibitem[NMRP20]{neuder_defending_2020}
Michael Neuder, Daniel~J. Moroz, Rithvik Rao, and David~C. Parkes.
\newblock Defending against malicious reorgs in tezos proof-of-stake.
\newblock In {\em {AFT} '20: 2nd {ACM} Conference on Advances in Financial Technologies, New York, NY, USA, October 21-23, 2020}, pages 46--58, 2020.

\bibitem[NTT21]{neu_ebb_2021}
Joachim Neu, Ertem~Nusret Tas, and David Tse.
\newblock Ebb-and-flow protocols: {A} resolution of the availability-finality dilemma.
\newblock In {\em 42nd {IEEE} Symposium on Security and Privacy, {SP} 2021, San Francisco, CA, USA, 24-27 May 2021}, pages 446--465. {IEEE}, 2021.

\bibitem[NTT22]{neu_two_2022}
Joachim Neu, Ertem~Nusret Tas, and David Tse.
\newblock Two more attacks on proof-of-stake ghost/ethereum.
\newblock In {\em Proceedings of the 2022 {ACM} Workshop on Developments in Consensus, ConsensusDay 2022, Los Angeles, CA, USA, 7 November 2022}, pages 43--52. {ACM}, 2022.

\bibitem[PAT23]{pavloff_ethereum_2023}
Ulysse Pavloff, Yackolley Amoussou{-}Guenou, and Sara {Tucci{-}Piergiovanni}.
\newblock Ethereum proof-of-stake under scrutiny.
\newblock In {\em Proceedings of the 38th {ACM/SIGAPP} Symposium on Applied Computing, {SAC} 2023, Tallinn, Estonia, March 27-31, 2023}, pages 212--221. {ACM}, 2023.

\bibitem[PAT24a]{pavloff_byzantine_2024}
Ulysse Pavloff, Yackolley Amoussou{-}Guenou, and Sara {Tucci{-}Piergiovanni}.
\newblock Byzantine attacks exploiting penalties in ethereum pos.
\newblock In {\em 54th Annual {IEEE/IFIP} International Conference on Dependable Systems and Networks, {DSN} 2024, Brisbane, Australia, June 24-27, 2024}, pages 53--65. {IEEE}, 2024.

\bibitem[PAT24b]{pavloff_incentive_2024}
Ulysse Pavloff, Yackolley Amoussou{-}Guenou, and Sara Tucci{-}Piergiovanni.
\newblock Incentive compatibility of ethereum's pos consensus protocol.
\newblock In {\em 28th International Conference on Principles of Distributed Systems, {OPODIS} 2024, December 11-13, 2024, Lucca, Italy}, volume 287 of {\em LIPIcs}. Schloss Dagstuhl - Leibniz-Zentrum f{\"{u}}r Informatik, 2024.

\bibitem[Pry22]{prysm_code}
Prysm.
\newblock Code consensus client, 2022.

\bibitem[Qua11]{quantum_proof_2011}
QuantumMechanic.
\newblock Proof of stake instead of proof of work, 2011.

\bibitem[Req19]{pullRequest_bouncing_2022}
Specification~Pull Request.
\newblock Bouncing attack patch, 2019.

\bibitem[Rou20]{roughgarden_transaction_2020}
Tim Roughgarden.
\newblock Transaction fee mechanism design for the ethereum blockchain: An economic analysis of {EIP-1559}.
\newblock {\em CoRR}, abs/2012.00854, 2020.

\bibitem[Sal20]{saleh_blockchain_2020}
Fahad Saleh.
\newblock {Blockchain without Waste: Proof-of-Stake}.
\newblock {\em The Review of Financial Studies}, 34(3):1156--1190, 2020.

\bibitem[SNM{\etalchar{+}}22]{schwarz_three_2021}
Caspar Schwarz{-}Schilling, Joachim Neu, Barnab{\'{e}} Monnot, Aditya Asgaonkar, Ertem~Nusret Tas, and David Tse.
\newblock Three attacks on proof-of-stake ethereum.
\newblock In {\em Financial Cryptography and Data Security - 26th International Conference, {FC} 2022, Grenada, May 2-6, 2022, Revised Selected Papers}, volume 13411 of {\em Lecture Notes in Computer Science}, pages 560--576. Springer, 2022.

\bibitem[SSZ16]{sapirshtein_optimal_2016}
Ayelet Sapirshtein, Yonatan Sompolinsky, and Aviv Zohar.
\newblock Optimal selfish mining strategies in bitcoin.
\newblock In {\em Financial Cryptography and Data Security - 20th International Conference, {FC} 2016, Christ Church, Barbados, February 22-26, 2016, Revised Selected Papers}, volume 9603 of {\em Lecture Notes in Computer Science}, pages 515--532. Springer, 2016.

\bibitem[SZ15]{sompolinsky_secure_2015}
Yonatan Sompolinsky and Aviv Zohar.
\newblock Secure high-rate transaction processing in bitcoin.
\newblock In {\em Financial Cryptography and Data Security - 19th International Conference, {FC} 2015, San Juan, Puerto Rico, January 26-30, 2015, Revised Selected Papers}, 2015.

\bibitem[TE18]{tsabary_eyal_2018}
Itay Tsabary and Ittay Eyal.
\newblock The gap game.
\newblock In {\em Proceedings of the 2018 {ACM} {SIGSAC} Conference on Computer and Communications Security, {CCS} 2018, Toronto, ON, Canada, October 15-19, 2018}, pages 713--728. {ACM}, 2018.

\bibitem[Woo16]{wood_polkadot_2016}
Gavin Wood.
\newblock Polkadot: Vision for a heterogeneous multi-chain framework.
\newblock {\em White paper}, 21(2327):4662, 2016.

\bibitem[YMR{\etalchar{+}}19]{yin_hotstuff_2019}
Maofan Yin, Dahlia Malkhi, Michael~K. Reiter, Guy Golan{-}Gueta, and Ittai Abraham.
\newblock Hotstuff: {BFT} consensus with linearity and responsiveness.
\newblock In {\em Proceedings of the 2019 {ACM} Symposium on Principles of Distributed Computing, {PODC} 2019, Toronto, ON, Canada, July 29 - August 2, 2019}, pages 347--356. {ACM}, 2019.

\bibitem[ZET20]{bar_efficient_2020}
Roi~Bar Zur, Ittay Eyal, and Aviv Tamar.
\newblock Efficient {MDP} analysis for selfish-mining in blockchains.
\newblock In {\em {AFT} '20: 2nd {ACM} Conference on Advances in Financial Technologies, New York, NY, USA, October 21-23, 2020}, pages 113--131. {ACM}, 2020.

\bibitem[ZLD23]{zhang_attestation_2023}
Mingfei Zhang, Rujia Li, and Sisi Duan.
\newblock Max attestation matters: Making honest parties lose their incentives in ethereum pos.
\newblock {\em {IACR} Cryptol. ePrint Arch.}, page 1622, 2023.

\end{thebibliography}

\titleformat{\chapter}[display]{\bfseries\Large\color{Prune}}
{   
    \vspace{-1cm}
    \centering 
    \begin{tikzpicture}
    \node (A) at (0,0) {};
    \node (B) at (2.5,0) {Appendix \thechapter };
    \node (C) at (5,0) {};
    \draw (A) -- (B);
    \draw (B) -- (C);
    \end{tikzpicture}
}
{.1ex}
{
    \centering
    \vspace{.7em}
    \Large
}
[\vspace{1em}]
\titlespacing{\chapter}{0pc}{0ex}{0.5pc}

\appendix

\chapter{Mathematical Elaborations}\label{chap:Calculus}

\section{Discrete case inacivity score during a probabilistic bouncinc attack.}\label{appendix:sec:discInaSco}

This gives two Bernoulli laws where the probability to have $k$ moves to the left 
at time $n$ are respectively:

\begin{equation}
    P_X(k,n) = \binom{n}{k}p_0^k(1-p_0)^{n-k}
\end{equation}

\begin{equation}
    P_Y(k,n) = \binom{n}{k}(1-p_0)^kp_0^{n-k}
\end{equation}

We can then compute the convolution that will give us the probability law to be:

\begin{align}
    &P_{X+Y}(s,2n) = [P_X+P_Y](s) = \sum_{k=0}^n P_X(k)\ast P_Y(s-k) \\
    = &\sum_{k=0}^n \binom{n}{k}p^{k}(1-p_0)^{n-k} \binom{n}{s-k}(1-p_0)^{k-s}p_0^{n-s+k} \\
    = &\sum_{k=0}^n \binom{n}{k} \binom{n}{s-k}p_0^{n+2k-s}(1-p_0)^{n+s-2k}
\end{align}

We are interested in the inactivity score in an attempt to study the evolution of stake of honest validators. To determine the stake we have to give a continuous function for the probability of the inactivity score. 

\paragraph{Continuous case} There are several ways to approach the continuous case. We use the same technique as before using a convolution. Starting by saying  that a random walk follows a Gaussian's distribution when time is big using the Central limit theorem. 
Knowing that the expectation of the two laws $P_X$ and $P_Y$ are $(5p_0-4)t$ and $(1-5p_0)t$ and their standard deviation is both $25p_0(1-p_0)t$, we have:

\begin{equation}
    P_X(x,t) = \frac{1}{\sqrt{\pi 50p_0(1-p_0)t}} e^{-\frac{(x- (5p_0-4)t)^2}{50p_0(1-p_0)t}}
\end{equation}
\begin{equation}
    P_Y(x,t) = \frac{1}{\sqrt{\pi 50p_0(1-p_0)t}} e^{-\frac{(x- (1-5p_0)t)^2}{50p_0(1-p)t}}
\end{equation}
\begin{equation}
    P_{X+Y}(s,t) = \int  P_X(x,t)P_Y(s-x,t)dx
\end{equation}

Which gives:
\begin{equation}
    P_{X+Y}(x,t) = \frac{1}{\sqrt{\pi 100p(1-p_0)t}} e^{-\frac{(x-3t/2)^2}{100p_0(1-p_0)t}}
\end{equation}

\section{From Gaussian white noise to log-normal distribution}\label{appendix:sec:lognormal}

In order to be able to find the probability of $s$, we need to change referential to stop $I$ from drifting with time. 
To do so we start by noticing that with the change of variables $u=-2^{26}\ln|s|$ this implies $\frac{d u}{d t} = -\frac{2^{26}}{s}\frac{ds}{dt} = I$. We now simplify what we looking for by introducing $\Tilde{u}$ and $\Tilde{I}$ the functions resolving these equations:

\begin{equation}
    \left\{
    \begin{array}{ll}
        &I = \Tilde{I} + Vt \\
        &u = \Tilde{u}+\frac{1}{2}Vt^2-2^{26}\ln(s_0) \\
    \end{array}
\right.
\end{equation}
Where $s_0=32$, for the initial stake. We can write the probability of $\Tilde{I}$ as:
\begin{equation}
    \phi(\Tilde{I},t)=\frac{1}{\sqrt{4\pi Dt}}e^{-\frac{\Tilde{I}^2}{4Dt}}
\end{equation}
Looking at the derivative of $\Tilde{u}$ we get:
\begin{equation}
    \frac{d\Tilde{u}}{dt} = \frac{du}{dt} - Vt = I - Vt = \Tilde{I}.
\end{equation}
Hence we find $d\Tilde{u}/dt = \Tilde{I}$. 
$\Tilde{I}$ being a Brownian motion, $\Tilde{u}$ is called an integrated Brownian motion. It is a well-known problem and this leads to :

\begin{equation}
    P(\Tilde{u},t)=\frac{1}{\sqrt{\frac{4}{3}\pi Dt^3}}\exp\left(-\frac{\Tilde{u}^2}{\frac{4}{3}Dt^3}\right).
\end{equation}



Where :
\begin{equation}
    \Tilde{u}=u-\frac{Vt^2}{2}+2^{26}\ln(s_0)
\end{equation}

We have that $d\Tilde{u}=-2^{26}\frac{ds}{s}$,
then the only remaining step is using the fact $P(s)=P(\Tilde{u})|\frac{d\Tilde{u}}{ds}|$, hence:

\begin{equation}
P(s)=\frac{2^{26}}{s}P(\Tilde{u}=-2^{26}ln(s/s_0)-Vt^2). 
\end{equation}

Thus, the probability of finding a stake $s$ at time $t$ for an honest validator during the probabilistic bouncing attack is:
\begin{equation}
    P(s,t)=\frac{2^{26}}{s\sqrt{\frac{4}{3}\pi Dt^3}}\exp\left(-\frac{(2^{26}\ln(s/32)+Vt^2/2)^2}{\frac{4}{3}Dt^3}\right)
\end{equation}
With $D$ and $V$, the diffusion and the velocity. In our case $V=3/2$ and $D=25p_0(1-p_0)$. The stake follows a log normal distribution.

The density of log normal distribution is:
\begin{equation}
    f_X(x ; \mu, \sigma)=\frac{1}{x \sigma \sqrt{2 \pi}} \exp \left(-\frac{(\ln x-\mu)^2}{2 \sigma^2}\right).
\end{equation}

The cumulative distribution function of the log-normal distribution is the following:
\begin{equation}
    F_X(x ; \mu, \sigma)=\frac{1}{2}+\frac{1}{2} \operatorname{erf}\left[\frac{\ln (x)-\mu}{\sigma \sqrt{2}}\right].
\end{equation}

\end{document}